\documentclass[opre,nonblindrev]{informs3rad}
\OneAndAHalfSpacedXI
\usepackage{xfrac}
\usepackage{bbm}
\usepackage{tcolorbox}
\usepackage{booktabs}
\DeclareRobustCommand{\mybox}[2][gray!15]{
\begin{tcolorbox}[  
        left=0pt,
        right=0pt,
        top=0pt,
        bottom=0pt,
        colback=#1,
        colframe=#1,
        enlarge left by=0mm,
        boxsep=10pt,
        arc=2pt,outer arc=2pt,
        ]
        #2
\end{tcolorbox}
}
\usepackage[ruled,vlined,linesnumbered]{algorithm2e}
\SetNoFillComment
\definecolor{cornellred}{rgb}{0.7, 0.11, 0.11}
\definecolor{dgreen}{rgb}{0.0, 0.5, 0.0}
\definecolor{ballblue}{rgb}{0.13, 0.67, 0.8}
\definecolor{royalblue(web)}{rgb}{0.25, 0.41, 0.88}
\definecolor{bleudefrance}{rgb}{0.19, 0.55, 0.91}
\definecolor{royalazure}{rgb}{0.0, 0.22, 0.66}
\usepackage{hyperref}
\hypersetup{
	colorlinks = true,
	linkcolor=royalazure,
	citecolor=royalazure,
	linkbordercolor = {white}
}
\usepackage{cleveref}

\usepackage{csquotes}
\makeatletter
\renewenvironment*{displayquote}
  {\begingroup\setlength{\leftmargini}{1cm}\csq@getcargs{\csq@bdquote{}{}}}
  {\csq@edquote\endgroup}
\makeatother

\usepackage{thmtools}
\usepackage{thm-restate}

\usepackage{accents}
\usepackage{enumitem}
\usepackage{multirow}
\usepackage{tikz}
\usepackage{pgfplots}
\usetikzlibrary{patterns}
\usepgfplotslibrary{fillbetween}
\usetikzlibrary{intersections}

\usepackage{mathtools}
\usepackage{natbib}
 \bibpunct[, ]{(}{)}{,}{a}{}{,}%
 \def\bibfont{\small}%
 \def\bibsep{\smallskipamount}%

\setcitestyle{authoryear}

\usetikzlibrary{arrows, decorations.markings}

\tikzstyle{vecArrow} = [thick, decoration={markings,mark=at position
	1 with {\arrow[semithick]{open triangle 60}}},
	double distance=1.4pt, shorten >= 5.5pt,
	preaction = {decorate},
postaction = {draw,line width=1.4pt, white,shorten >= 4.5pt}]
\tikzstyle{innerWhite} = [semithick, white,line width=1.4pt, shorten >= 4.5pt]

\newenvironment{myprocedure}[1][htb]
{  
\begin{algorithm}[#1]%
}{\end{algorithm}}

\newcommand*{\R}{\mathbb{R}}

\TheoremsNumberedThrough 
\EquationsNumberedThrough

	\providecommand{\given}{}
	\DeclarePairedDelimiterX{\set}[1]\{\}{\renewcommand\given{\nonscript\:\delimsize\vert\nonscript\:\mathopen{}}#1}
	\let\Pr\relax
	\DeclarePairedDelimiterXPP{\Pr}[1]{\mathbb{P}}[]{}{\renewcommand\given{\nonscript\:\delimsize\vert\nonscript\:\mathopen{}}#1}
	\DeclarePairedDelimiterXPP{\Ex}[1]{\mathbb{E}}[]{}{\renewcommand\given{\nonscript\:\delimsize\vert\nonscript\:\mathopen{}}#1}

\newcolumntype{P}[1]{>{\centering\arraybackslash}c{#1}}

\newcommand{\durationpdf}{g}
\newcommand{\durationcdfi}{\bar{\durationdistribution}}

\newcommand{\bench}[1]{\textbf{OPT}[#1]}
\newcommand{\tpre}{\tau}

	\newcommand{\Z}{\mathbb{Z}}
	\newcommand{\inventory}{c}
	\newcommand{\mininventory}{{\inventory_{\min}}}

	\newcommand{\totaltime}{T}
	\newcommand{\reward}{r}
	\newcommand{\rewards}{\mathbf \reward}
	\newcommand{\choice}{\phi}
	\newcommand{\pen}{\Psi}
	\newcommand{\durationdistribution}{G}
	\newcommand{\durationdistributions}{{\mathbf{\durationdistribution}}}
	\newcommand{\assortment}{S}
	\newcommand{\assortmentspace}{\mathcal{S}}
	\newcommand{\duration}{d}
	\newcommand{\typedistribution}{F}

	\newcommand{\inventorydual}{\theta_{i, t}}
	\newcommand{\probdual}{\lambda_{t}}
\newcommand{\probdualtype}{\lambda_{t,\type_t}}
	\newcommand{\typesequence}{\{\type_t\}_{t = 1}^\totaltime}
	
	\newcommand{\typedistributionsequence}{\{\typedistribution_t\}_{t = 1}^\totaltime}

	\newcommand{\EAR}[2][]{\text{Expected-LP}\ifthenelse{\not\equal{}{#1}}{_{#1}}{}\!\left[{\def\givenn{\middle|}#2}\right]}
	\newcommand{\algo}{\mathcal{A}}

	\newcommand{\Rev}[2][]{\text{\bf Rev}\ifthenelse{\not\equal{}{#1}}{_{#1}}{}\!\left[{\def\givenn{\middle|}#2}\right]}

	\newcommand{\alloc}{y}
	
	\newcommand{\curalloctype}{\alloc_{\assortment, t, \type_t}}
	\newcommand{\curalloctypeopt}{\alloc^*_{\assortment, t, \type_t}}
	\newcommand{\curexantetype}{X_{\assortment, t, \type_t}}
	\newcommand{\curchoice}{\choice^{\type_t}(\assortment, i)}
	\newcommand{\curreward}{\reward_i^{\type_t}}
	\newcommand{\price}{P}
	\newcommand{\curprice}{\price_{i, t}^{\type_t}}

	\newcommand{\prechoice}{\choice^{\type_{\tpre}}(\assortment, i)}

	\newcommand{\noaccents}[1]{#1}
	
	\newcommand{\newagentvar}[3][\noaccents]{%
		\expandafter\newcommand\expandafter{\csname #2\endcsname}{#1{#3}}%
		\expandafter\newcommand\expandafter{\csname #2s\endcsname}{#1{\boldsymbol{#3}}}%
		\expandafter\newcommand\expandafter{\csname #2smi\endcsname}[1][i]{#1{\boldsymbol{#3}}_{-##1}}%
		\expandafter\newcommand\expandafter{\csname #2i\endcsname}[1][i]{#1{#3}\agind[##1]}%
		\expandafter\newcommand\expandafter{\csname #2ith\endcsname}[1][i]{#1{#3}_{(##1)}}%
	}
	\newagentvar{typespace}{{\cal Z}}
	\newagentvar{typesubspace}{S}
	
	\newagentvar{type}{z}
	\newagentvar{othertype}{s}
	\newagentvar{val}{v}
	\newagentvar{hval}{\bar \val}
	\newagentvar{hbudget}{\bar \wealth}
	\newagentvar{budget}{B}
	\newagentvar{lbudget}{\underaccent{\bar}{ \wealth}}
	\newagentvar{lowestval}{0}
	\newagentvar{cumval}{V}
	\newagentvar{cumprice}{P}
	\newagentvar{welcurve}{W}
	\newagentvar{revcurve}{R}
	\newagentvar{outcome}{w}
	\newagentvar{outcomespace}{{\cal W}}

	\newcommand{\sampleassortment}{\hat{\assortment}}
	
	\newcommand{\DP}{\mathcal{E}}

	\newcommand{\exponential}{\pen(x) = \frac{e^{1-x}-e}{1 - e}}

	\newcommand{\infiniteboundtext}{1 - \tfrac{1}{\sqrt{\mininventory + 3}}}

\newcommand*{\rom}[1]{\expandafter\romannumeral #1}
\newcommand{\Rom}[1]{\uppercase\expandafter{\romannumeral #1\relax}}
\newcommand{\valueDPsimple}{\mathcal{V}}

\newcommand{\indica}{\mathtt{I}}

\newcommand{\algValue}{\mathcal{Q}}
\newcommand{\loss}{\varepsilon^*(\mininventory)}
\newcommand{\pdfoversets}{\mathcal{F}_{\type_t,\hat{\assortment},\bar{\assortment}}}
\newcommand{\gammaopt}{\gamma^*(\mininventory)}

\newcommand{\inflateRatio}{\eta}

\usepackage{subfig}

\newcommand{\condition}{\,\mid\,}

\newcommand{\prob}[2][]{\text{\bf Pr}\ifthenelse{\not\equal{}{#1}}{_{#1}}{}\!\left[{\def\givenn{\middle|}#2}\right]}
\newcommand{\expect}[2][]{\text{\bf E}\ifthenelse{\not\equal{}{#1}}{_{#1}}{}\!\left[{\def\givenn{\middle|}#2}\right]}

\newcommand{\tparen}{\big}
\newcommand{\tprob}[2][]{\text{\bf Pr}\ifthenelse{\not\equal{}{#1}}{_{#1}}{}\tparen[{\def\given{\tparen|}#2}\tparen]}
\newcommand{\texpect}[2][]{\text{\bf E}\ifthenelse{\not\equal{}{#1}}{_{#1}}{}\tparen[{\def\given{\tparen|}#2}\tparen]}

\newcommand{\sprob}[2][]{\text{\bf Pr}\ifthenelse{\not\equal{}{#1}}{_{#1}}{}[#2]}
\newcommand{\sexpect}[2][]{\text{\bf E}\ifthenelse{\not\equal{}{#1}}{_{#1}}{}[#2]}

\newcommand{\indicator}[1]{{\mathbbm{1}\left\{ #1 \right\}}}

\MANUSCRIPTNO{OPRE-2020-10-687}
\newcommand{\revcolor}[1]{{\color{black}#1}}

\IncMargin{-\parindent}
\newif\ifalgtwoappendix
\algtwoappendixfalse
\newif\ifsubassortmentproof
\subassortmentprooftrue
\newif\ifinfiniteapx
\infiniteapxtrue

\begin{document}

\TITLE{Near-Optimal Bayesian Online Assortment of Reusable Resources}

\RUNAUTHOR{Feng, Niazadeh, Saberi}

\RUNTITLE{Near-Optimal Bayesian Online Assortment of Reusable Resources}

\ARTICLEAUTHORS{%
\AUTHOR{Yiding Feng}
\AFF{Industrial Engineering and Decision Analytics, Hong Kong University of Science \& Technology (HKUST), \EMAIL{ydfeng@ust.hk}}
\AUTHOR{Rad Niazadeh}
\AFF{University of Chicago Booth School of Business, Chicago, IL, \EMAIL{rad.niazadeh@chicagobooth.edu}}
\AUTHOR{Amin Saberi}
\AFF{Management Science and Engineering, Stanford University, Stanford, CA, \EMAIL{saberi@stanford.edu}}
}

\ABSTRACT{

Motivated by the applications of rental services in e-commerce, we consider revenue maximization in online assortment of reusable resources for a stream of arriving consumers with different types. We design competitive online algorithms with respect to the optimum online policy in the Bayesian setting, in which types are drawn independently from known heterogeneous distributions over time. In the regime where the minimum of initial inventories $\mininventory$ is large, our main result is a near-optimal $1-\min\left(\frac{1}{2},\sqrt{\log(\mininventory)/\mininventory}\right)$ competitive algorithm for the general case of reusable resources. 
Our algorithm relies on an expected LP benchmark for the problem, solves this LP, and simulates the solution through an independent randomized rounding. 
The main challenge is obtaining point-wise inventory feasibility in a computationally efficient fashion from these simulation-based algorithms. To this end, we use several technical ingredients to design \emph{discarding policies} -- one for each resource. These policies handle the trade-off between the inventory feasibility under reusability and the revenue loss of each of the resources. However, discarding a unit of a resource changes the future consumption of other resources. To handle this new challenge, we also introduce \emph{post-processing} assortment procedures that help with designing and analyzing our discarding policies as they run in parallel, which might be of independent interest. {As a side result, by leveraging techniques from the literature on prophet inequality, we further show an improved near-optimal $1-1/\sqrt{\mininventory+3}$ competitive algorithm for the special case of non-reusable resources.} We finally evaluate the performance of our algorithms using the numerical simulations on {the} synthetic data.

}
\maketitle
\section{Introduction}
\label{sec:intro}
Assortment planning refers to the decision of a revenue-maximizing firm as to which subset of products to display to its consumers. In classic retail applications, the focus is mostly on sale; however, with the advent of online e-commerce platforms, several applications have emerged where the focus is on renting out \emph{reusable resources}. A reusable resource -- also referred to as a rental product -- leaves the stock for some time duration after being assigned to a consumer, and can be reassigned to a new consumer once it is back. Examples are virtual machines in cloud computing platforms such as AWS, houses in vacation rental online marketplaces such as Airbnb, and local professional services in online labor platforms such as Thumbtack. In order to extract more revenue, a (personalized) assortment policy can decide to display a different subset of these resources once a new consumer interacts with the platform; the task of such a policy is to manage the sequence of assortments in long-run given the inventory restrictions.

Motivated by the above applications, we study the online assortment of reusable resources, in which a platform sequentially makes irrevocable assortment decisions for a stream of arriving consumers. Each consumer has a \emph{type} that determines her choice probabilities given each possible assortment -- referred to as the consumer's choice model. The platform collects a one-time personalized payment each time a consumer rents a product. Different products have different rental fees that are determined by the consumer's type. To model the uncertainty about the duration of a rental, we consider a stochastic model where the rental durations are drawn independently each time a product is rented. The type of the arriving consumer also determines her rental duration distributions for different products. To model the platform's prior information about its future consumers -- which is usually formed based on the past consumers' data in an online platform --  we take a Bayesian approach. We assume the types are drawn independently over time from \emph{known heterogeneous distributions}.

Once a new consumer arrives, the platform observes her realized type from the known type distribution, displays a new subset of products, and allows the consumer to select a product stochastically from this subset based on her choice model. The goal is to design an online assortment algorithm in order to maximize the expected total collected rental fees (also known as the revenue) during the decision making horizon. Importantly, we consider the setting where each product has an initial inventory and the algorithm should always assort an available subset of products. Availabilities are determined by the current inventory levels of the products -- quantities that decrease as a unit of the product is rented and increase as it is returned to the stock.

 Given the sequence of distributions for the entire decision making horizon, a \emph{revenue benchmark} in our problem is an upper-bound on the expected revenue of any feasible online algorithm as the platform, where the expectation is over the randomness in the algorithm, the type sequence, and the consumers' choices. We measure the performance of our online algorithms using the notion of \emph{competitive ratio}, i.e., the worst-case ratio between the expected revenue of the online algorithm and the targeted revenue benchmark.

  As is common in the literature and practice of assortment optimization, we assume general consumer choice models that are weak substitute, i.e., assorting a new product only weakly decreases the choice probability of another assorted product (see \Cref{asp:substitutability}). With no further assumptions, even the one-shot assortment optimization can be computationally hard. To resolve this issue, we assume having oracle access to a blackbox algorithm that can solve the one-shot assortment optimization (see \Cref{asp:oracle}). Again, this is a common assumption in the literature to get around the computational issue associated with the one-shot problem when general consumer choice models are considered~\citep{GNR-14,RST-17}. In the oracle-access computational model, we aim to design polynomial-time competitive online algorithms with respect to an appropriate revenue benchmark for our problem. 

In this setting, the question of finding the optimum online policy that maximizes the expected total revenue is purely computational, and an exponential-size dynamic programming (DP) can formulate the optimum online. While no formal computational hardness is known for computing the optimum online in our problem, it is conjectured to be computationally hard, even in the oracle-access model.\footnote{For example, see \citet{PT-1987} for PSPACE-hardness of finding the optimum policy in partially observable Markov decision processes, and related discussions in ~\citet{RST-17,ANSS-19}.} Therefore, it is natural to consider the expected revenue of the optimum online policy as a revenue benchmark, and study whether it is amenable to polynomial-time competitive online algorithms. More specifically, we ask the following question in this paper: 
\vspace{0.3cm}
\begin{displayquote}
\emph{How close a polynomial time online algorithm can be to the optimum online policy in terms of expected revenue? In particular, can we obtain constant or near-optimal competitive online algorithms with respect to the optimum online policy (when initial inventories are large)?}
\end{displayquote}
\vspace{0.3cm}

 A significant progress towards providing a compelling answer to the above question is the result of \citet{RST-17}, which establishes an approximate dynamic programming approach for the exponential-size optimum online DP. They show that a greedy algorithm that uses a linear approximation of the optimal revenue-to-go function obtains at least $\tfrac{1}{2}$ of the expected revenue of the optimum online. They further study the setting when rental durations are infinite (i.e., the resources are not reusable) and rental fees are type-independent. In this setting, they show how to perform a rollout on a simple static policy to obtain $1-\min\left(\tfrac{1}{2}, \tfrac{1}{\sqrt[3]{\mininventory}}\right)$ fraction of the expected revenue of the optimum online policy, where $\mininventory$ is the minimum initial inventory across different products.\footnote{If rental fees are type-dependent, \citet{RST-17} show the same policy obtains $1-\min\left(\tfrac{1}{2}, \tfrac{R}{\sqrt[3]{\mininventory}}\right)$ competitive ratio guarantee, where $R$ is the ratio between maximum and minimum rentals fees across different types.} This last result is essentially a near-optimal competitive algorithm when the initial inventories are large and the resources are non-reusable. 
 {Note that the large initial inventory regime is 
relevant in many applications of assortment optimization 
 and it is the main focus of our paper as well. }
 
\subsection{Our Contributions} The main contribution of our paper is the following result.

\vspace{0.1cm}
{
\begin{displayquote}
\emph{\textbf{(Main Result)} For the general case of Bayesian online assortment of reusable resources, we propose a polynomial-time online algorithm that obtains a near-optimal competitive ratio of $1-\min\left(\frac{1}{2},\loss\right)$, where $\varepsilon^*(x)=O\left(\sqrt{\log(x)/x}\right)$.}
\end{displayquote}}
\vspace{0.1cm}
\noindent The above competitive ratio guarantee holds even when the rental fees, consumer choices, and rental duration distributions are type-dependent and vary arbitrarily across different types.

To obtain the above result, our work diverges from \citet{RST-17} by considering a different revenue benchmark. In particular, we consider a linear programming relaxation of the optimum online policy, refereed to as the \emph{Bayesian expected LP}. To define this benchmark, suppose a feasible online policy knows the exact realizations of future types, but does not know the realizations of consumer choices and rental durations. The optimum such policy, known as the \emph{clairvoyant optimum online}, clearly provides a revenue benchmark. Now consider a relaxation of this policy by only requiring the inventory feasibility constraints of reusable resources to hold in expectation over the randomness in types, consumer choices, and rental durations. Given the sequence of type distributions, this relaxation is encoded by an LP with exponential number of variables and polynomial number of time-varying packing constraints to ensure the inventory feasibility of reusable resources in-expectation. See \Cref{sec:stochastic-lp-solving} for details.

It turns out that we can simply solve the Bayesian expected LP in polynomial-time by solving its dual program using the ellipsoid method -- given access to the offline assortment oracle. {Given the LP solution, a simple but powerful technique in the Bayesian online optimization is to use a \emph{simulation-based} rounding algorithm to mimic the optimal solution of this LP \citep[for examples of this approach, refer to][]{AHL-12,DSA-12,MSZ-18,gallego2016online,wang2018online,DSSX-18,BM-19}.} After observing the type of the arriving consumer, this algorithm independently samples an assortment from a distribution over subsets of products that comes from the LP solution, ignoring the inventory constraints. This algorithm has no loss in terms of the expected revenue compared to the LP solution; however, it only respects the inventory constraints of each product in expectation -- and not necessarily under every sample path of the existing randomness. 

Our main technical contribution is providing techniques to transform the simulation-based algorithm into a point-wise feasible online algorithm in polynomial-time, with constant or negligible multiplicative loss in the expected revenue. To this end, we run a separate procedure over time -- one for each product --  together with the simulation-based policy. After an assortment is sampled at each time, each procedure decides whether or not to \emph{discard} the corresponding product if it is in the sampled subset to maintain the inventory feasibility of this product. 
{This specific architecture of sampling according to the LP solution and then using product-specific discarding rules have been explored in the past, e.g., see \cite{AHL-12,gallego2016online,wang2018online,DSSX-18,BM-19}.\footnote{{See specifically the ``Primal Routing Algorithm'' in Section 7 of \cite{gallego2016online} and the ``Separation
Algorithm'' in Section 4.2 of \cite{wang2018online}.}}
Similarly, we aim to design polynomial-time online discarding policies (and other necessary algorithmic constructs) that handle the trade-off between maintaining the inventory feasibility and the discarding revenue loss for each product. The main new challenges specific to our problem is that (i) products are reusable, (ii) products are weak-substitute --- and hence discarding a product increases the choice probability of other products, and (iii) discarding can be potentially randomized --- which combined with reusability creates complicated correlation structures among selection indicator random variables across time and makes it challenging to argue about point-wise feasibility of policies. In what follows, we sketch our main technical contributions and how they overcome these challenges.} 

\vspace{1mm}
\paragraph{{(\Rom{1})~General rental duration distributions/ near-optimal discarding (\Cref{sec:large inventory}):}} The main idea behind our near-optimal discarding procedure is discarding each available sampled product independently at random with a small probability. This is a simple yet reasonable approach 
\citep[cf.][]{HKS-07} 
as the simulation-based algorithm respects the inventory constraints of each product in expectation. By independent randomized discarding with probability $\gamma>0$, we leave some slack in the inventory feasibility constraint by ensuring that the expected value of the number of units of the product under rental is at most $(1-\gamma)$ times its initial inventory amount at any time. If this quantity as a sum of independent rental indicator random variables concentrates around its expectation, we will then avoid violation of the inventory constraint with high probability when $\gamma=O\left(\sqrt{\log(\mininventory)/{\mininventory}}\right)$. Moreover, it only loses $\gamma$ fraction of the expected revenue form this product, as desired.

However, the above simple approach does not work as described, because (i) the resources are reusable and the inventories are limited, hence the rental indicator random variable of a product at some time $\tau<t$ can be correlated with the rental indicator random variable of the same product at time $t$ if the rental duration of time $\tau$ is at least $t-\tau$ and the last unit of the product is rented at $\tau$; (ii) once a discarding procedure drops a product from the sampled assortment, there will be less cannibalization of other products -- as the consumer choices are weak substitute. This in turn increases the probability of other products being chosen by the arriving consumer, and hence increases the expected number of units of different products under rental in future for the resulting algorithm -- compared to what is expected from the simulation-based algorithm.

We fix the above issues by proposing a post-processing step after the independent randomized discarding, which we refer to as \emph{sub-assortment sampling}. In a nutshell, the goal of the sub-assortment sampling is to find a distribution over available subsets, so that the products which are not discarded will be rented with \emph{exactly} the same probability as in the optimal solution of the Bayesian expected LP. It is not even clear a priori whether such a distribution exists; nevertheless, we show it does and provide a polynomial-time construction to sample from this distribution. Using the properties of the sub-assortment sampling, we propose a coupling trick to show our desired concentration despite the fact that the rental indicator random variables are correlated across time. Note that the task of sub-assortment sampling is quite general, and it might be of independent interests in other applications.\footnote{{We would like to highlight that after appearance of an online version of our paper, through a personal communication with authors of 
\citet{GGU-20}, we were informed that this paper (which was not available online at the time) independently and concurrently discovered a procedure similar to our 
{sub-assortment sampling}
for settings with adversarial arrival 
and reusable resources.}}

\vspace{1mm}
\paragraph{{(\Rom{2})~General rental duration distributions/ $\tfrac{1}{2}$-competitive discarding (\Cref{sec:small inventory}):}}
As an alternative discarding policy for the general case of reusable resources, consider an exponential-size DP that keeps track of the state of each unit of the product (i.e., when each unit returns to the inventory) and solves the discarding task optimally. We introduce an approximate version of this DP{, which we also refer to as \emph{optimistic DP},} that can be solved in polynomial-time. The goal is to maximize the per-unit revenue-to-go of the product when the inventory is automatically \emph{replenished} by an exogenous process every time we make a discarding decision, so that the entire inventory of the product is always available on-hand. This new DP is inventory independent and thus is polynomial-size. We then consider a discarding algorithm that makes the same decisions as the optimistic DP. The result is a non-adaptive thresholding discarding rule, that is, an available product is only discarded if its rental fee is below a certain threshold. These thresholds are computed upfront and only depend on the product, time, and the realized type.

The main intuition behind why the above discarding algorithm is a reasonable approximation is as follows. We can show the expected revenue-to-go of this algorithm is a concave function of the inventory level -- i.e., the higher the inventory level, the lower the per-unit expected revenue-to-go. As a result, when there is no replenishment in reality, it obtains at least the same expected per-unit revenue-to-go as the DP with replenishment. We then analyze the worst-case ratio between the value of this inventory independent DP and the per-unit revenue of the expected LP using a ``factor revealing linear program'' and its dual. This approach establishes a lower-bound of $\tfrac{1}{2}$ for this ratio. {It is worth noting that similar proof techniques based on dual-fitting have been used in the literature for other problems with non-reusable resources~\citep[e.g., see][]{DA-09,A-07,AHL-12,gallego2016online,wang2018online}. Our work extends the existing analysis to prove performance guarantees for LP-based discarding policies when resources are reusable.}

{We highlight that our DP for designing the above approximate discarding policy given the expected LP solution shares similar but not completely identical recursive structures with the approximate dynamic programming approach in \cite{RST-17} for directly approximating the optimum online policy; in fact, in contrast to their approach, the Bellman update equation of our DP uses the the solution of the expected LP (the optimal assortment sampling probabilities). Therefore, while both DPs provide the same approximation factor of $\frac{1}{2}$, ours is with respect to the \emph{stronger} benchmark of expected LP, versus theirs that is with respect to the optimum online policy. This subtle difference turns out to be the key in combining the performance guarantees of our two algorithms in \Cref{sec:large inventory} and \Cref{sec:small inventory}, in order to obtain a simple hybrid algorithm that achieves the theoretical \emph{``best of both worlds''} competitive ratio guarantee with respect to the stronger expected LP benchmark (and as we observe later, improved performance in numerical simulations).}

\paragraph{(\Rom{3})~Hybrid simulation-based algorithm (\Cref{sec:hybrid}):}{Having access to the above simulation-based algorithms with different discarding rules, we aim to define a \emph{hybrid algorithm} that enjoys the competitive ratios of both small and large inventory regimes. To this end, we make an upfront decision on which discarding policy to use for each product.
In particular, we use the value function of the optimistic DP for each product separately to calculate the ratio $\mathcal{R}_i$ between the expected revenue-to-go of following the optimistic DP for this product and the contribution of this product to the expected-LP's objective. We then compare this ratio with $1-\varepsilon^*(c_i)$ (see \eqref{eq:loss} for definition of function $\varepsilon^*(\cdot)$) to partition the products into \emph{large inventory} (i.e., when $\mathcal{R}_i+\varepsilon^*(c_i)<1$) and \emph{small inventory} (i.e., when $\mathcal{R}_i+\varepsilon^*(c_i)>1$). For each large inventory product, we run the randomized discarding policy in \Cref{sec:large inventory}, and for each small inventory product we run the optimistic DP discarding policy in \Cref{sec:small inventory}. We also use the sub-assortment sampling procedure for post-processing to correct the resulting increase in choice probabilities of non-discarded products due to weak-substitution. By using the facts that (i) the competitive ratio analyses of both algorithms decouple across products, (ii) both analyses compare the expected revenue obtained from each product with the contribution of that product to the expected-LP's objective, and that (iii) sub-assortment sampling corrects the choice probabilities of non-discarded products, we show the resulting hybrid policy combines the two competitive ratios.\footnote{{It is worth noting that the $\frac{1}{2}$-competitive approximate DP algorithm of \cite{RST-17} cannot be combined using this approach with our near-optimal discarding policy, as this approximate DP algorithm competes with the optimum online policy and not the expected-LP.}} While this hybrid algorithm attains the best of both worlds competitive ratio of $1-\min\left(\frac{1}{2},\loss\right)$, it is also likely to outperform both policies in practical scenarios; the results of our numerical simulations in \Cref{sec:numerical} empirically support this claim. We also present a second hybrid algorithm using the method of conditional expectations. See \Cref{sec:hybrid} for more details.}

{We also complement our results by considering the special case of non-reusable resources. By leveraging techniques from the literature on prophet inequality and extending them to the Bayesian assortment optimization problem, we provide a near-optimal improved competitive ratio of $\left(1-\tfrac{1}{\sqrt{\mininventory+3}}\right)$ with respect to the expected LP in this setting. See \Cref{sec:discussion} and \Cref{apx:sand barrier} for more details.}


\paragraph{{(\rom{3})~Numerical simulations (\Cref{sec:numerical}):}}
 We finally provide numerical justification 
for the revenue performance of our proposed policies.
Adapting the setups of the numerical experiments in \citet{GNR-14,RST-17} to our setting,
we compare the revenue of our proposed policies --- i.e., the hybrid algorithm and the simulation with optimal discarding under infinite rental durations --- with other
policies in the literature. {In our numerical simulations, we consider various scenarios with both general rental duration distributions
and infinite rental durations.
In all of these scenarios, our policies 
noticeably outperform the other policies in terms of the expected revenue.}


\section{Preliminaries}
We first formalize our problem, the model, and all the required assumptions in \Cref{sec:model}. We then briefly explain various aspects of our expected LP benchmark in \Cref{sec:stochastic-lp-solving}.

\label{sec:prelim}

\subsection{Model and Problem Definition}
\label{sec:model}
The platform offers $n$ different rental products, indexed by $[n] = \{1, 2, \dots, n\}$. Each rental product $i$ has an initial inventory of $\inventory_i \in \Z_+$.  Consumers who are interested to rent these products arrive sequentially at times $t=1,2,\ldots,T$. Consumer $t$ has type $\type_t\in\typespace_t$, where $\typespace_t$ denotes the (discrete) space of possible types at time $t$. We assume types are drawn independently from \emph{known} probability distributions $\typedistribution_t:\typespace_t\rightarrow [0,1]$ at times $t=1,\ldots,T$. 

Upon the arrival of consumer $t$,  her type $\type_t$ is revealed to the platform. Given this type and the history up to time $t$, the platform offers an assortment of available products $\assortment_t \in \assortmentspace$ from its inventory, where $\assortmentspace\subseteq 2^{[n]}$ is the collection of all feasible assortments that can be offered ignoring the inventory availability.  Given the assortment $\assortment_t$, the consumer chooses a rental product $i_t\in\assortment_t$, pays a rental fee to the platform, and keeps the product for a stochastic rental duration $d_t\in\mathbb{Z}_+$.

We consider the setting where the consumer choice behavior, rental fees of different products, and rental duration distributions of different products depend on the type $\type_t$ at each time $t$. Formally, a consumer type $\type$ is defined as a tuple $\langle\choice^\type, \mathbf{r}^\type,\durationdistributions^\type \rangle$, so that:
\begin{itemize}
\item The choice of a consumer with type $\type$ is modeled by a general choice model function $\choice^\type: \assortmentspace
\times [n] \rightarrow [0,1]$, where $\choice^\type(\assortment, i)$ is the probability that consumer with type $\type$ chooses product $i$ to rent when assortment set $S\in\assortmentspace$ is offered.  
\item For a consumer with type $\type$, $\rewards^\type = (\reward_1^\type, \reward_2^\type, \dots, \reward_n^\type) 
\in \R^n$, where $\reward_i^\type$ denotes the rental fee of product $i$. Moreover, $\durationdistributions^\type =
(\durationdistribution_1^\type, \durationdistribution_2^\type,
\dots, \durationdistribution_n^\type)$, where $\durationdistribution_i^\type$ denotes the c.d.f. of rental duration of product $i$ for type $\type$. We use $\durationpdf_i^{\type}:[T]\rightarrow [0,1]$ to denote the p.d.f. of rental duration of product $i$ for type $\type$. Moreover, let  $\durationcdfi_i^{\type}(\cdot)\triangleq 1-\durationdistribution^\type_i(\cdot)$. 
\end{itemize}
\noindent Note that we assume rental durations are independent across time, that is, if at time $t$ a consumer of type $\type$ chooses a product $i$, a fresh sample $d_t\sim \durationdistribution_i^\type$ is realized as the rental duration of this product. We further impose the following assumptions on our choice models and feasible assortments, which are common in previous literature  \citep[cf.][]{GNR-14, RSTT-14}

\begin{assumption}[Weak substitutability]
	\label{asp:substitutability}
	For all $t\in[\totaltime], \type \in \typespace_t$ and $i \in [n]$,
	$\choice^\type (\emptyset, i) = 0$. Moreover, for all 
	$\assortment\in \assortmentspace$ and $j\in [n]/ \{i\}$, $\choice^\type (\assortment, i)\geq \choice^\type (\assortment\cup\{j\}, i).$
\end{assumption}
\begin{assumption}[{Downward-closed feasibility}] 
\label{asp:down-closed}
If $S\in\assortmentspace$ and $S'\subseteq S$ then $S'\in\assortmentspace$, i.e., a feasible assortment will remain feasible after removing any subset of its offered products. 
\end{assumption}

\begin{remark}
In online hospitality services such as Airbnb, users report the duration of their stay to the platform before platform shows them a listing. In such a variation, the platform makes the assortment decision by using the exact realizations of current rental times for the arriving type. Indeed, this is a special case of our model where the rental time distributions are point mass.
\end{remark}

Given type distributions $\typedistributionsequence$, the goal is to design online algorithms -- playing the role of the platform -- that maximize the expected revenue granted from rental fees; here, the expectation is over randomness of the algorithm (if randomized) and the environment, i.e., types, consumer choices, and rental durations. A revenue benchmark for this problem is defined to be any upper-bound on the expected revenue obtained by any feasible online algorithm (which might or might not be achievable by a feasible online algorithm). Fixing a revenue benchmark, we evaluate the performance of any online algorithm by its \emph{competitive ratio} against this benchmark. Informally speaking, competitive ratio is the worst-case ratio between the expected total revenue of the online algorithm and 
the benchmark, where the worst-case is over all possible type distributions.

\begin{definition}[Competitive Ratio]
	\label{def:competitive ratio}
	An online algorithm $\algo$ is \emph{$\alpha$-competitive} against a given 
	revenue
	benchmark if
	\begin{align*}
		\inf\limits_{T \geq 1}
		\inf\limits_{\typedistributionsequence}
		\frac{\Rev[\algo]{\typedistributionsequence}}{\bench{\typedistributionsequence}}
		\geq \alpha~,
	\end{align*}
	where $\Rev[\algo]{\cdot}$ is the expected revenue of algorithm $\algo$ and $\bench{\cdot}$ is the given revenue benchmark.
\end{definition}

For a general consumer choice model, the exact or even approximate offline assortment optimization can be computationally hard~\citep{KFV-08}. In order to avoid this obstacle when designing polynomial time online algorithms for general consumer choice models, we assume having access to an algorithm that solves the offline assortment problem. For simplicity, we assume the solver is exact throughout the paper, but all of our results still hold with a multiplicative degrade of $\beta$ in the competitive ratios if the solver is a $\beta$-approximation algorithm for some $0<\beta<1$. 
\begin{assumption}[{Offline oracle}]
	\label{asp:oracle}
	For all $t\in[\totaltime], \type \in \typespace_t$, and $\mathbf{\hat{R}}\in\mathbb{R}_+^n$, we have oracle access to an algorithm that finds a subset $\hat{\assortment}\in\assortmentspace$ such that:
\begin{equation*}
\hat{\assortment}\in\underset{\assortment\in\assortmentspace}\argmax~{\displaystyle\sum\nolimits_{i=1}^n\hat{R}_i\choice^\type(\assortment,i)}
\end{equation*}

\end{assumption}

\subsection{Bayesian Expected LP Benchmark}
\label{sec:stochastic-lp-solving}
A key ingredient in all of our algorithms is the \emph{Bayesian expected LP benchmark} -- a concept commonly used in previous literature on online allocations, mechanism design, and assortment optimization to remedy issues of the above benchmarks \citep[e.g., see][]{chawla2010multi,ala-14,MSZ-18,gallego2016online,wang2018online,ANSS-19}. This benchmark, denoted by $\EAR{\typedistributionsequence}$, uses linear programming to capture the optimum algorithm that only requires to satisfy the inventory constraints in expectation, where expectation is taken over randomness in rental durations and consumer types given type distributions $\typedistributionsequence$:

\begin{align}
\label{eq:exante-stochastic}\tag{$\EAR{\typedistributionsequence}$}
	\begin{array}{llll}
	{\max
	\limits_{\mathbf \alloc \geq \mathbf 0}}~~&
	\displaystyle\sum\nolimits_{t=1}^\totaltime 
	\displaystyle\sum\nolimits_{\type_t\in\typespace_t}
	\displaystyle\sum\nolimits_{\assortment \in \assortmentspace}
	\displaystyle\sum\nolimits_{i=1}^n\typedistribution_t(\type_t)
	\curreward\curchoice\curalloctype
	&~~~~~\text{s.t.}& \\[1em]
	  &
	\displaystyle\sum\nolimits_{\tpre = 1}^{t}
	\displaystyle\sum\nolimits_{\type_\tpre\in\typespace_\tpre}
	\displaystyle\sum\nolimits_{\assortment \in \assortmentspace}
	\typedistribution_\tpre(\type_\tpre)\durationcdfi^{\type_{\tpre}}_i(t-\tpre)
	\prechoice\alloc_{\assortment, \tpre, \type_\tpre}\leq \inventory_i
	&~~~~~i \in [n],\ t \in [\totaltime]&
	\\[1em]
	&\displaystyle\sum\nolimits_{\assortment\in \assortmentspace}\curalloctype \leq 1
	&~~~~~t \in [\totaltime],\ \type_t\in\typespace_t &
\end{array}
\end{align}
Here, variables $\{\curalloctype\}_{t\in[\totaltime],\assortment\in\assortmentspace,\type_t\in\typespace_t}$ correspond to probabilities 
that assortment $\assortment$ is offered 
to consumer $t$ given type $\type_t$ is realized, and first constraint shows inventory feasibility in expectation.

 \revcolor{A few explanations are in order.} 
 First, the optimal objective value of this LP is an upper-bound on the expected revenue of the clairvoyant optimum online benchmark, and hence the weaker non-clairvoyant optimum online (\Cref{prop:relaxation}; see \Cref{apx:LP} for the proof). Second,
$\EAR{\typedistributionsequence}$
can be solved efficiently using an oracle for the offline assortment (\Cref{prop:EAR running time}; see \Cref{apx:LP} for the proof). We use this computational block as a pre-processing step in all of our algorithms.\footnote{{In fact, one needs to run the ellipsoid method for the dual of this LP using the offline assortment solver as the separation oracle in order to find the optimal solution. In practice, to obtain a faster algorithm, one can use \emph{cutting plane} methods such as \cite{vaidya1996new} or even faster almost-linear-time cutting plane methods such as \cite{lee2015faster} that use the separation oracle more efficiently.}} 

\begin{restatable}{proposition}{EARupperbound}
\label{prop:relaxation}
	For any type distributions 
	$\typedistributionsequence$,
	the expected total revenue of the clairvoyant optimum online benchmark is upper-bounded by
	$\EAR{\typedistributionsequence}$. 
\end{restatable}

\begin{restatable}{proposition}{EARpolytime}
\label{prop:EAR running time}
Given an algorithm for offline assortment (\Cref{asp:oracle}), an optimal assignment $\{\curalloctypeopt\}$ of~$\EAR{\typedistributionsequence}$ can be computed efficiently in time $\textrm{Poly}(n,T, \sum_{t\in[\totaltime]}\lvert\typespace_t\rvert)$. Moreover, $\{\curalloctypeopt\}$ has no more than $\textrm{Poly}(n,T, \sum_{t\in[\totaltime]}\lvert\typespace_t\rvert)$ non-zero entries.  
\end{restatable}

\revcolor{In \Cref{apx:benchmark comparison}, we compare
the Bayesian expected LP benchmark with other benchmarks 
considered in the literature.}



\section{Near-optimal Algorithm for General Rental Durations}
\label{sec:stochastic}
In this section, we 
present our main result --
a near-optimal online simulation-based algorithm
with competitive ratio
at least $
\max\left(\tfrac{1}{2},1 - \loss\right)$ against Bayesian expected LP benchmark, where
\begin{equation}
\label{eq:loss}
    \varepsilon^*(x)\triangleq \underset{\gamma\in[0,1]}{\min}~1-\left(1-\gamma\right)\left(1 - \exp\left(-\frac{\gamma^2x}
    {2-\gamma}\right)\right)
\end{equation}
Let $\gammaopt$ be the optimal assignment of $\gamma$ in \Cref{eq:loss}. It is not hard to verify that
$\loss=O\left(\sqrt{\log(\mininventory)/\mininventory}\right)$ and is achieved at $\gammaopt=O\left(\sqrt{\log(\mininventory)/\mininventory}\right)$.
We first sketch our approach in \Cref{sec:sketch}. We then introduce a simulation-based algorithm
with competitive ratio $1 - \loss$
in \Cref{sec:large inventory},
and a different simulation-based algorithm to 
guarantee a competitive ratio of at least $\tfrac{1}{2}$
(even for small $\mininventory$)
in \Cref{sec:small inventory}. We finally present two simple hybrid algorithms that can obtain the best of two competitive ratios in \Cref{sec:hybrid}.
\subsection{High-level Sketch of Our Approach} 
\label{sec:sketch}
Let $\{\curalloctypeopt\}$ be the optimal assignment of $\EAR{\typedistributionsequence}$. 
As $\emptyset\in\assortmentspace$, without loss of generality we can only consider optimal assignments where:
\begin{equation*}
\displaystyle\sum\nolimits_{\assortment\in \assortmentspace}\curalloctypeopt= 1~~~~~~\forall t \in [\totaltime],\ \type_t\in\typespace_t
\end{equation*}
All of our simulation-based online algorithms in this paper follow four steps:
\mybox{

\setlength{\itemsep}{1pt}
  \setlength{\parskip}{1pt}
  \setlength{\parsep}{1pt}
  \openup 0.6em
 \noindent - At time $t=0$ (before starting):
 \vspace{1mm}
 \begin{enumerate}[label=(\roman*)]
     \item\textbf{Pre-processing}: Compute an optimal assignment $\{\curalloctypeopt\}$ of $\EAR{\typedistributionsequence}$ by invoking the offline oracle described in \Cref{asp:oracle}. Also, compute any other offline parameters that are occasionally needed by the algorithm.
 \end{enumerate}
\vspace{2mm}

\noindent - At each time $t=1,2,\ldots,T$:
\vspace{1mm}
\begin{enumerate}[label=(\roman*)]
\setlength{\itemsep}{1pt}
  \setlength{\parskip}{1pt}
  \setlength{\parsep}{1pt}
\setcounter{enumi}{1}
\item\textbf{Simulation}: Upon realizing consumer type $\type_t$ at time $t$, an outer procedure suggests $\hat{\assortment}\in\assortmentspace$ to be assorted by sampling $\hat{\assortment}$ from the distribution $\{\curalloctypeopt\}_{\assortment\in\assortmentspace}$ over $\assortmentspace$.
\item\textbf{Discarding}: For each product $i\in\hat{\assortment}$, a separate inner discarding procedure decides whether to remove this product from the final assortment, given the history up to time $t$ and realized type $\type_t$. If no units of product $i$ is available on-hand, it is discarded automatically to guarantee inventory feasibility. Otherwise, the inner procedure of product $i$ decides to discard or not. Let $\bar{\assortment}\subseteq \hat{\assortment}$ be the set of undiscarded products. 
\item \textbf{Post-processing}: Given $\type_t$, $\hat{\assortment}$ and $\bar{\assortment}$, pick a probability distribution $\pdfoversets$ over all subsets of $\bar{\assortment}$. Then sample
an assortment $\tilde\assortment\sim\pdfoversets$
and offer it to the consumer.
\end{enumerate}}
In the above four-step layout, step~(\rom{2}) is a loss-less randomized rounding for the optimal solution of $\EAR{\typedistributionsequence}$; however the resulting assortment only guarantees inventory feasibility of each product in expectation. 
The role of step~(\rom{3}) and step~(\rom{4}) is to identify a (randomized) subset of this feasible in-expectation assortment, to not only guarantee inventory feasibility in each sample path, but to also guarantee that the expected loss due to discarded products is small.



\subsection{
 Large Initial Inventory:
Towards
Competitive Ratio 
\texorpdfstring{$\boldsymbol{1 - \loss}$}{}}
\label{sec:large inventory}
 \newcommand\mycommfont[1]{\footnotesize\ttfamily\textcolor{royalazure}{#1}}
\SetCommentSty{mycommfont}

 The main idea behind the algorithm of this subsection is discarding each product independently at random with probability $\gamma= O(\sqrt{\log(\mininventory)/\mininventory})$ in step~(\rom{3}) at each time $t$. Intuitively speaking, this discarding tries to leave enough probability for not violating any of the inventory constraints at each time $t$. To see this, if discarding a product does not change the choice probability of another assorted product, the expected {number of unavailable units} of each product $i$ at each time $t$ is at most $(1-\gamma)\inventory_i$ -- due to the feasibility in-expectation of sampled sets in step~(\rom{2}). Now consider the rental indicator random variables of product $i$, i.e., random variables indicating whether this product is rented at each time or not. If these random variables are mutually independent across time, then we can use simple concentration bounds for the sum of independent random variables to prove our claim.
 
 There are two major issues with the above approach:
\begin{enumerate}[label=(\Roman*)]
     \item Under weak substitutability (\Cref{asp:substitutability}), discarding product $i$ weakly increases the choice probability of another assorted product $j\neq i$. Therefore, the probability of an available unit of product $j$ being rented at each time $\tau<t$ becomes larger than expected, which in turn increases the expected {number of unavailable units} of this product at time $t$ if we only simulate the expected LP's optimal solution and discard each product independently with probability $\gamma$. 
     \item As resources are reusable and inventories are limited, the rental indicator random variable of product $i$ at time $\tau<t$ is \revcolor{possibly positively} correlated with the rental indicator random variable of the same product at time $t$; in fact, first indicator forces the second indicator to be zero  when the realized rental duration $d_\tau$ at time $t$ is no smaller than $t-\tau$, last unit of the product is rented at time $\tau$, and no units of the product return during $[\tau+1,t]$.
     
\end{enumerate}
 We address the first issue in \Cref{sec:subassortment} by changing the algorithm, and address the second issue in \Cref{sec:analysis} by modifying the analysis. 

\subsubsection{Sub-assortment Sampling}
\label{sec:subassortment}
To fix the first issue, we propose the \emph{sub-assortment sampling} procedure -- a post-processing procedure to be used in step~(\rom{4}). This procedure ensures that products which were not discarded 
in step~(\rom{2}) are rented by the arriving consumer with \emph{exactly} the same
probability as in the optimal solution of the expected LP benchmark. More formally, the sub-assortment sampling induces a distribution $\pdfoversets$ over subsets of $\bar\assortment$ at each time $t$, so that 
\begin{equation}
\label{eq:subassortment}
\forall i\in\bar\assortment:~~~~~~~\expect[\tilde \assortment\sim \pdfoversets]{\phi^{\type_t}(\tilde\assortment,i)}
=\phi^{\type_t}(\hat\assortment, i)~.
\end{equation}
It is not clear a priori whether such a distribution $\pdfoversets$ exists, yet alone can be sampled from in polynomial time (polynomial in number of products $n$); nevertheless, for any general choice model satisfying weak substitutability (\Cref{asp:substitutability})
and downward-closed feasibility (\Cref{asp:down-closed}), we show such a distribution $\pdfoversets$ exists and we introduce Procedure~\ref{alg:sample assortment} that recursively samples a set from $\pdfoversets$ in polynomial time.\footnote{
It is noteworthy that \citet{GGU-20} independently and concurrently discovered an
idea similar to our 
sub-assortment sampling
for settings with adversarial arrival 
and reusable resources.}
\ifsubassortmentproof
\else
We defer the proof of \Cref{prop:assortment polytope}
to \Cref{apx:subassortment sampling}
for space reason.
\fi

\begin{myprocedure}
\caption{\textsc{Sub-assortment Sampling}}
\label{alg:sample assortment} 

\KwIn{choice model $\choice$, assortment $\assortment=\{1,2,\ldots,m\}$,	target probabilities $\{p_i\}_{i\in \assortment}$}

\vspace{2mm}
Let $\sigma:[m]\rightarrow [m]$ be a permutation
	such that $1\geq \frac{p_{\sigma(1)}}{\choice(\assortment,\sigma(1))} \geq
	\frac{p_{\sigma(2)}}{\choice(\assortment,\sigma(2))} 
	\geq \cdots \geq
	\frac{p_{\sigma(m)}}{\choice(\assortment,\sigma(m))}\geq 0$
	
	\tcc{Define 
	$\frac{p_{\sigma(j)}}{\choice(\assortment,\sigma(j))} = 1$
	if $\choice(\assortment,\sigma(j)) = 0$ or $S=\emptyset$.}
\vspace{2mm}
\If{$\frac{p_{\sigma(m)}}{\choice(\assortment,\sigma(m))}=1$}{
\vspace{2mm}
\Return $\tilde\assortment\gets\assortment$
}
\vspace{1mm}
\Else{
\vspace{2mm}
Let $q_0=1-\frac{p_{\sigma(1)}}{\choice(\assortment,\sigma(1))}$, $q_m=\frac{p_{\sigma(m)}}{\choice(\assortment,\sigma(m))}$, and $q_j=\frac{p_{\sigma(j)}}{\choice(\assortment,\sigma(j))}-\frac{p_{\sigma(j+1)}}{\choice(\assortment,\sigma(j+1))}$ for $j=1,\ldots,m-1$

\tcc{Note that $\sum_{j=0}^m q_j=1$}
\vspace{1mm}
Sample $j^* \sim \{q_j\}_{j=0}^{m}$

\If{$j^* = 0$}{
\Return $\emptyset$
}

\If{$j^* = m$}{
\Return $\assortment$
}

Let $\assortment' \gets \{\sigma^{-1}(j)\}_{j=1}^{j^*}$

\Return
 	$\tilde\assortment\gets\textsc{Sub-assortment Sampling}\left(\choice,\assortment',\{\choice(\assortment, i)\}_{i\in \assortment'}\right)$
}
\end{myprocedure}

\begin{restatable}{proposition}{polytope}
\label{prop:assortment polytope}
For any weak substitutable
and downward-closed feasible choice model $\phi$,
any assortment $\assortment\in \assortmentspace$,
and any target probabilities $\{p_i\}_{i\in \assortment}$
such that $p_i \leq \choice(\assortment, i)$ for all $i\in \assortment$,
Procedure~\ref{alg:sample assortment} 
outputs a randomized assortment $\tilde\assortment$ that satisfies
(i) $\tilde\assortment\subseteq\assortment$; (ii)
$\expect[\tilde\assortment]{\choice(\tilde\assortment, i)} = p_i$ 
for all $i\in \assortment$.
Moreover, it runs in time $\textrm{Poly}(n)$.
\end{restatable}

\begin{remark}
To guarantee \Cref{eq:subassortment} given any $(\type_t,\hat\assortment,\bar\assortment)$ at step~(\rom{4}), we invoke \Cref{prop:assortment polytope} by setting $\choice\gets\choice^{\type_t}$, $\assortment\gets \bar\assortment$, and $p_i\gets \choice^{\type_t}(\hat{S},i)$ for all $i\in\bar\assortment$. Note that $p_i=\choice^{\type_t}(\hat{S},i)\leq \choice^{\type_t}(\bar\assortment,i)$ for all $i\in\bar\assortment$, simply because of weak substitutability and the fact that $\bar\assortment\subseteq \hat\assortment$. 
\end{remark}
\ifsubassortmentproof
{\begin{proof}{\emph{Proof of \Cref{prop:assortment polytope}.}}
Without loss of generality, we assume $\sigma$ is the identity permutation, i.e., $\sigma(i)=i$ for $i\in[m]$. To show the polynomial
running time,
observe that 
(a) the running time in each recursion
is $\textrm{Poly}(n)$;
and 
(b) the number of iterations of this recursive algorithm is 
at most $n$, 
since $|\assortment| \leq n$ at the beginning and the size of the $\assortment'$ that is the input of the next recursive call shrinks by $1$ at each iteration, i.e.,  $|\assortment'| \leq |\assortment| - 1$.

Property (i) holds by construction.
We show property (ii) by induction on 
$m = |\assortment|$,
i.e., size of 
assortment $\assortment$. 
In this induction, 
we use another simple property (iii)
that
$\phi(\assortment, i)\left(\sum_{j=i}^m q_j\right)  = p_{i}$
for all $i\in \assortment$,
which immediately hold by construction.

\noindent\emph{Base Case ($m=1$).} 
In this case
Procedure~\ref{alg:sample assortment} randomly outputs
$\emptyset$ or $\assortment$.
By property (iii), the induction statement
holds.

\noindent\emph{Inductive step ($m > 1$).} 
Fix an arbitrary product $i\in \assortment$.
Notice that by construction
$
\expect[\tilde\assortment]{\choice(\tilde\assortment, i)
    \condition j^* < i} = 0$,
and 
$
\expect[\tilde\assortment]{\choice(\tilde\assortment, i)
    \condition j^* = m}
    =
    \choice(\assortment, i)$.
For any realized value $j^* = i, \dots, m - 1$,
and its corresponding $\assortment' = \{1,\ldots,j^*\}$, we can use the induction hypothesis for the assortment $\assortment'$
with probabilities $p'_i=\choice(S, i)$ for each $i\in\assortment'$. This is true simply because $\lvert\assortment'\rvert \leq m-1$ and that $\choice(S, i) \leq \choice(S', i)$ for each $i\in\assortment'$ as the choice model $\choice$ is weak substitute. By 
invoking the induction hypothesis when we use $\assortment'$ in the next recursive call,  we have 
    $\expect[\tilde\assortment]{\choice(\tilde\assortment, i)
    \condition j^* = j} = \choice(\assortment, i)$ for all $j = i, \dots, m - 1$.
    Thus, invoking property (iii),
\begin{align*}
\expect[\tilde\assortment]{\choice(\tilde\assortment, i)
    }
    =
    \sum_{j=0}^m
    q_j\expect[\tilde\assortment]{\choice(\tilde\assortment, i)
    \condition j^* = j}
    =
    \left(
    \sum_{j=i}^m q_j\right)\phi(\assortment,i)
    =
    p_i
\end{align*}
which completes the inductive step and finishes the proof.
\hfill\halmos
\end{proof}}
\fi

\subsubsection{The Algorithm and Analysis} 
\label{sec:analysis}
Now we present our first simulation-based algorithm (\Cref{alg:large inventory}) 
with its competitive ratio guarantee 
(\Cref{thm:competitive-ratio large inventory}).
\revcolor{We
defer the formal
proof of \Cref{thm:competitive-ratio large inventory}
to \Cref{apx:alg 1 proof}.}

\begin{algorithm}[htb]
 	\caption{Simulation-based Algorithm with Random Discarding}
 	\label{alg:large inventory}
 	\KwIn{discarding probability $\gamma\in[0,1]$}
 	\vspace{2mm}
 	\emph{\underline{Pre-processing:}} Compute the optimal assignment $\{\curalloctypeopt\}$ of $\EAR{\typedistributionsequence}$ by invoking the offline assortment oracle  (\Cref{asp:oracle})
 	\vspace{2mm}
 	
 	\For{$t=1$ to $T$}{
 	\tcc{consumer $t$ with type $\type_t\sim\typedistribution_t$ arrives} 
 	\vspace{1mm}
 	
 	\emph{\underline{Simulation:}} Upon realizing consumer type $\type_t$, sample $\hat{\assortment}_t\sim \{\curalloctypeopt\}_{\assortment\in\assortmentspace}$
 	
 	\vspace{2mm}
 	\emph{\underline{Discarding:}} Initialize $\bar\assortment_t\gets\hat\assortment_t$
 	
 	\For{each product $i\in\hat\assortment_t$}{
 	\vspace{1mm}
 	Flip an independent coin and remove $i$ from $\bar\assortment_t$ with probability $\gamma$
 	
 	\If {there is no available unit of product $i$}{
 	Remove $i$ from $\bar\assortment_t$}}
 	
 	\vspace{2mm}
 	\emph{\underline{Post-processing:}} Let $\tilde{S}_t\gets \textsc{Sub-assortment Sampling} \left(\choice^{\type_t}, \bar\assortment_t,\{\choice(\hat\assortment_t, i)\}_{i\in \bar\assortment_t}\right)$
 	
 	\tcc{Send a query call to Procedure~\ref{alg:sample assortment} with appropriate input arguments}
 	\vspace{2mm}
 	
 	Offer assortment $\tilde{\assortment}_t$ to consumer $t$
 	
 	}
\end{algorithm}

\begin{theorem}
	\label{thm:competitive-ratio large inventory}
	By setting $\gamma=\gammaopt$, the competitive ratio of \Cref{alg:large inventory} against the Bayesian expected LP benchmark $\EAR{\typedistributionsequence}$ is at least $1 -\loss=1-O\left(\sqrt{\log(\mininventory)/\mininventory}\right) $. 
	Moreover, it runs in time $\textrm{Poly}(n,T, \sum_{t\in[\totaltime]}\lvert\typespace_t\rvert)$ given oracle access to an offline algorithm for assortment optimization (\Cref{asp:oracle}).
\end{theorem}

\revcolor{
Before presenting the proof sketch of \Cref{thm:competitive-ratio large inventory},
we first discuss how to use a simple concentration argument
to obtain a competitive ratio upper bound of 
$1-O\left(\sqrt{\log(\mininventory n T)/\mininventory}\right)$
for a slightly modified version of \Cref{alg:large inventory}.
The simple concentration argument works as follows:
consider a modified version of \Cref{alg:large inventory}
where the random discarding step is fully correlated.
Namely, instead of flipping an independent coin for
each product $i\in \bar \assortment_t$,
we flip a single coin, and set
$\bar \assortment_t$
as all available products in $\hat\assortment_t$
with probability $\gamma$, and empty set otherwise.
Consider this modified \Cref{alg:large inventory}
until the first time that one of the sampled products is not available. 
Note that before such a bad event happens,
all allocations are independent over consumers,
and the sub-assortment sampling step is trivial
(i.e., return $\tilde \assortment_t = \bar \assortment_t$ deterministically). 
We can bound the probability of this bad event at each time and for each product by an exponentially small probability in $\mininventory$
due to the Chernoff bound. 
Applying union bound for every time $t\in[T]$ and 
every $i\in[n]$, and using the fact that when such a bad event happens we have no control over the revenue,
the resulting competitive ratio 
can be upper-bounded by $1-O\left(\sqrt{\log(\mininventory n T)/\mininventory}\right)$. Note that this competitive ratio 
has an extra $\sqrt{\log(nT)}$ dependence, and thus is 
strictly worse than the competitive ratio 
$1 -\loss$ stated in \Cref{thm:competitive-ratio large inventory},
which is independent of the number of products $n$ 
and the number of consumers $T$.

To prove \Cref{thm:competitive-ratio large inventory}
with the competitive ratio $1 -\loss$ 
that is independent of the number of products $n$ 
and the number of consumers $T$,
we use a careful coupling argument in our analysis of \Cref{alg:large inventory} that couples the rental indicator random variables of our algorithm with an alternative hypothetical algorithm. This hypothetical algorithm \emph{ignores} inventory constraints of all the products and only simulates the expected LP's optimal solution combined with independent discarding of each product with probability $\gamma$. This algorithm generates an independent sequence of rental indicator random variables, allowing us to use simple concentration bounds. Importantly, this coupling trick is only possible because of the guarantee of the sub-assortment sampling procedure in \Cref{eq:subassortment} (see the formal proof in \Cref{apx:alg 1 proof}).
}


\begin{remark}
Our analysis in this section mainly 
focuses on the asymptotic regime where $\mininventory$ is large; however, we can still plot the competitive ratio $1-\loss$ of \Cref{alg:large inventory}
for small values of $\mininventory$ by numerically evaluating $\loss$ using \Cref{eq:loss}.
See the black solid curve in \Cref{fig:competitive ratio}.
\end{remark}

\begin{figure}
    \centering
    \begin{tikzpicture}\begin{axis}
[height=6.5cm,width=8.5cm,
xmin=1,xmax=100,ymin=0,ymax=1,
axis x line*=bottom,
axis y line*=left,
xtick={1, 50, 100},
ytick={0.25, 0.5, 0.75, 1},
xlabel={$\mininventory$},
ylabel={competitive ratio},
]
\addplot [draw=black, line width=0.5mm] coordinates {
(1.0000, 0.0946)
(2.0000, 0.1629)
(3.0000, 0.2153)
(4.0000, 0.2573)
(5.0000, 0.2920)
(6.0000, 0.3213)
(7.0000, 0.3466)
(8.0000, 0.3688)
(9.0000, 0.3884)
(10.0000, 0.4060)
(11.0000, 0.4219)
(12.0000, 0.4363)
(13.0000, 0.4495)
(14.0000, 0.4617)
(15.0000, 0.4730)
(16.0000, 0.4834)
(17.0000, 0.4932)
(18.0000, 0.5023)
(19.0000, 0.5109)
(20.0000, 0.5189)
(21.0000, 0.5265)
(22.0000, 0.5337)
(23.0000, 0.5406)
(24.0000, 0.5471)
(25.0000, 0.5532)
(26.0000, 0.5591)
(27.0000, 0.5648)
(28.0000, 0.5701)
(29.0000, 0.5753)
(30.0000, 0.5803)
(31.0000, 0.5850)
(32.0000, 0.5896)
(33.0000, 0.5940)
(34.0000, 0.5982)
(35.0000, 0.6023)
(36.0000, 0.6063)
(37.0000, 0.6101)
(38.0000, 0.6138)
(39.0000, 0.6174)
(40.0000, 0.6208)
(41.0000, 0.6242)
(42.0000, 0.6274)
(43.0000, 0.6306)
(44.0000, 0.6337)
(45.0000, 0.6366)
(46.0000, 0.6396)
(47.0000, 0.6424)
(48.0000, 0.6451)
(49.0000, 0.6478)
(50.0000, 0.6504)
(51.0000, 0.6530)
(52.0000, 0.6554)
(53.0000, 0.6579)
(54.0000, 0.6602)
(55.0000, 0.6625)
(56.0000, 0.6648)
(57.0000, 0.6670)
(58.0000, 0.6692)
(59.0000, 0.6713)
(60.0000, 0.6733)
(61.0000, 0.6754)
(62.0000, 0.6773)
(63.0000, 0.6793)
(64.0000, 0.6812)
(65.0000, 0.6830)
(66.0000, 0.6849)
(67.0000, 0.6867)
(68.0000, 0.6884)
(69.0000, 0.6901)
(70.0000, 0.6918)
(71.0000, 0.6935)
(72.0000, 0.6951)
(73.0000, 0.6967)
(74.0000, 0.6983)
(75.0000, 0.6999)
(76.0000, 0.7014)
(77.0000, 0.7029)
(78.0000, 0.7043)
(79.0000, 0.7058)
(80.0000, 0.7072)
(81.0000, 0.7086)
(82.0000, 0.7100)
(83.0000, 0.7113)
(84.0000, 0.7127)
(85.0000, 0.7140)
(86.0000, 0.7153)
(87.0000, 0.7165)
(88.0000, 0.7178)
(89.0000, 0.7190)
(90.0000, 0.7202)
(91.0000, 0.7214)
(92.0000, 0.7226)
(93.0000, 0.7238)
(94.0000, 0.7249)
(95.0000, 0.7261)
(96.0000, 0.7272)
(97.0000, 0.7283)
(98.0000, 0.7294)
(99.0000, 0.7305)
(100.0000, 0.7315)
};
\addplot [draw=blue, dashed, line width=0.7mm] coordinates {
(1.0000, 0.5000)
(2.0000, 0.5000)
(3.0000, 0.5000)
(4.0000, 0.5000)
(5.0000, 0.5000)
(6.0000, 0.5000)
(7.0000, 0.5000)
(8.0000, 0.5000)
(9.0000, 0.5000)
(10.0000, 0.5000)
(11.0000, 0.5000)
(12.0000, 0.5000)
(13.0000, 0.5000)
(14.0000, 0.5000)
(15.0000, 0.5000)
(16.0000, 0.5000)
(17.0000, 0.5000)
(18.0000, 0.5023)
(19.0000, 0.5109)
(20.0000, 0.5189)
(21.0000, 0.5265)
(22.0000, 0.5337)
(23.0000, 0.5406)
(24.0000, 0.5471)
(25.0000, 0.5532)
(26.0000, 0.5591)
(27.0000, 0.5648)
(28.0000, 0.5701)
(29.0000, 0.5753)
(30.0000, 0.5803)
(31.0000, 0.5850)
(32.0000, 0.5896)
(33.0000, 0.5940)
(34.0000, 0.5982)
(35.0000, 0.6023)
(36.0000, 0.6063)
(37.0000, 0.6101)
(38.0000, 0.6138)
(39.0000, 0.6174)
(40.0000, 0.6208)
(41.0000, 0.6242)
(42.0000, 0.6274)
(43.0000, 0.6306)
(44.0000, 0.6337)
(45.0000, 0.6366)
(46.0000, 0.6396)
(47.0000, 0.6424)
(48.0000, 0.6451)
(49.0000, 0.6478)
(50.0000, 0.6504)
(51.0000, 0.6530)
(52.0000, 0.6554)
(53.0000, 0.6579)
(54.0000, 0.6602)
(55.0000, 0.6625)
(56.0000, 0.6648)
(57.0000, 0.6670)
(58.0000, 0.6692)
(59.0000, 0.6713)
(60.0000, 0.6733)
(61.0000, 0.6754)
(62.0000, 0.6773)
(63.0000, 0.6793)
(64.0000, 0.6812)
(65.0000, 0.6830)
(66.0000, 0.6849)
(67.0000, 0.6867)
(68.0000, 0.6884)
(69.0000, 0.6901)
(70.0000, 0.6918)
(71.0000, 0.6935)
(72.0000, 0.6951)
(73.0000, 0.6967)
(74.0000, 0.6983)
(75.0000, 0.6999)
(76.0000, 0.7014)
(77.0000, 0.7029)
(78.0000, 0.7043)
(79.0000, 0.7058)
(80.0000, 0.7072)
(81.0000, 0.7086)
(82.0000, 0.7100)
(83.0000, 0.7113)
(84.0000, 0.7127)
(85.0000, 0.7140)
(86.0000, 0.7153)
(87.0000, 0.7165)
(88.0000, 0.7178)
(89.0000, 0.7190)
(90.0000, 0.7202)
(91.0000, 0.7214)
(92.0000, 0.7226)
(93.0000, 0.7238)
(94.0000, 0.7249)
(95.0000, 0.7261)
(96.0000, 0.7272)
(97.0000, 0.7283)
(98.0000, 0.7294)
(99.0000, 0.7305)
(100.0000, 0.7315)
};
\end{axis}\end{tikzpicture}
    \caption{Competitive ratio of simulation-based algorithms:
    Black solid curve corresponds to \Cref{alg:large inventory};
    blue dashed curve corresponds to 
    the hybrid between \Cref{alg:large inventory}
    and \Cref{alg:SB} (see \Cref{sec:hybrid}).
    }
    \label{fig:competitive ratio}
\end{figure}
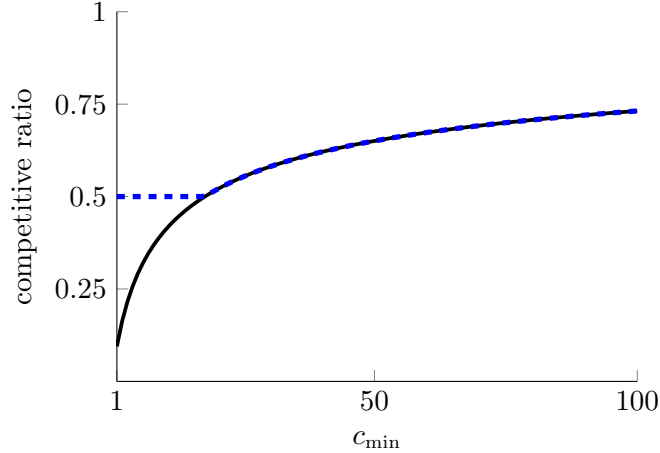

\vspace{-0.5cm}
\subsection{
Small Initial Inventory: Towards
Competitive Ratio \texorpdfstring{$\mathbf{\tfrac{1}{2}}$}{}}
\label{sec:small inventory}

In this subsection, we propose our  second simulation-based algorithm. The main difference between this algorithm and \Cref{alg:large inventory} is in the discarding step: if no units of product $i$ is available on-hand, it is discarded automatically to guarantee inventory feasibility; otherwise, it is only selected in the final assortment if $\curreward\geq \curprice$, where $\curprice$ are \emph{non-adaptive} thresholds computed by the algorithm up-front (will be discussed later). The aim of this discarding procedure is to  to guarantee that only available products with high enough  rental fees are assorted.


 Technically speaking, for each product $i$ one can consider a separate DP to optimally make discarding decisions. This DP will maximize per-unit revenue-to-go of assorted units of product $i$ over the finite time horizon $[1:T]$,  given the randomized suggestion of step (\rom{2}) (recall \Cref{sec:sketch}). The drawback is the need for a high-dimensional state variable that keeps track of the on-hand product inventories, as well the inventory of units of the product that are in use (and will return to inventory at different times). A major ingredient of our algorithm is to replace this high-dimensional DP with a simple one that is \emph{inventory independent}, and uses an optimistic upper-bound of $\inventory_i$ on the actual inventory in the Bellman equation for updating optimal per-unit revenue-to-go of product $i$.

 \subsubsection{Dynamic Programming for Per-unit Revenue-to-go with Replenishment.} 
 \label{sec:stochastic-finite-simple-DP}
In the rest of this subsection, let $\curexantetype\triangleq\curalloctypeopt\typedistribution_t(\type_t)$ for every $\assortment$,$t$ and $\type_t$. Suppose at each time $t$, a new independent consumer type $\type_t\sim \typedistribution_t$ is realized. Let $\hat{S}\sim\curexantetype$ denote the randomized subset sampled in step~(\rom{2})(simulation step). Fix a product $i\in[n]$, with initial inventory of $\inventory_i$ units. Now, consider a hypothetical scenario where an exogenous process \emph{replenishes the inventory} at each time to guarantee we always have $\inventory_i$ units of the product on-hand, no matter how many units are currently under rental. In this new problem, the goal is to design an online policy to discard or accept units of the reusable product once suggested in step~(\rom{2}), in order to maximize the per-unit revenue-to-go of renting this product. We can formulate this problem using a simple dynamic programming, where $\valueDPsimple_{i,t}$ is the optimal per-unit revenue-to-go of product $i$ during time interval $[t:T]$. Compared to the original high-dimensional DP for solving the optimum discarding, this DP is \emph{optimistic}
in that 
it ``imagines'' the deficiency in inventory
is replenished every period.

 As a convention, let $\valueDPsimple_{i,T+1}=0$. To write the Bellman update equation of the optimistic DP at time $t$ using backward induction, suppose type $\type_t$ is realized and $\hat{\assortment}=\assortment$ (which happens w.p. $\curexantetype$). If optimal policy decides to discard product $i$, then per-unit revenue-to-go will be $\valueDPsimple_{i, t + 1}$. 
 If optimal policy decides to not discard $i$, 
 then with probability 
$(1-\curchoice)$ the per-unit revenue-to-go will still be  $\valueDPsimple_{i, t + 1}$.\footnote{In this hypothetical scenario, 
we assume that the 
probability that the consumer select product $i$ equals $\curchoice$
regardless of whether another product $i'$ is discarded from $\assortment$.} 
 However, with probability $\curchoice$ consumer rents one of the $\inventory_i$ units (remember that inventory will always be full), and therefore generates a total revenue-to-go of $(\inventory_i-1)\valueDPsimple_{i, t + 1}$ (i.e., due to contribution of units not rented at time $t$; these units will transfer to the inventory at time $t+1$) plus $\curreward+\valueDPsimple_{i, t + d}$ upon realization of rental time 
 $d\sim\durationdistribution^{\type_t}_i$ (i.e., due to contribution of the rented unit). To summarize, we will have the following Bellman update equation:
\begin{align}
		\label{eq:bellman-equation-finite}
		\begin{split}
			&\valueDPsimple_{i,t}= \displaystyle\sum\nolimits_{\type_t \in \typespace_t}
		\displaystyle\sum\nolimits_{\assortment \in \assortmentspace}
		\curexantetype \\
& 
		\times\max\left\{\valueDPsimple_{i, t + 1}~,~\left(1 - \curchoice\right)\valueDPsimple_{i, t + 1}+ \curchoice\left(\frac{1}{\inventory_i}\sum\nolimits_d\durationpdf_i^{\type_t}(d)\left(\curreward+\valueDPsimple_{i, t + d}\right)+\frac{\inventory_i-1}{\inventory_i}\valueDPsimple_{i, t + 1}\right )\right\}
	\end{split}
	\end{align}

Note the update rule of the above dynamic programming can be simplified by rearranging the terms; Interestingly, the rule will be independent of $\inventory_i$ and $\curchoice$ as they cancel out: 
\begin{equation}
\left[\textrm{At time $t$ with type $\type_t$, $i\in \hat{\assortment}$ will be accepted}\right] \Longleftrightarrow \curreward \geq \valueDPsimple_{i, t + 1}-\displaystyle\sum\nolimits_d\durationpdf^{\type_t}_i(d)\valueDPsimple_{i, t + d}
\end{equation}

\begin{remark}
 {Later in \Cref{sec:discussion} and \ref{sec:stochastic-infinite}, we will replace this simple DP with a slightly modified one that has \emph{inventory dependent} state, but is still low dimensional when rental times are infinite. This allows us to obtain (an almost) optimal competitive ratio for this special case.}
\end{remark}
\subsubsection{The Algorithm and Analysis}
We now present our second simulation-based algorithm (\Cref{alg:SB}), with its competitive ratio guarantee (\Cref{thm:stochastic-finite-competitive-ratio}).
\footnote{{The competitive ratio in 
\Cref{thm:stochastic-finite-competitive-ratio} is 
optimal even if rental times are infinite.
Consider the following example:
there is a single non-reusable product with a single unit.
There are two time periods $T = 2$.
Consumer 1 has a deterministic type
that deterministically purchases this item with fee $1$.
With probability $\epsilon$,
consumer 2 has a type that
deterministically purchases this item with fee ${1}/{\epsilon}$.
Otherwise (i.e., with probability $1-\epsilon$),
consumer 2 has a type that purchases nothing.
In this example, the expected revenue of the Bayesian Expected LP benchmark as well as 
the clairvoyant policy is $2 - \epsilon$,
while the expected revenue of any online policy is 
at most $1$.}}

\begin{algorithm}[htb]
 	\caption{Simulation-based Algorithm with Non-adaptive  Per-unit Revenue Thresholds}
 	 \label{alg:SB}
 	\vspace{2mm}
 	
 	\emph{\underline{Pre-processing:}}
 	\begin{itemize}
 	    \item Compute the optimal assignment $\{\curalloctypeopt\}$ of $\EAR{\typedistributionsequence}$ by invoking the offline assortment oracle  (\Cref{asp:oracle})
 	    \item 	Set $\curexantetype\triangleq\curalloctypeopt\typedistribution_t(\type_t)$ for every $\assortment$,$t$ and $\type_t$ where optimal assignment has a non-zero entry
 	    \item Solve the dynamic programming with Bellman update described in \cref{eq:bellman-equation-finite} and boundary condition $\valueDPsimple_{i,T+1}=0$ for every product $i$ to obtain $\{\valueDPsimple_{i,t}\}_{i\in[n],t\in[\totaltime]}$
 	    \item Let $\curprice\triangleq \valueDPsimple_{i, t + 1}-\displaystyle\sum\nolimits_d\durationpdf_i^{\type_t}(d)\valueDPsimple_{i, t + d}$, for all $i\in[n],t\in[\totaltime],\type_t\in\typespace_t$.
 	\end{itemize}

 	\vspace{2mm}
 	
 	\For{$t=1$ to $T$}{
 	\tcc{consumer $t$ with type $\type_t\sim\typedistribution_t$ arrives} 
 	\vspace{1mm}
 	
 	\emph{\underline{Simulation:}} Upon realizing consumer type $\type_t$, sample $\hat{\assortment}_t\sim \{\curalloctypeopt\}_{\assortment\in\assortmentspace}$
 	
 	\vspace{2mm}
 	\emph{\underline{Discarding:}} Initialize $\bar\assortment_t\gets\hat\assortment_t$
 	
 	\For{each product $i\in\hat\assortment_t$}{
 	\vspace{1mm}
 	\If {$\curreward<\curprice$ or there is no available unit of product $i$}{
 	Remove $i$ from $\bar\assortment_t$}}
 	\tcc{Per-unit revenue thresholds $\{\curprice\}$ are computed once, i.e., are non-adaptive}
 	\vspace{2mm}
 {
 	\emph{\underline{Post-processing:}} Let $\tilde{S}_t\gets \textsc{Sub-assortment Sampling} \left(\choice^{\type_t}, \bar\assortment_t,\{\choice(\hat\assortment_t, i)\}_{i\in \bar\assortment_t}\right)$
 	
 	\tcc{Send a query call to Procedure~\ref{alg:sample assortment} with appropriate input arguments}}
 	\vspace{2mm}

 	
 	Offer assortment $\tilde{\assortment}_t$ to consumer $t$
 	
 	}
\end{algorithm}

	

\begin{restatable}{theorem}{algtwocc}
	\label{thm:stochastic-finite-competitive-ratio}
	The competitive ratio of \Cref{alg:SB} against offline Bayesian expected LP benchmark, i.e., $\EAR{\typedistributionsequence}$, is at least $1/2$. Moreover, it runs in time $\textrm{Poly}(n,T, \sum_{t\in[\totaltime]}\lvert\typespace_t\rvert)$ given oracle access to an offline algorithm for assortment optimization (\Cref{asp:oracle}).
\end{restatable}


\ifalgtwoappendix
\else
{\begin{proof}{\emph{Proof sketch of \Cref{thm:stochastic-finite-competitive-ratio}}.} 
The running time is proved by \Cref{prop:EAR running time}, and the fact that the simple DP in \Cref{sec:stochastic-finite-simple-DP} can be solved in polynomial time. The analysis of the competitive ratio can be decoupled across products. For each fixed product $i$, we do the analysis in two parts, each sketched below (see the full details in \Cref{apx:alg 2 proof}):
\begin{itemize}
    \item \texttt{Part~(\rom{1}}) -- \Cref{apx:part1}: we first compare \Cref{alg:SB} with the simple optimistic dynamic programming described in \Cref{sec:stochastic-finite-simple-DP} and show the total expected revenue of \Cref{alg:SB} due to rentals of product $i$ is at least $\inventory_i\valueDPsimple_{i,1}$. We prove this claim using induction, and the fact that showing a subset $\tilde{S}_t$ of sampled assortment $\hat\assortment_t$ can only increase the revenue-to-go of the discarding policy that follows the thresholds of the optimistic DP (as in the algorithm). 
    \item \texttt{Part~(\rom{2}}) -- \Cref{apx:part2}: We then compare this simple dynamic programming with expected LP benchmark and show for each product $i$, $\inventory_i\valueDPsimple_{i,1}$ is at least $1/2$ of the contribution of product $i$ to the optimal objective value of $\EAR{\typedistributionsequence}$ 
(\texttt{Part~(\rom{2}})). In order to prove this part, we use the connection between the optimistic DP of \Cref{sec:stochastic-finite-simple-DP} and a related factor-revealing LP that characterizes the competitive ratio of the optimistic DP. This connections leads us to apply duality arguments to find a lower-bound on the ratio of $\inventory_i\valueDPsimple_{i,1}$ and the contribution of product $i$ to the optimal objective value of $\EAR{\typedistributionsequence}$.
\end{itemize}
\end{proof}}

\subsection{Hybrid between \Cref{alg:large inventory} and \Cref{alg:SB}}
\label{sec:hybrid}

{\noindent\textbf{Best of both worlds discarding.} In both \Cref{alg:large inventory}
and \Cref{alg:SB},
we have discarding policies (one for each of the products) that run 
independently form each other. Also, both competitive ratio analyses essentially decouple across different products, as we analyze the revenue performances of these discarding policies for each product separately. Moreover, in both analyses, we compare the expected revenue of each product $i$ with the contribution of that product in the expected-LP. 

Considering all of the above design and analysis aspects of our two algorithms, we can propose a hybrid algorithm
where we decide on the choice of the discarding policy for each product $i$ based on its initial inventory $\inventory_i$ upfront. Once we finalize these choices, we run the (possibly different) discarding algorithms in parallel and separately for different products during the discarding step of our final hybrid algorithm. In order to achieve the best of both worlds revenue performance guarantees of small and large inventory regimes, we partition the set of products into those with large initial inventory and those with small initial inventory (will be formally defined later) at the beginning. Given this partition, we assign a ``randomized discarding'' policy (as described in
in \Cref{alg:large inventory}) to make discarding decisions of product $i$ across times $t\in[T]$
if $\inventory_i$ is large, and we use a  ``discarding with per-unit revenue thresholds'' (as described in
in \Cref{alg:SB}) if $\inventory_i$ is small. 

To distinguish between large and small $\inventory_i$, we first solve the dynamic programming of the optimistic DP discarding policy in \Cref{sec:small inventory} for each product $i\in[n]$, by using its Bellman update equation (described in  \Cref{eq:bellman-equation-finite}). We then use the value function $\valueDPsimple_{i,1}$ of the optimistic DP for product $i$ to label this product as either large or small inventory. In particular, define $\mathcal{R}_i$ to be the ratio between between the expected revenue-to-go of the optimistic DP for product $i$ and the contribution of this product to the expected-LP's objective, i.e., 
$$
\mathcal{R}_i\triangleq \frac{c_i\valueDPsimple_{i,1}}{\displaystyle\sum\nolimits_{t=1}^{\totaltime}
			\displaystyle\sum\nolimits_{\type_t \in \typespace_t}
			\displaystyle\sum\nolimits_{\assortment \in \assortmentspace}\curexantetype\curchoice
			\curreward}
$$
Note that $\mathcal{R}_i\in[0.5,1]$ from \Cref{thm:stochastic-finite-competitive-ratio}. Also note that $\varepsilon^*(c)$ is continuous and monotone increasing in $c$, $\varepsilon^*(0)=1$, and $\varepsilon^*(+\infty)=0$ --- see \Cref{eq:loss} for the definition $\varepsilon^*(\cdot)$. We next compare the ratio $\mathcal{R}_i$ with $1-\varepsilon^*(c_i)$. In fact, if $\mathcal{R}_i+\varepsilon^*(c_i)<1$, we then except the competitive ratio of randomized discarding to be no smaller than that of the optimistic DP, and hence we label the product as ``large inventory''. Otherwise, we expect the optimistic DP to beat the randomized discarding in terms of the competitive ratio, and hence we label the product as ``small invenotry''. Finally, we perform the sub-assortment sampling as a post-processing step in order to correct the choice probabilities of non-discarded products (which might have been increased due to weak-substitution and that other products are either discarded or not even selected in the assortment as they were not available at the first place). We denote the resulting algorithm by \texttt{Sim+Hybrid(i)}.
\begin{theorem}
\label{thm:hybrid}
	The competitive ratio of \texttt{Sim+Hybrid(i)} against offline Bayesian expected LP benchmark, i.e., $\EAR{\typedistributionsequence}$, is at least $1-\min\left(\frac{1}{2},\loss\right)$ (see \Cref{fig:competitive ratio}).
\end{theorem}


\begin{proof}{\emph{Proof sketch.}}
As it can be seen by following the lines of both the proof of \Cref{thm:competitive-ratio large inventory} in \Cref{sec:analysis} and the proof of \Cref{thm:stochastic-finite-competitive-ratio} in \Cref{apx:alg 2 proof}, the arguments for the revenue performance of the corresponding discarding policies for each product $i$ (and also the resulting competitive ratio) is independent of how discarding of another product $i'$ is handled, and therefore the competitive ratio guarantees of the randomized discarding and the optimistic DP from these theorems still hold for the hybrid algorithm. For brevity, we do not repeat these proofs. We conclude that the proposed hybrid algorithm has a competitive ratio of at least $\max\left(\frac{1}{2},1-\loss\right)$.\hfill\halmos

\end{proof}}

\medskip
{\noindent\textbf{Monte-Carlo simulation to help.}
In principle, given the sequence of type distributions $\typedistributionsequence$, one can simulate both \Cref{alg:large inventory} and \Cref{alg:SB} and estimate their expected future revenues using Monte-Carlo simulation, starting at any time $t\in[T]$ (given any history up to time $t$). Now, a simple hybrid algorithm can switch to the algorithm with the higher expected revenue and runs this algorithm for the next time step given the current history of rental products. By repeatedly applying this method at each time $t$ given the history up to this time  -- a techniques known as the \emph{method of conditional expectation} -- we end up with an alternative hybrid algorithm that essentially is the Be-The-Leader policy among the two policies at each time, meaning that its expected future revenue at each time is at least the expected future revenue of each of the two policies (which can be proved using induction). Hence, this hybrid algorithm clearly obtains the best of both worlds competitive ratio of $1-\min\left(\frac{1}{2},\loss\right)$ similar to our previous hybrid algorithm. Such a policy can also choose to switch at a lower frequency; but no matter what frequency it picks, it is expected to outperform both policies in expectation, both in theory and practice. We use this alternative hybrid algorithm, which is denoted by \texttt{Sim+Hybrid(ii)}, in our numerical simulations in \Cref{sec:numerical} as well.

\begin{theorem}
\label{thm:hybrid 2}
	The competitive ratio of \texttt{Sim+Hybrid(ii)} against offline Bayesian expected LP benchmark, i.e., $\EAR{\typedistributionsequence}$, is at least $1-\min\left(\frac{1}{2},\loss\right)$ (see \Cref{fig:competitive ratio}).
\end{theorem}

\begin{proof}{\emph{Proof sketch.}}
Consider the following backward induction.
The induction hypothesis is that 
given current state $\mathbf{J}$ (i.e., 
the number of available units of each product in the inventory, as well as the return time of each allocated unit of every product),
and time period $t$,
the expected revenue from time $t$ to time $T$
in \texttt{Hybrid-Sim(ii)} 
is weakly higher than 
\Cref{alg:large inventory} and \Cref{alg:SB}.
Base case $t = T$ is straightforward.
Suppose the induction hypothesis
is correct for $t + 1, \dots, T$.
Consider the inductive step for time $t$
and current state $\mathbf{J}$.
Suppose Monte-Carlo simulation suggests that 
given state $\mathbf{J}$,
\Cref{alg:large inventory}
achieves higher expected revenue from time $t$ to time $T$,
and its induced new state is $\mathbf{J}'$
(the analysis for the other case is similar).
In this case, we know that
the expected revenue 
in \texttt{Hybrid-Sim(ii)} 
from time $t$ to time $T$
can be decomposed
into the following two terms:
(i) the expected revenue in \texttt{Hybrid-Sim(ii)}
at time $t$,
and 
(ii)
the expected revenue in \texttt{Hybrid-Sim(ii)} 
from time $t + 1$ to time $T$ 
under state $\mathbf{J}'$.
By construction,
term (i)
equals to 
the expected revenue in \Cref{alg:large inventory}
at time $t$.
By our induction hypothesis for time $t + 1$
with state $\mathbf{J}'$,
term (ii) is weakly higher than 
the expected revenue in \Cref{alg:large inventory}
from time $t + 1$ to time $T$
given state $\mathbf{J'}$.
Hence, 
the expected revenue 
in \texttt{Hybrid-Sim(ii)} 
from time $t$ to time $T$ given state $\mathbf{J}$
is weakly higher than \Cref{alg:large inventory}
and \Cref{alg:SB} as well,
which concludes the backward induction.
\hfill\halmos
\end{proof}

\begin{remark}
\label{remark:hybrid}
We would like to highlight that while both  \texttt{Sim+Hybrid(i)} and \texttt{Sim+Hybrid(ii)} attain the theoretical best of both worlds competitive ratio that was mentioned earlier, and as we see in our numerical simulations in \Cref{sec:numerical} they both outperforms other existing policies in practical scenarios of our problem, they differ in terms of computational requirements. In fact, \texttt{Sim+Hybrid(i)} can easily make upfront decisions for the choice of discarding policy of each product $i$, with almost no extra computation compared to \Cref{alg:large inventory} and \Cref{alg:SB}; nevertheless, \texttt{Sim+Hybrid(ii)} needs to run Monte-Carlo simulation several times (depending on the switching frequency) by sampling from future types, which makes it less practically appealing. 
\end{remark}}

 \section{Conclusion}
 \label{sec:conclusion}

We studied designing near-optimal algorithms for the online assortment of reusable resources in the Bayesian setting. We proposed an algorithmic framework based on four modular steps: (\rom{1}) solving the expected LP, (\rom{2}) simulating the solution, (\rom{3}) running a separate discarding procedure for each product to maintain point-wise inventory feasibility (while only losing a negligible fraction of the revenue of each product), and (\rom{4}) performing a post-processing step to adjust choice probabilities of non-discarded items. Using this framework, we designed an algorithm that is $1-\min\left(\frac{1}{2},O\left(\sqrt{\log(\mininventory)/\mininventory}\right)\right)$ under the general rental duration distributions, and an improved near-optimal algorithm with competitive ratio  $1-1/\sqrt{(\mininventory+3)}$ under infinite rental durations. Not only our algorithms outperform the existing algorithms in the literature theoretically, we further verified their revenue performance advantages through numerical simulations.

 As a road-map for future, it is interesting to study what other practical aspects of a real-world assortment problem beyond reusable resources can be modeled, and to what extent mathematical programming techniques can be used to design competitive algorithms there. On the technical side, the most immediate open problem stemming from our work is finding the optimal competitive ratio for the case of general rental duration distributions. In particular, can one shave the logarithmic factor in our competitive ratio and obtain a $1-O(1/\sqrt{\mininventory})$ competitive algorithm, similar to the best known competitive ratio in the non-reusable case? As a  different yet more ambitious future direction, it would be interesting to study classes of stochastic online optimization similar to the Bayesian online assortment further, in order to discover the computational hardness of computing or approximating the optimum online policy, i.e., the DP policy. An interesting discovery here would be obtaining improved approximations against the optimum online benchmark versus the expected LP benchmark through polynomial-time policies, à la \cite{ANSS-19}, or proving its impossibility.

 \section*{Acknowledgement} We would like to thanks the anonymous associate editor and the referees for their comments and feedback throughout the revision process. Rad Niazadeh's research is partially supported by an Asness Junior Faculty Fellowship from the University of Chicago Booth School of Business. 
 

\setlength{\bibsep}{0.0pt}
\bibliographystyle{plainnat}
\OneAndAHalfSpacedXI
{\footnotesize
 \bibliography{refs}}
\newpage
\newpage
\ECDisclaimer
\ECSwitch
\section{Further Related Work}
\label{sec:furtherreview}
\revcolor{
Assortment planning for revenue management has an extensively growing 
literature over the recent decades. We refer the reader to  related surveys and books \citep[cf.][]{KFV-08, lan-90, HT-98} for a comprehensive study.
\citet{VM-99} study the static model 
which captures the trade-offs between
inventory costs and 
product variety under the multinomial logit 
consumer choice model. 
Later work have 
considered assortment under various consumer choice models,
e.g., 
demand substitution model \citep{SA-00}, multinomial logit models
\citep{TV-04, GIPD-04, LV-08, top-13},
the Lancaster choice model
\citep{GH-06},
ranked-list preference 
\citep{HGS-10, GLS-16,aouad2018approximability}, consider-then-choose choice models \citep{AFL-15}, and non-parametric (data-driven)
choice models \citep{FJS-13}.

Online assortment has been studied  more recently, both under the prior-free/adversarial and the Bayesian settings. For non-reusable products,
\citet{BKX-15} study a model with two products, i.i.d.\ consumer types and Poisson arrivals. \citet{CF-09} study a model with non-stationary consumer types and show $\tfrac{1}{2}$ performance guarantee
with respect to the clairvoyant optimum online benchmark. \citet{GNR-14} introduce ``inventory balancing" algorithms -- inspired by the seminal work of \cite{MSVV-05} for online ad allocation -- and analyze their performance guarantee in the prior-free setting using a primal-dual approach. This analysis is later improved and generalized  by \cite{MS-19}. {\citet{CMSX-16}
consider another variant, in which assortments are 
offered as add-ons with discounted price. 
The idea of discarding  
to ensure point-wise inventory feasibility is also used in this paper,
and they face a 
similar challenge due to product substitution.
To overcome this challenge, the authors introduce 
a sampling procedure with error for their specific goal.
In contrast, our sub-assortment sampling 
does not incur any errors and can be used for 
more general purposes.}

For online assortment of reusable products,
\citet{LR-10} study a model that assumes independent demands across products, without any choice behavior for the customers with infinite selling horizon. A follow up work from \citet{OS-18} extend this work to incorporate customer choice behavior and a finite selling horizon. \citet{CLS-17} study a model with multiple units of a single reusable product. They consider stochastic usage durations and advance reservations, and obtain data-dependent performance guarantees that are asymptotically optimal when the product inventory and the customer arrival rate scale up linearly at the same rate. Finally, \citet{RST-17} study the Bayesian setting and \citet{GGISUW-19,feng2019linear,feng2021online,goyal2020online} study the prior-free/adversarial setting 
{where both rental fees
and rental duration distributions
are type-independent and 
identical across time.}

In the Bayesian setting, our problem resembles some aspects of the ``prophet inequality" problem. This problem is originated from the seminal work of \cite{KS-78} in the 70's, and has since been studied quite extensively. Combinatorial generalizations of this problem are also studied in the literature. Examples are prophet inequalities for  matroids~\citep{SC-84,HKS-07,KW-12}, matchings~\citep{AHL-12}, and combinatorial auctions~\citep{DFKL-17}. In these generalizations, the natural dynamic programming for computing the optimum online policy are exponential-size~\citep{NSS-18,ANSS-19}, similar to our problem. See \cite{L-17}, for a comprehensive survey.

{Some of our techniques resemble those used for the prophet inequality matching \citep{AHL-12}, the Magician's problem and online contention resolution schemes~\citep{ala-14,brubach2021improved,FSZ-16}, the pricing with static calendar problem~\citep{MSZ-18}, 
online bipartite matching with 
reusable offline nodes~\citep{DSSX-18},
and the volunteer crowdsourcing problem~\citep{manshadi2020online}. The closest to us is the work of \citep{BM-19} for 
the network revenue management problem with 
reusable resource, in which the authors have independently and concurrently discovered policies similar to our algorithm for the small inventory regime. 

Other expected LP benchmarks have been used in the literature under the name of ``ex-ante relaxation'' for various stochastic online optimization and mechanism design problems. For example, see \cite{AHL-12,DFKL-17,LS-18,VB-18,MSZ-18,ANSS-19}. The factor-revealing technique we use to analyze our DP-based thresholding discarding rule resemble the LP-method of \cite{A-07}, the dual-fitting methods in \cite{ala-14,wang2018online}, and also other LP-based methods for approximate dynamic programming (see \cite{SBAPW-04} for a comprehensive study).}

Some of our techniques resemble those used for the prophet inequality matching \citep{AHL-12}, the Magician's problem~\citep{ala-14}, the pricing with static calendar problem~\citep{MSZ-18} (which also studies static assortment policies without reusable resources), and the volunteer crowdsourcing problem~\citep{manshadi2020online}. Similar to our Bayesian expected LP, other expected LP benchmarks have been used in the literature under the name of ``ex-ante relaxation'' for various stochastic online optimization and mechanism design problems. For example, see \citet{AHL-12,DSA-12,FSZ-16,DFKL-17,CDHKMS-17,LS-18,VB-18,MSZ-18,ANSS-19,DSSX-18}. 

Another interesting line of work related to us is the work of \cite{VB-18} and \cite{banerjee2020uniform}. The goal here is to obtain improved (additive or regret-based) approximations for various stochastic online optimization problems with packing constraints. These papers diverges from ours by taking a different technical approach and considering a different regime (i.e., when inventories increase with horizon); nevertheless, carrying over their techniques to the reusable resources problem stands as an interesting future research direction.
}

\section{Comparison Between
Different Benchmarks}
\label{apx:benchmark comparison}
\revcolor{
Consider the following hierarchy of revenue benchmarks, based on the given information on future uncertainty and required computational power:
\begin{enumerate}
    \item \emph{Optimum offline:} the expected revenue of the optimum offline algorithm that has full information about the realized type sequence $\typesequence$, exact realizations of rental durations for each type $\type_t$ and product $i$, and exact realizations of consumer choices for each possible assortment. 
    \item \emph{Clairvoyant optimum online:} the expected revenue of the optimum online algorithm that has full information about the realized type sequence $\typesequence$, but does not know the exact realizations of rental durations, nor the consumer choices. 
    \item \emph{Non-clairvoyant optimum online:} the expected revenue of the optimum online algorithm that only knows the sequence of type distributions $\typedistributionsequence$.
\end{enumerate}

\begin{figure}[htb]
\includegraphics[width=0.9\textwidth]{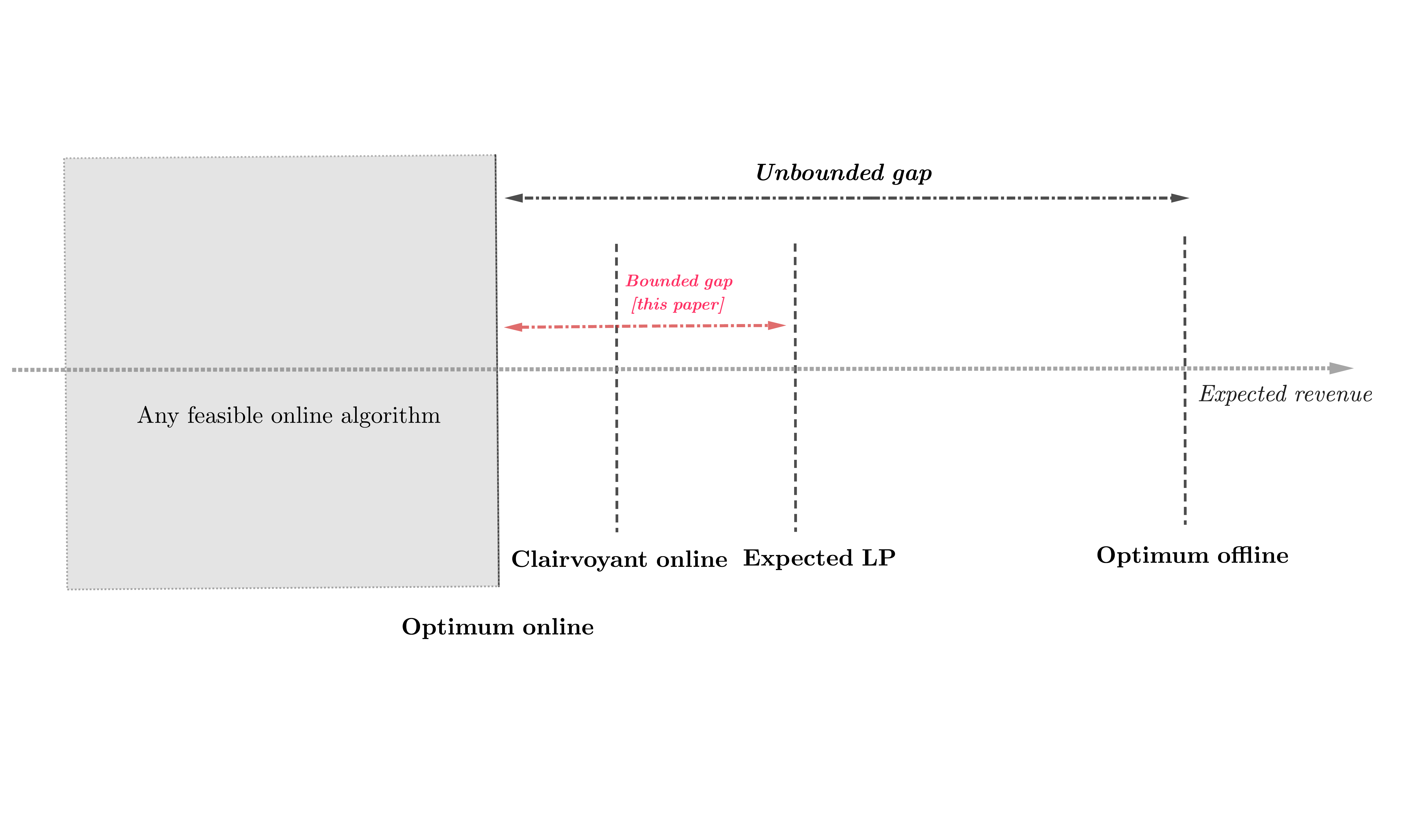}
\centering
\caption{\label{fig:benchmarks}Comparing different revenue benchmarks}
\end{figure}
\noindent \Cref{fig:benchmarks} compares these benchmarks in terms of their expected revenue. No constant competitive online algorithm exists against optimum offline, even when resources are not reusable; see Appendix~\ref{apx:hardness-optoffline}. (Non-clairvoyant) Optimum online is a weaker benchmark and requires solving an exponential-size dynamic programming~\citep{RST-17}. Clairvoyant optimum online -- introduced by \cite{GNR-14} for non-reusable products and extended to reusable products by \cite{GGISUW-19} -- is a middle ground benchmark in terms of expected revenue; moreover, it has to solve an exponential-size dynamic programming to compute its assortment decisions given the extra information $\typesequence$ (which is not provided to a normal online algorithm).

\subsection{Unbounded Competitiveness of Optimum Offline}
\label{apx:hardness-optoffline}

In this subsection, we construct a simple instance
which
shows 
the Bayesian expected LP $\EAR{\typedistributionsequence}$
is at most a $O(\sfrac{1}{T})$-approximation to the optimum offline.
Combining with \Cref{prop:relaxation},
it implies that no online algorithm can achieve a ($T$-independent) bounded 
competitive ratio against optimum offline.

\begin{theorem}
\label{thm:hardness-optoffline}
The Bayesian expected LP $\EAR{\typedistributionsequence}$
is at most a $O(\sfrac{1}{T})$-approximation to the optimum offline.
\end{theorem}

\begin{corollary}
No online algorithm
can achieve a 
($T$-independent) 
bounded 
competitive ratio against optimum offline.
\end{corollary}

\begin{proof}{\textsl{Proof of \Cref{thm:hardness-optoffline}.}}
Consider the following instance
where
there is a single product
and a single consumer type.
For this product, its initial inventory is one, 
reward is one.
Once this product is allocated, the realized rental duration 
is one (i.e., becoming available in the next time period) with probability $\sfrac{1}{2}$,
and infinite otherwise.
At every time period $t \in [T]$,
a consumer with this single consumer type arrives the platform
with probability one. 
Consumers select this product (if assorted) with probability one.

First, we claim that the expected revenue of the optimum offline
is at least $\sfrac{T}{2}$.
Note that the optimum offline observe the realization of rental
duration of
the product at each time period before assorting, and thus it can assort
the product whenever the realized rental duration at that time 
period is one. 
Since 
the expected number of such time period is $\sfrac{T}{2}$,
the expected revenue of optimum offline is at least $\sfrac{T}{2}$.

Next, we consider the Bayesian expected LP 
$\EAR{\typedistributionsequence}$
and its dual program on this instance:
\begin{equation*}
\arraycolsep=1.4pt\def\arraystretch{1}
\begin{array}{llllllll}
\max\quad\quad\quad&
\displaystyle\sum_{t=1}^T
y_t
\qquad\qquad\qquad&\text{s.t.}&
&\quad\quad\text{min}\quad\quad\quad &
\displaystyle\sum_{t=1}^T
\theta_t
\qquad\qquad\qquad&\text{s.t.} \\[1.4em]
 &
 \displaystyle\sum_{\tau=1}^t{\frac{1}{2}\,y_\tau}\leq 1 & 
 t\in[T]~.
 & 
& &
 \displaystyle\sum_{\tau=t}^T{\frac{1}{2}\,\theta_\tau}\geq 1 & 
 t\in[T]~.
\end{array}
\end{equation*}
Consider the following feasible dual solution 
with objective value 2
in the dual program:
$$\theta_t = 0,~\forall t\in[T - 1], ~~~~ \theta_T = 2.$$
Invoking the weak duality between the above primal-dual linear
programs finishes 
the proof.
\hfill\Halmos
\end{proof}
}
 \section{Discussions and Subsidiary Results}
 \label{sec:discussion}

{
\subsection{Static vs.\ dynamic substitution and sub-assortment sampling.}
\label{sec:discussion static dynamic}
In this section, we discuss the connection between
the sub-assortment sampling 
and the assortment optimization for both static substitution and
dynamic substitution studied in the literature 
\citep[cf.][]{MSZ-18}.
We first introduce the definitions 
of static substitution and dynamic substitution,
then discuss their connection 
with the sub-assortment sampling 
in the offline (single-shot) assortment problem,
and finally we extend the discussion 
to the online (multi-period) assortment problem 
considered in this paper.

In our model, 
the platform
can only form assortments of
\emph{available} products.
An alternative modeling assumption is that 
the platform can form assortments
regardless of availability of the products.
In this alternative model, 
we distinguish two setups 
on 
how consumers make their decision based 
on the displayed assortment and 
availability of the products.
\begin{itemize}
    \item \emph{Static substitution:}
    Given assortment $\assortment$
    and availability status of all the products,
    a consumer with choice model $\choice$ 
    chooses item $i$ with probability 
    $\choice(\assortment,i)$,
    and 
    the sale is only final if the selected product $i$ is available.
    
    \item \emph{Dynamic substitution:}
    Given assortment $\assortment$
    and availability status of all the products,
    a consumer with choice model $\choice$
    chooses item $i$ 
    with probability $\choice(\hat\assortment, i)$
    where assortment $\hat\assortment \subseteq \assortment$
    contains all available products in $\assortment$.
    By definition, the consumer never chooses 
    an unavailable product.
\end{itemize}

Notably, in the {offline (single-shot)} assortment problem,
the sub-assortment sampling builds the following 
connection between the static substitution and dynamic substitution.

\begin{proposition}
For any weak substitutable and downward-closed feasible choice model $\choice$,
any assortment $\assortment\in \assortmentspace$,
and any possible availability of products, 
there exists an randomized assortment $\tilde\assortment$
such that 
for every product $i\in[n]$,
the probability that product $i$ is allocated 
under assortment $\assortment$
in the static substitution
equals to
the expected probability (over the randomness of $\tilde\assortment$) that product $i$ is allocated under
assortment $\tilde\assortment$ 
in the dynamic substitution.
\end{proposition}
\begin{proof}{\textsl{Proof.}}
Invoking \Cref{prop:assortment polytope}
by setting $\choice\gets \choice$,
$\assortment\gets\assortment$
and $p_i\gets \choice(\assortment, i)$
for all $i\in\assortment$
finishes the proof.
\hfill\halmos
\end{proof}

In the online (multi-period) assortment problem, our simulation-based algorithms first samples assortment $\hat\assortment_t$ based on the optimal
assignment of the Bayesian expected LP $\EAR{\typedistributionsequence}$
and then constructs assortment 
$\bar \assortment_t$ by removing unavailable products 
as well as products that the discarding procedure suggests to discard.
We adapt the definition 
of static substitution as follows.
\begin{itemize}
    \item \emph{Static substitution with cancellation:}
    Given assortment $\assortment$,
    a consumer with choice model $\choice$
    chooses product $i$ with probability $\choice(\assortment,i)$.
    After observing the realized choice of the consumer,
the platform can take back this product $i$
immediately with no cost for the purpose of discarding
or due to the unavailability.
\end{itemize}
Under static substitution with cancellation,
we can simplify our simulation-based algorithms 
(\Cref{alg:large inventory} and \Cref{alg:SB})
and preserve the same competitive ratio guarantee 
(\Cref{thm:competitive-ratio large inventory}
and \Cref{thm:stochastic-finite-competitive-ratio})
by modifying the post-processing step as follows:
offer assortment $\sampleassortment_t$
to consumer $t$, 
and take back product $i$ selected by
consumer $t$ if $i\not\in\bar\assortment_t$.
The competitive ratio guarantee for 
the aforementioned
modified \Cref{alg:SB} follows exactly the same argument as before.
To see why the competitive ratio is preserved 
for the aforementioned modified \Cref{alg:large inventory},
note that the similar argument as the one 
in \Cref{thm:competitive-ratio large inventory}
for \Cref{alg:large inventory}
still holds.
In particular, we can consider 
 a hypothetical scenario where we run 
 the modified \Cref{alg:large inventory}
 without inventory constraint and 
 the normal run of the modified \Cref{alg:large inventory};
 Interestingly, under static substitution with cancellation,
 the coupling between two scenarios 
 this time becomes trivial, and then the 
 concentration argument for the 
 hypothetical scenario concludes the argument.
 
 On the other hand, the online (multi-period) assortment 
 problem with dynamic substitution 
 is essentially the exact problem considered in 
 the main text.
 In this sense, the sub-assortment sampling 
 and the coupling argument in the proof of
 \Cref{thm:competitive-ratio large inventory}
 can be interpreted as a reduction 
 from the online assortment problem with
 dynamic substitution
 to the online assortment problem 
 with static substititution (with cancellation).
 }

{
\subsection{Improved competitive ratio for non-reusable case.}
\label{sec:discussion non-reusable}
As a side result and mainly for the purpose of our numerical simulations,  we consider the special case of Bayesian online assortment for non-reusable resources in \Cref{sec:stochastic-infinite}, which was studied in \citet{RST-17}. We follow our simulation-based approach, and then similar to \Cref{sec:stochastic-finite-simple-DP} we use a dynamic programming that captures the expected revenue-to-go of the optimal discarding policy for this product, when the product receives (probabilistic) suggestions from the simulation-based outer algorithm for being placed in the assortment.\footnote{This 
dynamic programming is a generalization of
the dynamic programming for prophet inequality matching
in \citet{AHL-12}.} The main difference this time is that  we can actually compute and run the \emph{exact optimum} discarding DP. To see this, fix a product $i\in [n]$. Notice that in the case of infinite rental durations, the current inventory level plays the role of the state (which is a monotone decreasing quantity, as units of the product never return to the inventory); thus the DP has polynomial-size state space and in principle can be computed efficiently, i.e., in polynomial time. 

It turns out that solving this DP is intimately related to the ``magician problem'', introduced and studied in \citet{ala-14}, which itself is an special case of online contention resolution schemes for a simple uniform matroid environment~\citep{FSZ-16}. By leveraging the techniques introduced in this paper and extending them to the Bayesian online assortment optimization problem, we show that a discarding algorithm that follows the optimum DP loses no more than $\tfrac{1}{\sqrt{\mininventory + 3}}$ fraction of the expected LP's revenue (which is the same near-optimal approximation factor as in the magician problem).
}


Specifically, let $\DP_{i, t}^I$ denote the expected revenue-to-go of product $i$ during $[t,T]$ from the optimal discarding policy, when $I$	is the remaining inventory at time $t$.
Similar to the discussion about the per-unit revenue
thresholds in \Cref{sec:stochastic-finite-simple-DP},
suppose assortment $\sampleassortment \sim \curexantetype$ is sampled 
at time $t$. If the DP discarding policy 
discards the product $i$ from the sampled assortment $\sampleassortment$,
the expected revenue-to-go becomes $\DP_{i, t + 1}^I$;
otherwise, with probability $\choice^{\type_t}(\sampleassortment, i)$\footnote{Similiar to \Cref{sec:stochastic-finite-simple-DP},
here we consider a hypothetical scenario
where 
the 
probability that the consumer select product $i$ equals $\choice^{\type_t}(\sampleassortment, i)$
regardless whether another product $i'$ is discarded from $\sampleassortment$.
}
the consumer selects product $i$ and the expected revenue-to-go becomes
$\DP_{i,t + 1}^{I - 1} + \curreward$.
With the remaining probability, product $i$ is not selected and 
the expected revenue-to-go is $\DP_{i, t + 1}^I$.
Hence, we 
can calculate $\DP_{i, t}^I$ using dynamic programming (i.e., backward induction), with 
the following Bellman equation:
\begin{align}
	\label{eq:bellman equation infinite}
	\begin{split}
		&\DP_{i, t}^I = 
		\displaystyle\sum\nolimits_{\type_t \in \typespace_t}
		\displaystyle\sum\nolimits_{\assortment \in \assortmentspace}
		\curexantetype  
		\max\{\DP_{i, t + 1}^I,
			\curchoice
			(\curreward + 
			\DP_{i, t + 1}^{I - 1})
			+
			(1 - \curchoice)\DP_{i, t + 1}^I
		\},
	\end{split}
\end{align}
After rearranging the terms, it is easy 
to observe that DP decisions are made by an inventory-dependent revenue thresholding rule (that is
independent of $\{\curchoice\}$, as these terms cancel out):
\begin{align*}
	\left[\textrm{At time $t$ with type $\type_t$
			and inventory $I_i$, 
	$i\in \hat{\assortment}$ will be accepted}\right] 
	\Longleftrightarrow 
	\curreward    \geq 
	\DP_{i, t + 1}^{I_i}-\DP_{i, t + 1}^{I_i - 1}
\end{align*}
We are now ready to present our refined simulation-based policy with a new discarding rule. This rule uses the above inventory-dependent thresholds, where these thresholds are computed upfront, but they depend on the product, time, and the current remaining inventory level of the product. See \Cref{alg:SB infinite DP} for details.

\begin{algorithm}[htb]
	\caption{Simulation-based Algorithm with Inventory-dependent Revenue Thresholds}
	\label{alg:SB infinite DP}
 	\vspace{2mm}
 	
 	\emph{\underline{Pre-processing:}}
 	\begin{itemize}
 	    \item Compute the optimal assignment $\{\curalloctypeopt\}$ of $\EAR{\typedistributionsequence}$ by invoking the offline assortment oracle  (\Cref{asp:oracle})
 	    \item 	Set $\curexantetype\triangleq\curalloctypeopt\typedistribution_t(\type_t)$ for every $\assortment$,$t$ and $\type_t$ where optimal assignment has a non-zero entry
  		\item Solve the dynamic programming with Bellman update described in \Cref{eq:bellman equation infinite} 
			and boundary condition $\DP_{i,T+1}^I= \DP_{i, t}^0 = 0$ for every product $i$ and $I \in [c_i]$ to obtain 
		$\{\DP_{i,t}^I\}_{i\in[n],t\in[\totaltime]}^{I \in [\inventory_i]}$
		\item Let $Q_{i, t}^I \triangleq
		\DP_{i, t + 1}^{I}-\DP_{i, t + 1}^{I-1}
		$, for all $i\in[n],t\in[\totaltime], I\in [c_i]$.
 	\end{itemize}

 	\vspace{2mm}
 	
 	\For{$t=1$ to $T$}{
 	\tcc{consumer $t$ with type $\type_t\sim\typedistribution_t$ arrives} 
 	\vspace{1mm}
 	
 	\emph{\underline{Simulation:}} Upon realizing consumer type $\type_t$, sample $\hat{\assortment}_t\sim \{\curalloctypeopt\}_{\assortment\in\assortmentspace}$
 	
 	\vspace{2mm}
 	\emph{\underline{Discarding:}} Initialize $\bar\assortment_t\gets\hat\assortment_t$
 	
 	\For{each product $i\in\hat\assortment_t$}{
 	\vspace{1mm}
 	Let $I_{i,t}$ be the currently available units
 	of product $i$ for product $i$.
 	
 	\vspace{1mm}
 	\If {$\curreward
			<
		Q_{i, t}^{I_{i,t}}$ or there is no available unit of product $i$}{
 	Remove $i$ from $\bar\assortment_t$}}
 	\tcc{Inventory-dependent revenue thresholds $\{Q_{i, t}^{I}\}$ are computed once upfront}
 	\vspace{2mm}
 	\emph{\underline{Post-processing:}} Let $\tilde{S}_t\gets \textsc{Sub-assortment Sampling} \left(\choice^{\type_t}, \bar\assortment_t,\{\choice(\hat\assortment_t, i)\}_{i\in \bar\assortment_t}\right)$
 	
 	\tcc{Send a query call to Procedure~\ref{alg:sample assortment} with appropriate input arguments}
 	\vspace{2mm}
 	
 	Offer assortment $\tilde{\assortment}_t$ to consumer $t$
 	
 	}
\end{algorithm}

\begin{restatable}{theorem}{infiniteDP}
	\label{thm:SB infinite DP}
	Let $\mininventory = \min_{i\in[n]}\inventory_i$ be the smallest inventory.
	The competitive ratio of \Cref{alg:SB infinite DP}
	against offline Bayesian expected LP benchmark, i.e., $\EAR{\typedistributionsequence}$
	is at least $\left(1-\tfrac{1}{\sqrt{\mininventory+3}}\right)$.
	Moreover, it runs in time
	$\textrm{Poly}(n,T, \sum_{t\in[\totaltime]}\lvert\typespace_t\rvert)$ given oracle access to an offline algorithm for assortment optimization (\Cref{asp:oracle}).
\end{restatable}

The proof follows a similar structure
as in proof of \Cref{thm:stochastic-finite-competitive-ratio}. See \Cref{apx:sand barrier} for details. 
\section{Omitted Proofs in Section~\ref{sec:prelim}}
\label{apx:LP}
\EARupperbound*
\begin{proof}{\emph{Proof.}}
Let $I_{\assortment,t}(\mathbf{\type})\in\{0,1\}$ be the indicator that clairvoyant optimum online benchmark offers set $\assortment$ at time $t$ for a given type sequence $\mathbf{\type}=\typesequence$. As clairvoyant optimum online benchmark is inventory feasible for any sample path of rental durations $\mathbf{\duration}=\{\duration_t\}_{t=1}^T$,  for all $i \in [n],\ t \in [\totaltime]$ we have:
\begin{equation*}
\displaystyle\sum\nolimits_{\tpre = 1}^{t}
	\displaystyle\sum\nolimits_{\assortment \in \assortmentspace}
	\indicator{\duration_\tpre \geq t - \tpre}
	\prechoice I_{\assortment,t}(\mathbf{\type}) \leq \inventory_i
\end{equation*}
Now, note that random variable $d_\tpre$ is independent from random variable $I_{\assortment,\tpre}(\mathbf{\type})$, as clairvoyant optimum online benchmark does not get to see $d_\tpre$ when assorts a set at time $\tpre$. Therefore,  setting  $\curalloctype=\texpect{I_{\assortment,t}(\mathbf{\type})\given \type_t}$
results in a feasible assignment in the corresponding linear program $\EAR{\typedistributionsequence}$ for any sequence of type distributions $\typedistributionsequence$; 
to see this, take expectation of the LHS of the above inequality over rental durations first, and then over $\mathbf{z}$ considering that types are independent. 
Moreover, the objective value of $\EAR{\typedistributionsequence}$ under this assignment will be equal to the expected revenue of the clairvoyant optimum online benchmark, which finishes the proof.
 \hfill\Halmos
\end{proof}
\EARpolytime*

Before proving the above proposition, we first state a technical lemma regarding the ellipsoid algorithm and sketch its proof
\citep[see][for a detailed proof]{GLS-1981}.
\begin{lemma}
\label{lem:ellipsoid-dual}
Suppose a primal linear program is bounded and feasible. Moreover, suppose its dual has a separation oracle with running time at most $\tau$. Then primal admits an optimal solution that has $\textrm{Poly}(\tau,m)$ non-zero entries and can be computed in $\textrm{Poly}(\tau,m)$ time, where $m$ is the number of dual variables (primal constraints). 
\end{lemma}
\begin{proof}{\emph{Proof sketch of \Cref{lem:ellipsoid-dual}.}}
Because primal LP is feasible and bounded, dual is also feasible and bounded. Now consider running ellipsoid algorithm with the separation oracle to solve the dual. Ellipsoid can solve the dual in $\textrm{Poly}(\tau,m)$ time. Now by looking at the run of ellipsoid and considering only dual constraints that ellipsoid returns as violated,  we can find $\textrm{Poly}(\tau,m)$ many constraints in the dual such that the dual value does not change if we throw away the remaining constraints. Note that constraints in the dual correspond to variables in the primal. So, we can throw away all the variables in the primal except the ones corresponding to the constraints in the dual that the separation oracles asks us to keep, and the LP value does not change. Now we have a compact LP in the primal, which we can again solve to get the succinct primal solution with $\textrm{Poly}(\tau,m)$ non-zero entries and in time $\textrm{Poly}(\tau,m)$.
\hfill\Halmos
\end{proof}
\begin{proof}{\emph{Proof of \Cref{prop:EAR running time}.}}
Given \Cref{lem:ellipsoid-dual}, we are ready to show that expected LP can be solved efficiently by the offline oracle. $\EAR{\typedistributionsequence}$ has exponentially many variables and $(nT+T)$ constraints. Hence its dual has only $(nT+T)$ variables and exponentially many constraints. Here is the dual LP, with dual variable $\inventorydual$ 
	and $\probdual$:
\begin{align*}
	\begin{array}{rlll}
		\min\limits_{\boldsymbol{\theta},\boldsymbol{\lambda} \geq \mathbf 0} &
		~~~~~~~~~~\displaystyle\sum\nolimits_{t=1}^\totaltime
		\displaystyle\sum\nolimits_{i=1}^n
		\inventory_i\inventorydual+ \displaystyle\sum\nolimits_{t =1}^\totaltime\sum\nolimits_{\type_t\in\typespace_t}\probdualtype &~~~ \text{s.t.} &\\[1em]
		&
		\probdualtype + \displaystyle\sum\nolimits_{i = 1}^n\displaystyle\sum\nolimits_{\tpre= t}^T\typedistribution_t(\type_t)\durationcdfi^{\type_t}_i(\tpre-t)\curchoice\theta_{i,\tpre}\geq \displaystyle\sum\nolimits_{i = 1}^n
		\curreward \typedistribution_t(\type_t)\curchoice
		&~~~\assortment \in \assortmentspace,\ t \in [\totaltime],\type_t\in\typespace_t
		&
	\end{array}
\end{align*}	
Now, fix $t\in[\totaltime], \type_t\in\typespace_t$, and consider the group of dual constraints corresponding to $t$ and $\type_t$. In order to obtain a separation oracle for this group of constraints, given $\probdualtype$ and $\{\theta_{i,\tpre}\}$ one needs to find a set $\hat{S}\in\underset{\assortment\in\assortmentspace}\argmax~{\displaystyle\sum\nolimits_{i=1}^n\hat{R}_i\choice^{\type_{t}}(\assortment,i)}$, where
\begin{equation*}
\hat{R}_i\triangleq\typedistribution_t(\type_t)\left(\curreward-\displaystyle\sum\nolimits_{\tpre= t}^T\durationcdfi^{\type_t}_i(\tpre-t)\theta_{i,\tpre}\right)
\end{equation*}
Then if $\probdualtype\geq \displaystyle\sum\nolimits_{i=1}^n\hat{R}_i\choice^{\type_{t}}(\hat\assortment,i)$ these constraints are satisfied and if not the constraint corresponding to $t,\type_t,\hat{\assortment}$ is a separating hyperplane.  Note that because of downward-closeness of $\assortmentspace$, i.e., \Cref{asp:down-closed}, and weak-substitutability of $\curchoice$, i.e., \Cref{asp:substitutability}, we can discard products for which $\hat{R}_i<0$, and hence without loss of generality assume $\hat{R}_i\geq 0$. Therefore, thanks to \Cref{asp:oracle}, one can find such a subset $\hat \assortment$ with one call to the oracle. By a polynomial time search over all possible $t$ and $\type_t$, dual will have a separation oracle with running time $\textrm{Poly}(n,T,\sum_{t\in[\totaltime}\lvert\typespace_t\rvert)$. Invoking \Cref{lem:ellipsoid-dual} finishes the proof.
\hfill\Halmos
\end{proof}
\section{Omitted Technical Details of \Cref{sec:large inventory} (Large Inventory)}
\label{apx:alg 1 proof}
\revcolor{
In this section, we provide the detailed analysis 
of the competitive ratio of \Cref{alg:large inventory},
stated in \Cref{thm:competitive-ratio large inventory}.

In \Cref{sec:large inventory}, 
we highlighted the one major issue due to the possibly positive 
correlation among allocations of the same product across time.
To work around this issue,
we use a careful coupling argument in our analysis of \Cref{alg:large inventory} that couples the rental indicator random variables of our algorithm with an alternative hypothetical algorithm. This hypothetical algorithm \emph{ignores} inventory constraints of all the products and only simulates the expected LP's optimal solution combined with independent discarding of each product with probability $\gamma$. This algorithm generates an independent sequence of rental indicator random variables, allowing us to use simple concentration bounds. Importantly, this coupling trick is only possible because of the guarantee of the sub-assortment sampling procedure in \Cref{eq:subassortment}, which will be clear later in the proof.

Before the proof of \Cref{thm:competitive-ratio large inventory}, we recall the multiplicative form of Chernoff concentration bound~\citep{C-1952} that is used in the proof. 

\begin{lemma}[Multiplicative Chernoff Bound]
\label{lem:chernoff}
Suppose $X_1,X_2,\ldots,X_t$ are independent random variables taking values in $\{0,1\}$. Let $X$ denote their sum and let $\mu=\expect[]{X}$ denote the sum's expected value. Then for any $\delta>0$,
$$
\prob{X> (1+\delta)\mu}\leq \exp\left(-\frac{\delta^2\mu}{2+\delta}\right)
$$
\end{lemma}

\begin{proof}{\emph{Proof of \Cref{thm:competitive-ratio large inventory}.}}
Running time is easily proved by \Cref{prop:EAR running time} and \Cref{prop:assortment polytope}. For the proof of competitive ratio,  we first claim that for any time period $t$ and product $i$,
the probability that there exists an available unit 
of product $i$ is at least $1 - \exp\left(-\frac{\gamma^2\mininventory}
    {2-\gamma}\right)$. To show this claim, define the following indicator random variables and their corresponding events:
\begin{align*}
    \indica_{i, t}: &\text{ event that a unit of product $i$ is allocated 
    at time $t$}~;
    \\
    \indica_{i, t}^{(1)}: &\text{ event that product $i$ is in assortment $\hat\assortment_t$}~; 
    \\
    \indica_{i, t}^{(2)}: &\text{ event that product $i$ is not removed from assortment $\bar\assortment_t$ in line 6
    of \Cref{alg:large inventory}}~; 
    \\
     \indica_{i, t}^{(3)}: &\text{ event that a unit of product $i$ is available in the inventory at the beginning of time $t$}~; 
    \\
    \indica_{i, t}^{(4)}: &\text{ event that product $i$ is 
    in assortment $\tilde\assortment_t$ 
    and selected by consumer $t$}~; \\
    \indica_{i, \tpre, t}: &\text{ event that 
    rental duration of product $i$ at time $\tpre$ is at least 
    $t - \tpre$~.}
\end{align*}
Notice that by definition, 
$\indica_{i, t} = \indica_{i, t}^{(1)} \cdot 
\indica_{i, t}^{(2)} \cdot 
\indica_{i, t}^{(3)} \cdot 
\indica_{i, t}^{(4)}$. Now, our claim is equivalent to
\begin{align}
\label{eq:small infeasible probability}
    \prob{\sum_{\tpre < t} \indica_{i,\tpre}\cdot \indica_{i,\tpre,t} \geq \inventory_i} \leq 
    \exp\left(-\frac{\gamma^2\mininventory}
    {2-\gamma}\right)~,
\end{align}
for any time $t$ and 
product $i$. To show inequality~\eqref{eq:small infeasible probability},
one concern is that $\{\indica_{i,t}\}$ are not independent across $t$.
To resolve this issue, 
consider a hypothetical scenario where 
we run \Cref{alg:large inventory}
without inventory constraints, 
and define random variables 
$\left\{\type_t^\dagger, \hat\assortment_t^\dagger,
\indica_{i, t}^\dagger, 
\indica_{i, t}^{(1)\dagger}, \indica_{i, t}^{(2)\dagger}, 
\indica_{i, t}^{(3)\dagger}, \indica_{i, t}^{(4)\dagger}, \indica_{i,\tpre,t}^\dagger\right\}_{i,t,\tau}$ 
exactly in the same way as in a normal run of \Cref{alg:large inventory} with inventory constraints.  
Note that by definition,
(i) $\indica_{i, t}^{(3)\dagger}$ is deterministically equal to  1;
(ii)
$\indica^\dagger_{i, t} = \indica_{i, t}^{(1)\dagger} \cdot 
\indica_{i, t}^{(2)\dagger} \cdot 
\indica_{i, t}^{(3)\dagger}
\cdot 
\indica_{i, t}^{(4)\dagger}
=\indica_{i, t}^{(1)\dagger} \cdot 
\indica_{i, t}^{(2)\dagger} \cdot 
\indica_{i, t}^{(4)\dagger}
$;
and
(iii)
$\{\indica^\dagger_{i,\tpre}\cdot\indica_{i,\tpre,t}^\dagger\}_\tpre$ are mutually independent across $\tpre$.

We now use a coupling between this hypothetical scenario and the normal run of \Cref{alg:large inventory} with inventory constraints, in which $\sum_{\tpre < t} \indica_{i,\tpre}^\dagger\cdot \indica_{i,\tpre,t}^\dagger\geq \sum_{\tpre < t} \indica_{i,\tpre}\cdot \indica_{i,\tpre,t}$ for all $i,t$. Clearly, we can define the coupling such that
$\type_t^\dagger \gets\type_t$, $\hat\assortment_t^\dagger \gets \hat\assortment_t$ (and therefore
$\indica_{i, t}^{(1)\dagger} = \indica_{i, t}^{(1)}$),
$\indica_{i, t}^{(2)\dagger} \gets \indica_{i, t}^{(2)}$,
and 
$\indica_{i,\tpre,t}^\dagger \gets
\indica_{i,\tpre,t}$.
Additionally, 
notice that 
\begin{align*}
\prob{\indica_{i, t}^{(4)}
\cdot \indica_{i, t}^{(3)} = 1 
\condition 
\indica_{i, t}^{(1)}, \indica_{i, t}^{(2)},\indica_{i, t}^{(3)}}
&{=}~
\choice^{\type_t}(\hat\assortment_t, i)
\cdot 
\indica_{i, t}^{(1)}\cdot \indica_{i, t}^{(2)}
\cdot\indica_{i, t}^{(3)}{=}
\choice^{\type_t^\dagger}(\hat\assortment_t^\dagger, i)
\cdot 
\indica_{i, t}^{(1)\dagger}\cdot \indica_{i, t}^{(2)\dagger}\cdot \indica_{i, t}^{(3)}
\\
&{=}~
\prob{\indica_{i, t}^{(4)\dagger}
\cdot \indica_{i, t}^{(3)} = 1 
\condition 
\indica_{i, t}^{(1)\dagger}, \indica_{i, t}^{(2)\dagger}, \indica_{i, t}^{(3)}}~,
\end{align*}
where the first equality holds due to the guarantee of the sub-assortment sampling in 
\Cref{prop:assortment polytope} and the second equality holds because of the coupling above. Also, the third equality holds automatically when $\indica_{i, t}^{(3)}=0$, and when $\indica_{i, t}^{(3)}=1$ it holds because of the guarantee of sub-assortment sampling in \Cref{prop:assortment polytope} and the fact that $\indica_{i, t}^{(3)}=\indica_{i, t}^{(3)\dagger}$. Again -- because of the coupling above -- we conclude that: 
$$
\prob{\indica_{i, t}^{(4)}
\cdot \indica_{i, t}^{(3)} = 1 
\condition 
\indica_{i, t}^{(1)}, \indica_{i, t}^{(2)}}=\prob{\indica_{i, t}^{(4)\dagger}
\cdot \indica_{i, t}^{(3)} = 1 
\condition 
\indica_{i, t}^{(1)\dagger}, \indica_{i, t}^{(2)\dagger}}~.
$$
Hence, we can further couple random variables
$\indica_{i, t}^{(4)\dagger}
\cdot \indica_{i, t}^{(3)}\gets
\indica_{i, t}^{(4)}
\cdot \indica_{i, t}^{(3)}$. Therefore, $\indica_{i, t}^{(4)\dagger}
\cdot \indica_{i, t}^{(3)\dagger}\geq \indica_{i, t}^{(4)\dagger}
\cdot \indica_{i, t}^{(3)}= \indica_{i, t}^{(4)}
\cdot \indica_{i, t}^{(3)}$,
which guarantees that 
%
$\indica^\dagger_{i, t}\cdot \indica_{i,\tpre,t}^\dagger
\geq \indica_{i, t}
\cdot \indica_{i,\tpre,t}$ 
for all $i, t, \tau$ in this coupling.
Thus, it is sufficient to show that
\begin{align*}
    \prob{\sum_{\tpre < t} \indica_{i,\tpre}^\dagger\cdot \indica_{i,\tpre,t}^\dagger \geq \inventory_i} \leq 
    \exp\left(-\frac{\gamma^2\mininventory}
    {2-\gamma}\right)~.
\end{align*}

To show the above bound, we first rewrite the expectation $\expect{\indica_{i,\tpre}^\dagger\cdot \indica_{i,\tpre,t}^\dagger}$ as follows,
\begin{align*}
    \expect{\indica_{i,\tpre}^\dagger\cdot \indica_{i,\tpre,t}^\dagger} &= 
    \expect{\indica_{i, \tpre}^{(1)\dagger} \cdot 
\indica_{i, \tpre}^{(2)\dagger} \cdot 
\indica_{i, \tpre}^{(4)\dagger}\cdot \indica_{i,\tpre,t}^\dagger} \\
&=
\expect{\indica_{i, \tpre}^{(1)\dagger} \cdot 
\indica_{i, \tpre}^{(2)\dagger} \cdot 
\expect{\indica_{i, \tpre}^{(4)\dagger}\cdot\indica_{i,\tpre,t}^\dagger\condition
\indica_{i, \tpre}^{(1)\dagger}, \indica_{i, \tpre}^{(2)\dagger}
}
} \\
&= 
\expect{\indica_{i, \tpre}^{(1)\dagger} \cdot 
\indica_{i, \tpre}^{(2)\dagger} \cdot 
\left(\choice^{\type_\tpre^\dagger}(\hat\assortment_\tpre^\dagger, i)\cdot\indica_{i, \tpre}^{(1)\dagger}\cdot \indica_{i, \tpre}^{(2)\dagger} 
\cdot \indica_{i,\tpre,t}^\dagger\right)
} \\
&=
\expect{\indica_{i, \tpre}^{(1)\dagger} \cdot 
\indica_{i, \tpre}^{(2)\dagger} \cdot 
\choice^{\type_\tpre^\dagger}(\hat\assortment_\tpre^\dagger, i)\cdot \indica_{i,\tpre,t}^\dagger}
=
(1-\gamma)
	\displaystyle\sum
	_{\type_\tpre\in\typespace_\tpre}
	\displaystyle\sum
	_{\assortment \in \assortmentspace}
	\typedistribution_\tpre(\type_\tpre)
    \durationcdfi^{\type_{\tpre}}_i(t-\tpre)
	\curchoice\alloc_{\assortment, t, \type_t}^*~,
\end{align*}
where we use the facts that 
$\expect{\indica_{i, t}^{(4)\dagger}\condition
\indica_{i, t}^{(1)\dagger}, \indica_{i, t}^{(2)\dagger}
} 
= \choice^{\type_t^\dagger}(\hat\assortment_t^\dagger, i)\cdot\indica_{i, t}^{(1)\dagger}\cdot \indica_{i, t}^{(2)\dagger}
$ 
by
\Cref{prop:assortment polytope};
$\expect{\indica_{i, t}^{(2)\dagger}\condition\indica_{i, t}^{(1)\dagger}}
=
(1-\gamma)\,
\indica_{i, t}^{(1)\dagger}$ by construction; and random variables $\indica_{i,\tpre}^\dagger$ and $\indica_{i,\tpre,t}^\dagger$ are independent conditioned on $\type_\tpre$. Thus, 
\begin{align*}
\expect{\sum_{\tpre < t} \indica_{i,\tpre}^\dagger\cdot \indica_{i,\tpre,t}^\dagger}
&=
(1-\gamma)\displaystyle\sum
_{\tpre < t}
	\displaystyle\sum
	_{\type_\tpre\in\typespace_\tpre}
	\displaystyle\sum
	_{\assortment \in \assortmentspace}
	\typedistribution_\tpre(\type_\tpre)
    \durationcdfi^{\type_{\tpre}}_i(t-\tpre)
	\curchoice\alloc_{\assortment, t, \type_t}^* \\
	&\leq (1-\gamma)\,\inventory_i
\end{align*}
We finish the proof of our claim -- i.e., tail bound in~\eqref{eq:small infeasible probability}-- by applying the multiplicative form of Chernoff bound (\Cref{lem:chernoff}) for the sequence of independent random variables $\left\{\indica_{i,\tpre}^\dagger\cdot \indica_{i,\tpre,t}^\dagger\right\}_{\tpre=1}^{t-1}$.

Now, fix any time $t$ and product $i$. Consider the 
expected contribution of product $i$ at time period $t$ in the revenue of \Cref{alg:large inventory}. By definition, it is equal to
\begin{align*}
    \expect{\reward_i^{\type_t}\,\indica_{i, t}} &=
    \expect{\reward_i^{\type_t}\,\indica_{i, t}^{(1)} \cdot 
\indica_{i, t}^{(2)} \cdot 
\indica_{i, t}^{(3)} \cdot 
\indica_{i, t}^{(4)}} 
\\
&=\expect{\reward_i^{\type_t}\,\indica_{i, t}^{(1)} \cdot 
\indica_{i, t}^{(2)} \cdot \indica_{i, t}^{(3)} \cdot
\expect{\indica_{i, t}^{(4)}\condition
\indica_{i, t}^{(1)}, \indica_{i, t}^{(2)},\indica_{i, t}^{(3)}
}
} \\
&= \expect{\reward_i^{\type_t}\,\indica_{i, t}^{(1)} \cdot 
\indica_{i, t}^{(2)} \cdot 
\indica_{i, t}^{(3)} \cdot 
\left(
\choice^{\type_t}(\hat\assortment_t, i)
\cdot
\indica_{i, t}^{(1)}
\cdot
\indica_{i, t}^{(2)}
\cdot
\indica_{i, t}^{(3)}
\right)} \\
&\geq 
(1-\gamma)\left(1 - \exp\left(-\frac{\gamma^2\mininventory}
    {2-\gamma}\right)\right)
	\displaystyle\sum
	_{\type_t\in\typespace_t}
	\displaystyle\sum
	_{\assortment \in \assortmentspace}\typedistribution_t(\type_t)
	\curreward\curchoice\curalloctype^*
\end{align*}
where 
the third equality holds by \Cref{prop:assortment polytope}; 
and 
the last inequality holds since $\indica_{i, t}^{(3)}$ is mutually independent from
$\reward_i^{\type_t}$, $\indica_{i, t}^{(1)}$,  
$\indica_{i, t}^{(2)}$ and $\hat\assortment_t$,
and the fact that by our tail bound in \eqref{eq:small infeasible probability} we have 
$$\expect{\indica_{i, t}^{(3)}}\geq 1 - \exp\left(-\frac{\gamma^2\mininventory}
    {2-\gamma}\right)$$
Thus, \Cref{alg:large inventory} is at least 
$\left.(1-\gamma)\left(1 - \exp\left(-\frac{\gamma^2\mininventory}
    {2-\gamma}\right)\right)\right.$-competitive
    against the Bayesian expected LP benchmark for any $\gamma\in[0,1]$.
Finally, setting $\gamma =\gammaopt$
finishes the proof.
\hfill\halmos
\end{proof}
}
\section{Omitted Technical Details of \Cref{sec:small inventory} (Small Inventory)}
\label{apx:alg 2 proof}
In this section, we provide the detailed analysis of the competitive ratio of \Cref{alg:SB}, stated in \Cref{thm:stochastic-finite-competitive-ratio}. The analysis consists of two major parts, as mentioned in the sketch of the proof of  \Cref{thm:stochastic-finite-competitive-ratio}. We first compare \Cref{alg:SB} with the simple dynamic programming described in \Cref{sec:stochastic-finite-simple-DP} and show the total expected revenue of \Cref{alg:SB} due to rentals of product $i$ is at least $\inventory_i\valueDPsimple_{i,1}$ 
(\texttt{Part~(\rom{1}}), Appendix~\ref{apx:part1}). 
We then compare this simple dynamic programming with the expected linear programming benchmark and show for each product $i$, $\inventory_i\valueDPsimple_{i,1}$ is at least $1/2$ of the contribution of product $i$ to the optimal objective value of $\EAR{\typedistributionsequence}$ 
(\texttt{Part~(\rom{2}}), Appendix~\ref{apx:part2}). Combining the two parts finishes the proof of \Cref{thm:stochastic-finite-competitive-ratio}.
\subsection{{Part~(\rom{1}}) of the Proof: Comparing with the Optimistic DP} 
\label{apx:part1}
In this part, we use vector variable $\mathbf{J}=(J_{1},\ldots,J_{T})$ to track the state of inventory of the fixed product $i$, where $J_{t}$ is the number of copies of product $i$ that will return to inventory at time $t$. Note that at each time $t$, a \emph{possible state} $\mathbf{J}$ has the form
\begin{equation*}
\mathbf{J}=(0,\ldots,0,J_{t},J_{t+1},\ldots,J_T)~,
\end{equation*}
where $\sum_{\tau=t}^{T}J_{t}=\inventory_i$. Now suppose $\algValue_{i,t}(\mathbf{J})$ denotes the per-unit revenue of product $i$ generated by \Cref{alg:SB} from time $t$ to $T$, when algorithm starts from initial state $\mathbf{J}$ at time $t$.\footnote{Equivalently, $\algValue_{i,t}(\mathbf{J})$ is equal to the total revenue of product $i$ generated in $[t:T]$ divided by $J_t$,  which is the number of on-hand units of product $i$ at the starting time $t$.} As a convention, let $\algValue_{i,T+1}(\mathbf{J})=0$. We prove the following stronger lemma, which also shows total revenue generated from product $i$ is no less than $\inventory_i\valueDPsimple_{i,1}$.
The intuition behind \Cref{lem:part-one-simple}
is as follows:
$\valueDPsimple_{i,t}$
is computed by an \emph{optimistic DP}
that 
``imagines'' the deficiency in inventory
is replenished every period,
and 
the expected revenue-to-go is a concave 
function of inventory level
--
i.e., 
the higher the inventory level, 
the lower the per-unit expected 
revenue-to-go.
\begin{restatable}{lemma}{optimisticDP}
\label{lem:part-one-simple}
For every $t\in[\totaltime]$ and every possible inventory state $\mathbf{J}$ at time $t$, $\algValue_{i,t}(\mathbf{J})\geq \valueDPsimple_{i,t}$.
\end{restatable}

\begin{proof}{\emph{Proof.}}
{To simplify the proof, we first assume that the algorithm does not perform the sub-assortment sampling as the post-processing step (and hence $\tilde{\assortment}_t$ is the same as $\bar{\assortment}_t$). We finish the proof first under this simplifying assumption. Then, we show why the proof still remains intact when we replace the assortment set $\bar{\assortment}_t$ in the proof with $\tilde{\assortment}_t\subseteq \bar{\assortment}_t$, which is the relevant case when use the sub-assortment sampling procedure in the post-processing step of \Cref{alg:SB}.}

The proof is based on backward induction on $t$. For the base of induction, $\algValue_{i,T+1}(\mathbf{J})=\valueDPsimple_{i,T+1}=0$. Now suppose for any time in $[t+1:T]$ and possible state $\mathbf{J}$ the statement of lemma holds. We can simply write the following update equation for $\algValue_{i,t}(\mathbf{J})$ (following exactly the same logic as the update rule of the dynamic programming in \cref{eq:bellman-equation-finite}, but considering that only $J_t$ products are on-hand):
\begin{align}
		\label{eq:bellman-alg}
		\begin{split}
			&\algValue_{i,t}(\mathbf{J})= \displaystyle\sum\nolimits_{\type_t \in \typespace_t}
		\displaystyle\sum\nolimits_{\assortment \in \assortmentspace}
		\curexantetype \Bigg(\indicator{\curreward<\curprice}\algValue_{i, t + 1}(\mathbf{J'})\\
		&+\indicator{\curreward\geq\curprice}
		\Bigg(\left(1 - 
		\choice^{\type_t}\left(
		\bar \assortment, i
		\right)
		\right)\algValue_{i, t + 1}(\mathbf{J'}) \\
		&
		+ \frac{\choice^{\type_t}\left(
		\bar \assortment, i
		\right)}{J_t}\sum\nolimits_d\durationpdf_i^{\type_t}(d)\left(\curreward+\expect{\algValue_{i, t + d}(\mathbf{\tilde{J}}^{(d,d)})+(J_t-1)\algValue_{i, t + 1}(\mathbf{\tilde{J}}^{(1,d)})}\right)\Bigg)\Bigg)
	\end{split}
	\end{align}
where $\bar\assortment\subseteq \assortment$ is the assortment 
offered in \Cref{alg:SB} for consumer $t$ after the discarding step,
$\mathbf{J'}\triangleq(0,\ldots,0,0,J_t+J_{t+1},J_{t+2},\ldots,J_T)$ is the next state at time $t+1$ if no rental happens at time $t$,  and $\mathbf{\tilde{J}}^{(\tau,d)}$ is the (randomized) state of \Cref{alg:SB} at time $t+\tau$ if a rental happens at time $t$ with rental duration $d$ (where the randomness in this state comes from all type realizations and consumer choices during future times $[t+1:t+d-1]$). Now, By applying our inductive hypotheses to \cref{eq:bellman-alg} and the weak-substitution (i.e., $\choice(S, i) \leq \choice(\bar\assortment, i)$ for all $i\in \bar\assortment$), we have: 
\begin{align}
\label{eq:bellman-alg2}
		\begin{split}
			&\algValue_{i,t}(\mathbf{J})\geq \displaystyle\sum\nolimits_{\type_t \in \typespace_t}
		\displaystyle\sum\nolimits_{\assortment \in \assortmentspace}
		\curexantetype \Bigg(\indicator{\curreward<\curprice}\valueDPsimple_{i,t+1}\\
& 
		+\indicator{\curreward\geq\curprice}\left(\left(1 - 
		\choice^{\type_t}\left(
		\bar \assortment, i
		\right)
		\right)\valueDPsimple_{i,t+1} + \frac{\choice^{\type_t}\left(
		\bar \assortment, i
		\right)}{J_t}\sum\nolimits_d\durationpdf_i^{\type_t}(d)\left(\curreward+\valueDPsimple_{i,t+d}+(J_t-1)\valueDPsimple_{i,t+1}\right)\right)\Bigg)\\
		&\geq \displaystyle\sum\nolimits_{\type_t \in \typespace_t}
		\displaystyle\sum\nolimits_{\assortment \in \assortmentspace}
		\curexantetype \Bigg(\indicator{\curreward<\curprice}\valueDPsimple_{i,t+1}\\
& 
		+\indicator{\curreward\geq\curprice}\left(\left(1 - 
		\choice^{\type_t}\left(
		 \assortment, i
		\right)
		\right)\valueDPsimple_{i,t+1} + \frac{\choice^{\type_t}\left(
		 \assortment, i
		\right)}{J_t}\sum\nolimits_d\durationpdf_i^{\type_t}(d)\left(\curreward+\valueDPsimple_{i,t+d}+(J_t-1)\valueDPsimple_{i,t+1}\right)\right)\Bigg)
	\end{split}
	\end{align}
Note also that the RHS of above inequality is non-increasing as a function of $J_t$, simply because if $\frac{1}{J_t}$  increases by additive $\epsilon>0$ then RHS increases by 
\begin{align*}
&\epsilon\cdot\indicator{\curreward\geq\curprice}\cdot \curchoice\cdot\left(\sum\nolimits_d\durationpdf_i^{\type_t}(d)(\curreward+\valueDPsimple_{i,t+d}-\valueDPsimple_{i,t+1})\right)\\
=&\epsilon\cdot\indicator{\curreward\geq\curprice}\cdot \curchoice\cdot\left(\curreward-\curprice\right)\geq 0
\end{align*}
Therefore, as $J_t\leq \inventory_i$, we have
\begin{align*}
		\begin{split}
			\algValue_{i,t}(\mathbf{J})&\geq \displaystyle\sum\nolimits_{\type_t \in \typespace_t}
		\displaystyle\sum\nolimits_{\assortment \in \assortmentspace}
		\curexantetype \Bigg(\indicator{\curreward<\curprice}\valueDPsimple_{i,t+1}+\\
& 
		+\indicator{\curreward\geq\curprice}\left(\left(1 - \curchoice\right)\valueDPsimple_{i,t+1} + \frac{\curchoice}{\inventory_i}\sum\nolimits_d\durationpdf_i^{\type_t}(d)\left(\curreward+\valueDPsimple_{i,t+d}+(\inventory_i-1)\valueDPsimple_{i,t+1}\right)\right)\Bigg)\\
		&=\displaystyle\sum\nolimits_{\type_t \in \typespace_t}
		\displaystyle\sum\nolimits_{\assortment \in \assortmentspace}
		\curexantetype \\
& 
		\times\max\left\{\valueDPsimple_{i, t + 1}~,~\left(1 - \curchoice\right)\valueDPsimple_{i, t + 1}+ \curchoice\left(\frac{1}{\inventory_i}\sum\nolimits_d\durationpdf_i^{\type_t}(d)\left(\curreward+\valueDPsimple_{i, t + d}\right)+\frac{\inventory_i-1}{\inventory_i}\valueDPsimple_{i, t + 1}\right )\right\}\\
		&=\valueDPsimple_{i, t }\qquad\qquad\qquad\qquad\qquad\qquad \textrm{(holds because of Bellman update in \cref{eq:bellman-equation-finite})}
	\end{split}
	\end{align*}
	{
Now suppose we use the sub-assortment sampling (Procedure~\ref{alg:sample assortment}) in the post-processing step. Performing this post-processing step does not decrease
the revenue guarantee
of discarding with ``per-unit revenue thresholds'' $\curprice$
in \Cref{alg:SB}. In fact, showing any subset of $\bar{S}$ (for example, the output $\tilde{S}\subseteq \bar{S}$ of the sub-assortment sampling procedure at each time) can only increase the per-unit revenue-to-go of the policy that follows the thresholds of the optimistic DP. To be more precise, look at the inductive step in the above proof at some time $t$. When $\curreward\geq\curprice$, the right-hand-side of the first inequality in \Cref{eq:bellman-alg2} can only \emph{increase} if we replace $\bar{S}$ with $\tilde{S}\subseteq \bar{S}$, simply because $\choice^{\type_t}(\bar{S},i)<\choice^{\type_t}(\tilde{S},i)$ due to weak-substitution and the coefficient of $\choice^{\type_t}(\bar{S},i)$ in the right-hand-side of the first inequality in \Cref{eq:bellman-alg2} is non-negative when $\curreward\geq\curprice$. The rest of the induction follows similar to the simplified case, which finishes the final proof.
\hfill\Halmos}
\end{proof}

\subsection{{Part~(\rom{2}}): Comparing the Optimistic DP with the Expected-LP}
\label{apx:part2}
In order to prove this part, we use the connection between DP of \Cref{sec:stochastic-finite-simple-DP} and a related LP. This connections leads us to apply simple duality arguments to find a lower-bound on the ratio of $\inventory_i\valueDPsimple_{i,1}$ and the contribution of product $i$ to the optimal objective value of $\EAR{\typedistributionsequence}$.

Note that scaling $\{\curreward\}$ does not change the ratio. Hence, 
without loss of generality, we normalize
the contribution of product $i$ to the optimal objective value of $\EAR{\typedistributionsequence}$,
i.e.,
\begin{equation*}
\displaystyle\sum\nolimits_{t=1}^{\totaltime}
			\displaystyle\sum\nolimits_{\type_t \in \typespace_t}
			\displaystyle\sum\nolimits_{\assortment \in \assortmentspace}\curexantetype\curchoice
			\curreward
			 =1
\end{equation*}
Given the above observation and the Bellman update in \eqref{eq:bellman-equation-finite}, consider the following primal linear program that gives a lower-bound on worst-case value of the ratio:
\begin{align}
		\begin{array}{llll}
			{\min
			\limits_{\mathbf \reward, \boldsymbol\valueDPsimple \geq \mathbf 0}}~~&
				\inventory_i\valueDPsimple_{i,1}~~~\text{s.t.}& \\[1em]
							&
			\valueDPsimple_{i,t}\geq	\displaystyle\sum\nolimits_{\type_t \in \typespace_t}
		\displaystyle\sum\nolimits_{\assortment \in \assortmentspace}
		\curexantetype~~~~~~~~~~~~~~~~~~~~~~~~~~\forall t \\
& 
		\times\max\left\{\valueDPsimple_{i, t + 1}~,~\left(1 - \curchoice\right)\valueDPsimple_{i, t + 1}+ \curchoice\left(\frac{1}{\inventory_i}\sum\nolimits_d\durationpdf_i^{\type_t}(d)\left(\curreward+\valueDPsimple_{i, t + d}\right)+\frac{\inventory_i-1}{\inventory_i}\valueDPsimple_{i, t + 1}\right )\right\}
			\\[1em]
			&
			\displaystyle\sum\nolimits_{t=1}^{\totaltime}
			\displaystyle\sum\nolimits_{\type_t \in \typespace_t}
			\displaystyle\sum\nolimits_{\assortment \in \assortmentspace}
			\curexantetype\curchoice
			\curreward
			\geq1& &
		\end{array}
	\end{align}
in which we allow both $\boldsymbol{\valueDPsimple}=\{\valueDPsimple_{i,t}\}_{t=1}^T$ and $\mathbf{\reward}=\{\curreward\}_{t\in[\totaltime],\type_t\in\typespace_t}$ to be variables (and this is an important feature of this technique). We first relax and simplify this program by switching the outer summation and max operator in the RHS of the first constraint (which only makes RHS smaller, hence a relaxation). Note also that:
\begin{align*}
&\displaystyle\sum\nolimits_{\type_t \in \typespace_t}
		\displaystyle\sum\nolimits_{\assortment \in \assortmentspace}\curexantetype\valueDPsimple_{i, t + 1}=\valueDPsimple_{i, t + 1}\triangleq A,\\
&\textrm{and}\\
	&\displaystyle\sum\nolimits_{\type_t \in \typespace_t}
		\displaystyle\sum\nolimits_{\assortment \in \assortmentspace}\curexantetype\left(\left(1 - \curchoice\right)\valueDPsimple_{i, t + 1}+ \curchoice\left(\frac{1}{\inventory_i}\sum\nolimits_d\durationpdf_i^{\type_t}(d)\left(\curreward+\valueDPsimple_{i, t + d}\right)+\frac{\inventory_i-1}{\inventory_i}\valueDPsimple_{i, t + 1}\right )\right)	\\
		&=\frac{1}{\inventory_i}\displaystyle\sum\nolimits_{\type_t \in \typespace_t}
		\displaystyle\sum\nolimits_{\assortment \in \assortmentspace}\curexantetype\curchoice\left(\sum\nolimits_d\durationpdf_i^{\type_t}(d)\left(\curreward+\valueDPsimple_{i, t + d}\right)-\valueDPsimple_{i, t + 1}\right)+\valueDPsimple_{i, t + 1}\triangleq B
\end{align*}
Now, by replacing the new constraint $\valueDPsimple_{i, t}\geq \max\{A,B\}$ with two constraints $\valueDPsimple_{i, t}\geq A$ and $\valueDPsimple_{i, t}\geq B$,  we get the following final primal linear programming:
\begin{align}
\label{eq:stochastic-finite-primal}\tag{\textbf{Primal-LP1}}
		\begin{array}{llll}
			{\min
			\limits_{\mathbf \reward, \boldsymbol\valueDPsimple \geq \mathbf 0}}~~&
				\inventory_i\valueDPsimple_{i,1}~~~~~\text{s.t.}& & \\[1em]
			&\valueDPsimple_{i,t}\geq	\frac{1}{\inventory_i}\displaystyle\sum\nolimits_{\type_t \in \typespace_t,\assortment \in \assortmentspace}
		\curexantetype\curchoice\left(\sum\nolimits_d\durationpdf_i^{\type_t}(d)\left(\curreward+\valueDPsimple_{i, t + d}\right)-\valueDPsimple_{i, t + 1}\right)+\valueDPsimple_{i, t + 1} & t\in[\totaltime]&\langle\alpha_t\rangle\\[1em] 
		&\valueDPsimple_{i,t}\geq	\valueDPsimple_{i,t+1}&t\in[\totaltime] &\langle\beta_t\rangle\\[1em]
			&\displaystyle\sum\nolimits_{t=1}^{\totaltime}
			\displaystyle\sum\nolimits_{\type_t \in \typespace_t}
			\displaystyle\sum\nolimits_{\assortment \in \assortmentspace}
			\curexantetype\curchoice
			\curreward
			\geq1& &\langle\mu \rangle
		\end{array}
	\end{align}
 Now, we write its dual program with dual variable $\{\alpha_t\}$, $\{\beta_t\}$ and $\mu$ as follows:

\begin{align}
\label{eq:stochastic-finite-dual}\tag{\textbf{Dual-LP1}}
		\begin{array}{llll}
			&{\max\limits_{\boldsymbol\alpha, \boldsymbol\beta, \mu\geq \mathbf 0}}~~
				\mu~~~~~\text{s.t.}& &\\[1em]
				&\mu\left(\displaystyle\sum\nolimits_{t=1}^{\totaltime}
			\displaystyle\sum\nolimits_{\assortment \in \assortmentspace}
			\curchoice
			\curexantetype\right)\leq \frac{1}{\inventory_i}\left(\displaystyle\sum\nolimits_{t=1}^{\totaltime}
			\displaystyle\sum\nolimits_{\assortment \in \assortmentspace}
			\curchoice
			\curexantetype\alpha_t\right)& \type_t\in\typespace_t &\langle\curreward \rangle\\[1em]
			&\alpha_t-\alpha_{t-1}+\beta_t-\beta_{t-1}\leq \frac{1}{\inventory_i} \displaystyle\sum\nolimits_{\type_{t-1} \in \typespace_{t-1},\assortment \in \assortmentspace}X_{S,t-1,\type_{t-1}}\choice^{\type_{t-1}}(S,i)\alpha_{t-1} & &\\
			&\qquad\qquad\qquad\qquad\qquad -\frac{1}{\inventory_i} \displaystyle\sum\nolimits_{d< t}   \displaystyle\sum\nolimits_{\type_{t-d} \in \typespace_{t-d},\assortment \in \assortmentspace}g_i^{\type_{t-d}}X_{S,t-d,\type_{t-d}}\choice^{\type_{t-d}}(S,i)\alpha_{t-d} & t\in[2:T] & \langle\valueDPsimple_{i,2:T} \rangle\\[1em]
			&\alpha_1+\beta_1\leq \inventory_i& & \langle\valueDPsimple_{i,1} \rangle
		\end{array}
\end{align}
Now, we try to guess a feasible solution for the dual program to obtain our desired lower-bound on the optimal primal objective. To this end, let $\forall t: \alpha_t=\mu\inventory_i$ (therefore, first set of constraints will be satisfied automatically) and let all other constraints to be tight. In particular, $\beta_1=\inventory_i-\alpha_1=\inventory_i(1-\mu)$ and for all $t\in [2:T]$:
\begin{align*}
\beta_t&=\beta_{t-1}+\frac{1}{\inventory_i} \displaystyle\sum\nolimits_{\type_{t-1} \in \typespace_{t-1},\assortment \in \assortmentspace}X_{S,t-1,\type_{t-1}}\choice^{\type_{t-1}}(S,i)\alpha_{t-1} &\\
			 &\qquad\quad-\frac{1}{\inventory_i} \displaystyle\sum\nolimits_{d< t}   \displaystyle\sum\nolimits_{\type_{t-d} \in \typespace_{t-d},\assortment \in \assortmentspace}g_i^{\type_{t-d}}X_{S,t-d,\type_{t-d}}\choice^{\type_{t-d}}(S,i)\alpha_{t-d} & \\
			&=\beta_{t-1}+\mu \displaystyle\sum\nolimits_{\type \in \typespace_{t-1},\assortment \in \assortmentspace}X_{S,t-1,\type}\choice^{\type}(S,i)&\\
			 &\qquad\quad-\mu\displaystyle\sum\nolimits_{d< t}   \displaystyle\sum\nolimits_{\type\in \typespace_{t-d},\assortment \in \assortmentspace}g_i^{\type}X_{S,t-d,\type}\choice^{\type}(S,i)& \\
			 &\overset{(1)}{=}\beta_1-\mu \displaystyle\sum\nolimits_{\tpre\leq t}\displaystyle\sum\nolimits_{\type \in \typespace_{t},\assortment \in \assortmentspace}X_{S,\tpre,\type}\choice^{\type}(S,i)\sum_{d\geq t-\tpre}\durationpdf_i^{\type}(d)&\\
			 &\overset{(2)}{\geq }\inventory_i(1-\mu)-\mu\inventory_i=\inventory_i(1-2\mu)~,
\end{align*}
where (1) can be obtained by changing the order of summations and rearranging the terms, and (2) holds because of inventory feasibility of $\{\curalloctypeopt\}$ in expectation in $\EAR{\typedistributionsequence}$. Now, setting $\mu=1/2$ guarantees dual feasibility as $\forall t: \alpha_t\geq 0,\beta_t\geq 0$, and other constraints of the dual are satisfied by construction. By applying weak duality, the desired ratio is at least $\mu=1/2$, which finishes the proof.

\section{Omitted Details for the Special Case of Infinite Rental Durations}
\label{apx:sand barrier}

\label{sec:stochastic-infinite}
In this section, we consider the setting where the rental times are infinite, i.e., once products are purchased by consumers, they never return to the platform. In this setting, the inventory level of each product is weakly decreasing over time. This enables us to develop inventory-dependent discarding policies in polynomial time using dynamic programming. These discarding policies are more refined than the per-unit revenue thresholding policy discussed in \Cref{sec:stochastic-finite-simple-DP} and achieve a better competitive ratio  (i.e., $\infiniteboundtext$), which is near-optimal when inventories are large.

\infiniteDP*

The proof of \Cref{thm:SB infinite DP} follows a similar structure
as for \Cref{thm:stochastic-finite-competitive-ratio}: 
we first use induction to show 
$\{\DP_{i, t}^I\}$ are lower-bounds the for future revenue-to-go of \Cref{alg:SB infinite DP} (\texttt{Part~(\rom{1}})),
then introduce a linear program to show $\DP_{i, 1}^{\inventory_i}$
is at least $1 - \sfrac{1}{\sqrt{\inventory_i + 3}}$ of the contribution
of product $i$ to the objective value of $\EAR{\typedistributionsequence}$
(\texttt{Part~(\rom{2}})).

\begin{proof}{\emph{Proof of \Cref{thm:SB infinite DP}.}}
The running time is proved by \Cref{prop:EAR running time}, and the fact that the simple DP in \Cref{sec:stochastic-finite-simple-DP} can be solved in polynomial time. The proof of competitive has two parts: 
\paragraph{\texttt{Part~(\rom{1})}.}
	Here we show for any product $i$, time $t$ and inventory $I$,
	the future expected revenue in \Cref{alg:SB infinite DP}
	from product $i$ with current inventory $I$ at time $t$ is
	at least $\DP_{i, t}^I$ by induction on time $t$ from $T + 1$ to 1.
	The base case $t = T + 1$ is satisfied due to our boundary assumption
	that $\DP_{i, T + 1}^I = 0$. 
	Now, suppose from time $t + 1$ to $T + 1$, the induction hypothesis holds.
	Consider the algorithm at time $t$.
	An assortment $\sampleassortment$ is sampled with probability $\curexantetype$.
	Let $\bar\assortment \subseteq \sampleassortment$ be the assortment 
	offered to consumer at time $t$ after discarding.
	For any product $i$ with inventory $I$,
	the algorithm discards it if
	$\DP_{i, t + 1}^{I - 1} + \curreward < \DP_{i, t + 1}^{I}$.
	By induction hypothesis at time $t + 1$,
	it guarantees to gain future expected utility 
	at least $\max\{
		\DP_{i, t + 1}^I, 
		\choice^{\type_t}(\bar\assortment, i)
		(\curreward + \DP_{i, t + 1}^{I - 1})
		+
		(1-\choice^{\type_t}(\bar\assortment, i))
		\DP_{i, t + 1}^I, 
	\}$.
	Since the choice model satisfies substitutability (\Cref{asp:substitutability}),
	$\choice^{\type_t}(\sampleassortment, i) \leq 
	\choice^{\type_t}(\bar\assortment, i)$ for all
	$i \in \bar\assortment$.
	Therefore, by similar calculations as in the proof of \Cref{lem:part-one-simple}, the induction hypothesis holds at time $t$, i.e.,
	$\DP_{i, t}^I$ lower-bounds the future expected revenue from product $i$ with inventory $I$ 
	at time $t$.
	Thus, the expected revenue of \Cref{alg:SB infinite DP} 
	from product $i$ from time $1$ is at least $\DP_{i, 1}^{\inventory_i}$.

\paragraph{ \texttt{Part~(\rom{2})}.}
Here we show $\DP_{i, 1}^{\inventory_i}$ is at least $\left(1 - \tfrac{1}{\sqrt{\inventory_i + 3}}\right)$ fraction of the contribution of product $i$ to the objective value of $\EAR{\typedistributionsequence}$.
Note that scaling $\{\curreward\}$ does not change the ratio. Hence, 
without loss of generality, we normalize
the contribution of product $i$ to the optimal objective value of $\EAR{\typedistributionsequence}$,
i.e.,
\begin{equation*}
\displaystyle\sum\nolimits_{t=1}^{\totaltime}
			\displaystyle\sum\nolimits_{\type_t \in \typespace_t}
			\displaystyle\sum\nolimits_{\assortment \in \assortmentspace}\curexantetype\curchoice
			\curreward
			 =1
\end{equation*}
Given the above observation and the Bellman update in \eqref{eq:bellman equation infinite}, consider the following primal linear program that gives a lower-bound on worst-case value of the ratio:
	\begin{align}
		\label{eq:infinite DP program}
		\begin{array}{llll}
			{\min
			\limits_{\mathbf \reward, \mathbf \DP \geq \mathbf 0}}~~&
				\DP_{i, 1}^{\inventory_i}
				&\text{s.t.}& \\[1em]
							&
				\DP_{i, t}^I \geq
				\displaystyle\sum\nolimits_{\type_t \in \typespace_t}
				\displaystyle\sum\nolimits_{\assortment \in \assortmentspace}
				\curexantetype  
				\max\Big\{\DP_{i, t + 1}^I, & & \\
																   &
				\qquad\qquad
				\curchoice
				(\curreward + 
				\DP_{i, t + 1}^{I - 1})
				+
				(1 - \curchoice)\DP_{i, t + 1}^I
			\Big\}
			& t \in [\totaltime],\ I \in [\inventory_i]&
			\\[1em]
			&
			\displaystyle\sum\nolimits_{t=1}^{\totaltime}
			\displaystyle\sum\nolimits_{\type_t \in \typespace_t}
			\displaystyle\sum\nolimits_{\assortment \in \assortmentspace}
			\curreward
			\curchoice
			\curexantetype
			\geq 1
			& &
		\end{array}
	\end{align}
	where the variables are $\{\curreward\}$ and $\{\DP_{i, t}^I\}$.
	To prove our result, it is sufficient to show that 
	for all inventory feasible $\{\curexantetype\}$ (i.e., 
	$\sum_{t, \type_t, S}\curexantetype\curchoice \leq \inventory_i$),
	the value of program \eqref{eq:infinite DP program} is 
	at least $1 - \tfrac{1}{\sqrt{\inventory_i + 3}}$.
	To see this, we relax the first constraint
	in \eqref{eq:infinite DP program} by 
	switching the outer summation and max operator,
	which provides 
	a relaxed linear program as follows,
	\begin{align*}
		\begin{array}{llll}
			{\min
			\limits_{\mathbf \reward, \mathbf \DP \geq \mathbf 0}}~~&
				\DP_{i, 1}^{\inventory_i}
				&\text{s.t.}& \\[1em]
							&
				\DP_{i, t}^I \geq
				\displaystyle\sum\nolimits_{\type_t \in \typespace_t}
				\displaystyle\sum\nolimits_{\assortment \in \assortmentspace}
				\curexantetype  
				\curchoice
				(\curreward + 
				\DP_{i, t + 1}^{I - 1}) && \\
										& 
				\qquad\qquad+
				\left( 1 - 
					\displaystyle\sum\nolimits_{\type_t \in \typespace_t}
					\displaystyle\sum\nolimits_{\assortment \in \assortmentspace}
					\curexantetype  
					\curchoice
				\right)
				\DP_{i, t + 1}^I
				& t \in [\totaltime],\ I \in [\inventory_i]&
				\langle\alpha_t^I\rangle
				\\[1em]
				&
				\DP_{i, t}^I \geq \DP_{i, t + 1}^I
				& t \in [\totaltime],\ I \in [\inventory_i]& 
				\langle\beta_t^I\rangle
				\\[1em]
				&
				\displaystyle\sum\nolimits_{t=1}^{\totaltime}
				\displaystyle\sum\nolimits_{\type_t \in \typespace_t}
				\displaystyle\sum\nolimits_{\assortment \in \assortmentspace}
				\curreward
				\curchoice
				\curexantetype
				\geq 1
				& & 
				\langle\mu\rangle
			\end{array}
		\end{align*}
		We write its dual program with dual variable $\alpha_t^I$, $\beta_t^I$ and 
		$\gamma$ as follows
	\begin{align*}
		\begin{array}{llll}
			{\max
			\limits_{\boldsymbol{\alpha, \beta},\mu \geq \mathbf 0}}~~&
				\mu
				&\text{s.t.}& \\[1em]
				& \alpha_1^{\inventory_i} + \beta_1^{\inventory_i} \leq 1
				& & \langle\DP_{i, 1}^{\inventory_i}\rangle 
				\\[1em]
				& \alpha_1^{I} + \beta_1^{I} \leq 0
				& I \in [\inventory_i - 1]& \langle\DP_{i, 1}^{I}\rangle 
				\\[1em]
							&
				\alpha_{t + 1}^{I} + \beta_{t + 1}^I \leq
				\displaystyle\sum\nolimits_{\type_t \in \typespace_t}
				\displaystyle\sum\nolimits_{\assortment \in \assortmentspace}
				\curexantetype  
				\curchoice
				\alpha_{t}^{I + 1}
				+
				&& \\
										& 
				\qquad\qquad+
				\left( 1 - 
					\displaystyle\sum\nolimits_{\type_t \in \typespace_t}
					\displaystyle\sum\nolimits_{\assortment \in \assortmentspace}
					\curexantetype  
					\curchoice
				\right)
				\alpha_{t}^I + \beta_{t}^I
				& t \in [\totaltime - 1],\ I \in [\inventory_i]&
				\langle\DP_{i, t+ 1}^I\rangle
				\\[1em]
				&
				\displaystyle\sum\nolimits_{t=1}^{\totaltime}
				\displaystyle\sum\nolimits_{\assortment \in \assortmentspace}
				\curchoice
				\curexantetype 
				\mu & & \\
					   &
				\qquad\qquad
				\leq 
				\displaystyle\sum\nolimits_{t=1}^{\totaltime}
				\displaystyle\sum\nolimits_{\assortment \in \assortmentspace}
				\left(
				\curchoice
				\curexantetype
				\left(
				\displaystyle\sum\nolimits_{I = 1}^{\inventory_i}
		\alpha_t^I \right)\right)
				& \type_t \in \typespace& 
				\langle \curreward\rangle
			\end{array}
		\end{align*}
		By weak duality, any feasible dual solution offers 
		a lower bound for the primal. 
		
Let 
$q_t \triangleq \sum_{\type_t\in \typespace_t}\sum_{\assortment\in\assortmentspace}\curexantetype\curchoice$.
By definition, $\sum_{t=1}^\totaltime q_t \leq \inventory_i$.
Now, suppose we are interested in generating a feasible dual assignment where the all the constraints corresponding
are tight,
i.e., a solution assignment for 
the following
program~\eqref{eq:factor revealing}.
\begin{align}
    \label{eq:factor revealing}
    		\begin{array}{llll}
			{\max
			\limits_{\boldsymbol{\alpha, \beta}, \mu \geq \mathbf 0}}~~&
				\mu
				&\text{s.t.}& \\[1em]
				& \alpha_1^{\inventory_i} + \beta_1^{\inventory_i} = 1
				& &  
				\\[1em]
				& \alpha_1^{I} + \beta_1^{I} = 0
				& I \in [\inventory_i - 1]& 
				\\[1em]
							&
				\alpha_{t + 1}^{I} + \beta_{t + 1}^I =
				q_t\alpha_t^{I + 1} 
				+
				\left(1 - q_t
				\right)
				\alpha_{t}^I + \beta_{t}^I
				& t \in [\totaltime - 1],\ I \in [\inventory_i]&
				\\[1em]
				&
				\mu =
				\displaystyle\sum\nolimits_{I = 1}^{\inventory_i}
		\alpha_t^I 
				& \revcolor{\type} \in \typespace& 
			\end{array}
\end{align}
Notice that \Cref{lem:factor revealing} below
is sufficient to finish the 
proof of \Cref{thm:SB infinite DP}. 

\begin{restatable}{lemma}{sand}
\label{lem:factor revealing}
    For any non-negative sequence $\{q_t\}_{t=1}^\totaltime$
    such that $\sum_{t=1}^\totaltime q_t \leq \inventory_i$, the objective value 
    of program~\eqref{eq:factor revealing} 
    is at least $1-\frac{1}{\sqrt{\inventory_i+3}}$.
\end{restatable}
\ifinfiniteapx
To show \Cref{lem:factor revealing},
we use a technical lemma (\Cref{def:sand/barrier}, \Cref{lem:sand/barrier})
developed in
\citet{ala-14} described below.

	\begin{definition}[Sand/Barrier Process, \citealp{ala-14}]
	\label{def:sand/barrier}
	Consider a tape of infinite length with one unit of infinitely divisible sand at position 
	0 and a barrier at position 1.
	A sequence of $q_1, \dots, q_\totaltime$ (with all $q_t \in [0, 1]$)
	and a parameter $\gamma \in (0, 1)$ are given as the input.
	The sand and the barrier are gradually moved to the right in $\totaltime$ rounds.
	At the $t$-th round the following takes place.
	The left most $\gamma$ fraction of the sand on the tape is selected and the
	$q_t$ fraction of this sand is moved one position to the right.
	This can be defined formally as follows.
	Let $s_t^j$ denote the amount of sand
	at postition $j$ at the beginning of round $t$
	and let $y_t^j\in[0, 1]$ denote the fraction of sand
	selected from position $j$ during round $t$.
	$y_t^j$ are chosen such that $\sum_i s_t^jy_t^j = \gamma$
	and such that for some integer $\theta_t$,
	$y_t^j = 1$ for any $j < \theta_t$,
	and $y_t^j = 0$  for any $j > \theta_t$.
	So during round $t$,
	a $q_ty_t^js_t^j$ amount of sand is moved from each position $j$ 
	to position $j + 1$.
	The barrier is moved one position to the right at the end
	of any round when the total amount of sand at the position
	of the barrier is more than $1- \gamma$.
	We will use $\lambda_t$ to denote the position of the barrier at the beginning of round $t$.
\end{definition}

\begin{lemma}[Sand/Barrier Process, \citealp{ala-14}]
	\label{lem:sand/barrier}
	For any $\gamma$ and
	any
	sequence $q_1, \dots, q_n$,
	throughout the sand/barrier process
	the position of the barrier 
	at the beginning of round $t$,
	i.e., $\lambda_t$,
	satisfies
	\begin{align*}
		\lambda_t \leq \sum_{\tpre = 1}^{t - 1}q_{\tpre}\gamma
		+\frac{1 - \gamma^{\lambda_t}}{1 - \gamma}.
	\end{align*}
\end{lemma}

\begin{proof}{\emph{Proof of \Cref{lem:factor revealing}.}}
We now construct such a feasible  solution of 
program~\eqref{eq:factor revealing}
by using the sand/barrier argument. 
Let $\{s_t^j\}, \{y_t^j\}$, and $\{\lambda_t\}$ be variables defined in sand/barrier process of \Cref{def:sand/barrier}, with parameter $T$, $\{q_t\}_{t\in[T]}$ and $\gamma = 1 - \tfrac{1}{\sqrt{\inventory_i+3}}$. Think of different inventory levels as different positions on the tape of sands, but in reveres order. In other words, position $\inventory_i -I$ corresponds to inventory level $I$ (basically, sand moves from higher to lower inventory levels). We have $T$ number of rounds of sand/barrier process.  At round $t$ of the process, there is a total of $s_t^{\inventory_i-I}$ unit of sand at each position $\inventory_i-I$ for $I\in[\inventory_i]$. The amount of sand $s_t^{\inventory_i-I}$ is divided into selected sand $ s_t^{\inventory_i - I}y_t^{\inventory_i - I}$ that determines $\alpha_t^I$, and non-selected sand $s_t^{\inventory_i - I}(1 - y_t^{\inventory_i - I})$, that determines $\beta_t^I$. In summary, we construct the dual solution as follows,
		\begin{align*}
			\mu = \gamma,
			\;\;
			\alpha_t^{I} = s_t^{\inventory_i - I}y_t^{\inventory_i - I},
			\;\;
			\mbox{and}
			\;\;
			\beta_t^{I} = s_t^{\inventory_i - I}(1 - y_t^{\inventory_i - I}).
		\end{align*}
Note that all constraint in 
program~\eqref{eq:factor revealing}
is satisfied by above dual construction.
Based on the description of sand/barrier process in \Cref{def:sand/barrier}, only $\inventory_i$ positions of tape will have sand (and hence this construction is well-defined) if the barrier never passes position $\inventory_i$, i.e., $\lambda_{T + 1} \leq \inventory_i$.
		Invoking \Cref{lem:sand/barrier} and rearranging the terms
		finish the proof.
\hfill\Halmos		
\end{proof}
\else
We defer the proof of \Cref{lem:factor revealing}
to Appendix~\ref{apx:sand barrier}.
\hfill\Halmos
\fi
\end{proof}

 \section{Numerical Experiments}
 \label{sec:numerical}

\paragraph{Experimental setup.}

In our test problems, we have six products indexed by $\{1,2,3,4,5,6\}$. Each product has an initial inventory $c_i=30$.
We consider consumers who choose from their own ``consideration sets'' ~\citep{HS-1969,AFL-15}. Given the consideration set of a consumer, we assume she chooses a product from this set based on a Multinomial Logit (MNL) choice model. More specifically, we have six types of consumers. Consumer of type $j\in\{1,2,3,4,5,6\}$ considers products $[1:j]$.
Given assortment set $S$, she chooses product $i\in S$ with probability $\phi^j(S,i)=\alpha^j_i/(\alpha^j_0+\sum_{l\in S}\alpha^j_l)$. For a consumer of type $j$, $\alpha^j_i=0$ for $i>j$ (i.e., she does not pick a product outside of her consideration set). Each non-zero $\alpha^j_i$ is drawn independently and uniformly at random from interval $[0.9,1.1]$.
To give consumers the possibility of picking an outside option, we set $\alpha_0^j$ so that $\alpha_0^j/ (\alpha^j_0+\sum_{l\in [1:j]}\alpha^j_l)=0.1$. Note that consumers of type $1$ are the most choosy ones (single-minded), while the consumers of type $6$ are the least choosy ones. 
For a consumer of type $j$,
the rental fee $r_i^j$ of product $i$
is drawn uniformly at random 
from the interval 
$[10\cdot \inflateRatio(j, \kappa), 25\cdot \inflateRatio(j, \kappa)]$
where $\inflateRatio(j, \kappa) 
\triangleq 1 + 2\kappa\cdot\frac{6 - j}{5}$,
and 
$\kappa\geq 0$ is a parameter that varies in our simulation cases
(will be specified later).
After generating rental fees $\mathbf r^j = (r_1^j,\dots, r_6^j)$ for consumer of type $j$,
we reorder them so that $r_1^j\geq r_2^j\geq \ldots \geq r_6^j$.

We consider a discrete-time selling/renting horizon of length $T=300$, and a non-stationary model for the arrival of different types in a similar fashion to \cite{RST-17}. Roughly speaking, in order to have worst-case sequence of type distributions, we intend to have an arrival ordering where choosy consumers arrive later in the selling horizon. In this way, policies need to either carefully preserve enough inventories for these consumers (in case of no rental), or take finite rental durations of products into account and carefully manage in a way that enough items return to the inventory  when choosy consumers arrive (in case of rentals). To capture this style of arrival, we divide the selling horizon into chunks of equal length $\tau=T/6$. Let $\tau^j=(6-j)\tau+1$ for $j=1,2,\ldots,6$. At each time $t$, a consumer of type $j$ arrives with probability proportional to $e^{-0.001\kappa\lvert t-\tau^j \rvert}$, i.e., 
\begin{equation}
\textrm{Probability of type $j$ arriving at time $t$}=\frac{e^{-0.001\kappa\lvert t-\tau^j \rvert}}{\sum_{l=1}^6 e^{-0.001\kappa\lvert t-\tau^l \rvert}}~,
\end{equation}
Note that $\kappa=0$ corresponds to i.i.d.\ arrival of types from the uniform distribution. As $\kappa$ increases, the arrival becomes more heterogeneous (converging towards a deterministic arrival in a descending order of types).

We consider two categories of scenarios: first without rentals (see \Cref{fig:exp1}) and second with rentals (see \Cref{fig:exp2}). For each category, we consider four test scenarios corresponding to $\kappa\in\{0,1,2,3\}$. 
In the first case, we set 
initial inventory $c_i = 30$ for each product $i$.
In the second case, 
we 
set initial inventory $c_i = 20$ for each product $i$
 and consider stochastic rental durations.
The rental duration of a product $i$ for 
a consumer of type $j$ 
is set to be the random variable $T/10+X_j$, where $X_j$ is geometrically distributed with a mean that is drawn uniformly at random from interval 
$[20\cdot \inflateRatio(7-j,\kappa),
30\cdot \inflateRatio(7-j,\kappa)]$ 
at the beginning of the simulation. 

In order to measure the expected revenue of different policies in each scenario, we consider running the target policy over independent sample paths of the input (including type realizations, consumer choices, and rental durations) by doing 50 iterations of Monte-Carlo simulation. We then record the generated revenues of all runs/sample paths, and use ``box and whisker'' plots to demonstrate the median revenues and the resulting confidence intervals for these quantities (\Cref{fig:exp1} and \Cref{fig:exp2}). Moreover, we compute average revenue over all runs/sample paths as an estimate for expected revenue of each policy (\Cref{tab:exp1} and \Cref{tab:exp2}).

\paragraph{Policies.} \revcolor{In our numerical experiments, we compare the expected revenue of eleven different policies/benchmarks that are proposed in this paper or previous work:}
\begin{enumerate}
\item \emph{Inventory Balancing} (\texttt{IB}): 
an adaptation of the inventory balancing algorithm in \citet{GNR-14} by keeping track of inventory levels of reusable products. For our experiments, we consider a variant that uses exponential penalty function $\exponential$.
\item \emph{Myopic greedy policy} (\texttt{GR}): given available products at each time, \texttt{GR} offers the (myopic) optimal assortment at this time. This policy is a special case of \texttt{IB} when $\Psi(x)=\indicator{x>0}$.

\item \emph{Greedy with linear approximate of value-to-go} (\texttt{GR-APXLinValue}): this policy is the approximate DP policy suggested in \cite{RST-17}, which is essentially a greedy myopic policy with respect to a linear approximation of optimal revenue-to-go function. See Section 3 of \citet{RST-17} for more details.
\item \emph{Rollout policy} (\texttt{Rollout}): this policy is obtained by applying rollout\footnote{In this context, rollout of a static policy is computed by first computing the revenue-to-go function of the static policy using bakward induction, and then greedily picking the best assortment at each time assuming future revenues are equal to the computed revenue-to-go values of the static policy.} on the static policy propsoed in Section~4 of \citet{RST-17}.

\item \emph{Decomposition}
(\texttt{Decomp}):
this policy is another approximate DP policy
suggested in \citet{LV-08},
which first decomposes the original DP 
across the products and then constructs value function approximations by solving a separate DP for each product.
This policy  is a widely used heuristics in practice,
but has no theoretical performance guarantees.
\revcolor{
\item \emph{Simulation-based policy with random discarding} (\texttt{Sim+Random}): this policy is the simulation-based policy that randomly discard each product with probability $\gamma$ and is discussed in \Cref{sec:large inventory} (\Cref{alg:large inventory}).
For our experiments, we set $\gamma = 0.1$.

\item \emph{Simulation-based policy with per-unit revenue thresholds} (\texttt{Sim+Infusion}): this policy is the simulation-based policy that assumes infusion/replenishment of inventories to compute its discarding rule and is discussed in \Cref{sec:small inventory} (\Cref{alg:SB}).

\item \emph{Simulation-based policy with inventory-dependent revenue thresholds} (\texttt{Sim+OPTDis}): this policy is designed for the special case of no rentals. It combines an exact randomized rounding of expected LP with the optimum online discarding rules for each product to maintain inventory feasibility of all the products.
We implement this policy for the 
no-rental scenario.
See its formal description in \Cref{sec:stochastic-infinite} (\Cref{alg:SB infinite DP}).
}

\item \emph{Simulation-based policy with hybrid discarding rules} (\texttt{Sim+Hybrid(i)}): 
this policy is the first type of 
simulation-based 
hybrid policy discussed in \Cref{sec:hybrid}.
Specifically, 
for each product, it compares the approximation guarantees 
for different discarding policies,
i.e., random discarding in 
\Cref{alg:large inventory}
and per-unit revenue thresholds in \Cref{alg:SB}
(as well as 
inventory-dependent revenue thresholds in \Cref{alg:SB infinite DP}
for the no rental scenario).
Then it chooses the discarding policy 
with higher approximation guarantee 
for each product separately. 

\item \emph{Simulation-based policy with hybrid
by Monte-Carlo simulation} (\texttt{Sim+Hybrid(ii)}):
this policy is the second type of 
simulation-based 
hybrid policy discussed in \Cref{sec:hybrid}.
Specifically, after every 10 periods, it estimates
the expected future revenue of \Cref{alg:large inventory}
and \Cref{alg:SB}
(as well as \Cref{alg:SB infinite DP}
for the no rental scenario) by doing 20 iterations of Monte-Carlo simulation; and switches to the policy with higher
estimated future revenue.


\item \emph{Bayesian expected LP}: this benchmark is the Bayesian LP benchmark $\EAR{\typedistributionsequence}$ defined in \Cref{sec:prelim}. It provides an upper-bound on the expected revenue of any feasible online policy (when the expectation is taken over the randomness of types, consumer choices, and rental durations, as well as internal randomness of the policy).
\end{enumerate}


\begin{table}[ht]
\centering
\caption{Comparing average revenue of different policies; No rentals and $\forall i:c_i=30$; Policies/benchmarks with $^*$ are proposed or analyzed in this paper. }
\label{tab:exp1}
\begin{tabular}[t]{lcccc}
\toprule
&$\kappa=0$ &$\kappa=1$&$\kappa=2$&$\kappa=3$\\
\midrule
\texttt{Expected-LP}$^*$ & 3540.84 & 7156.37 & 11277.60 & 14898.52 \\
\texttt{IB} & 3217.19 & 6061.12 & 8983.46 & 11223.27 \\
\texttt{GR} & 3010.94 & 5312.55 & 7587.48 & 9097.80 \\
\texttt{GR-APXLinValue} & 2985.39 & 5297.21 & 7203.79 & 9310.97 \\
\texttt{Rollout} & 3127.90 & 5514.05 & 7606.67 & 9406.78 \\
\texttt{Decomp} & 3054.57 & 5785.07 & 8521.41 & 10227.28 \\
\texttt{Sim+Random}$^*$ & 3185.99 & 6382.40 & 10455.97 & 13741.67 \\
\texttt{Sim+Infusion}$^*$ & 3353.12 & 6715.98 & 10415.32 & 13718.31 \\
\texttt{Sim+OPTDis}$^*$ & 3339.37 & 6389.41 & 9905.52 & 13723.17 \\
\texttt{Sim+Hybrid(i)}$^*$ & 3358.77 & 6392.97 & 10537.14 & 13799.11 \\
\texttt{Sim+Hybrid(ii)}$^*$ & 3345.92 & 6774.01 & 10827.79 & 14286.46 \\
\bottomrule
\end{tabular}
\end{table}%

\begin{figure}[ht]
  \centering
       \subfloat[$\kappa = 0$]
      {\includegraphics[width=0.5\textwidth]{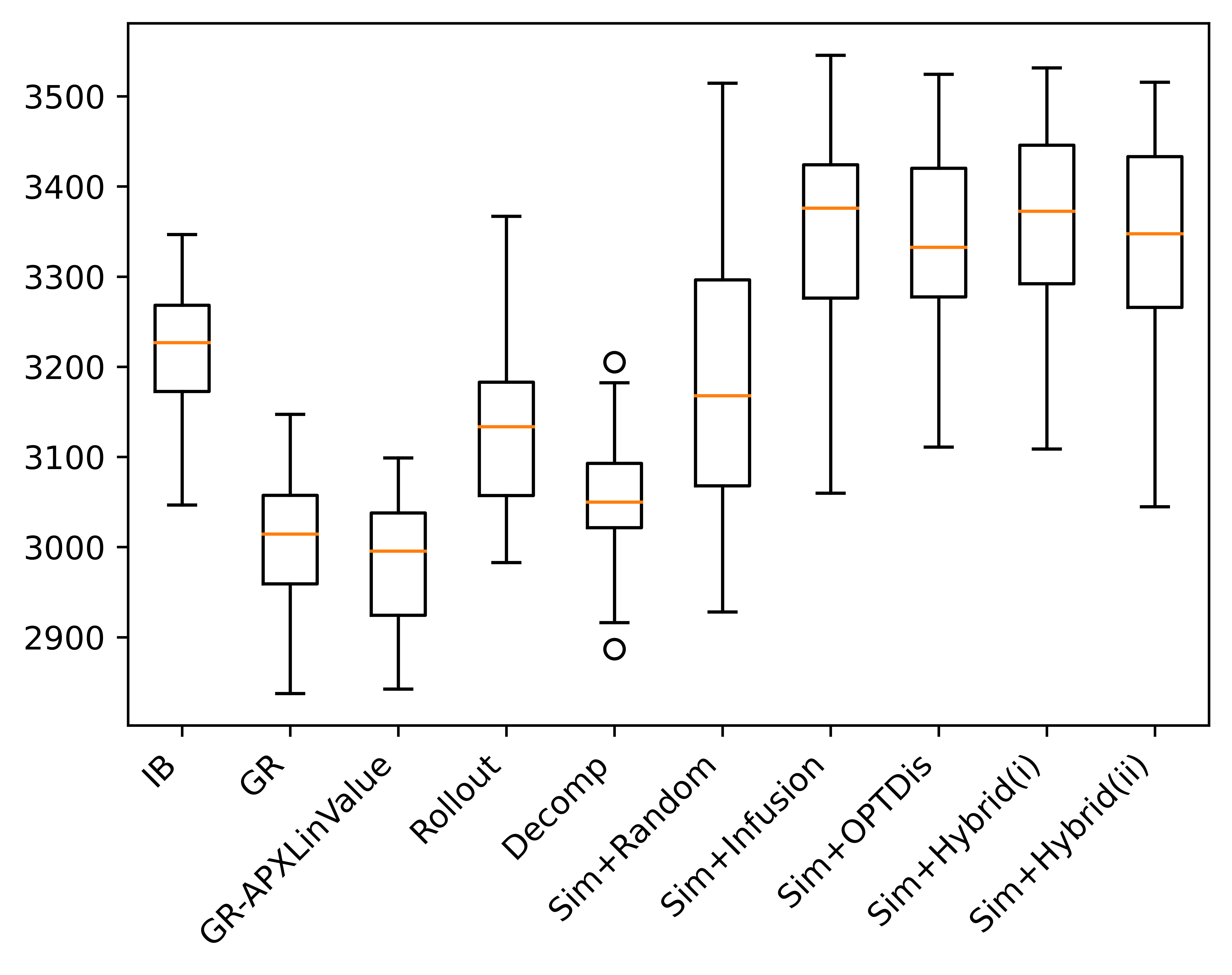}}
      \subfloat[$\kappa = 1$]
      {\includegraphics[width=0.5\textwidth]{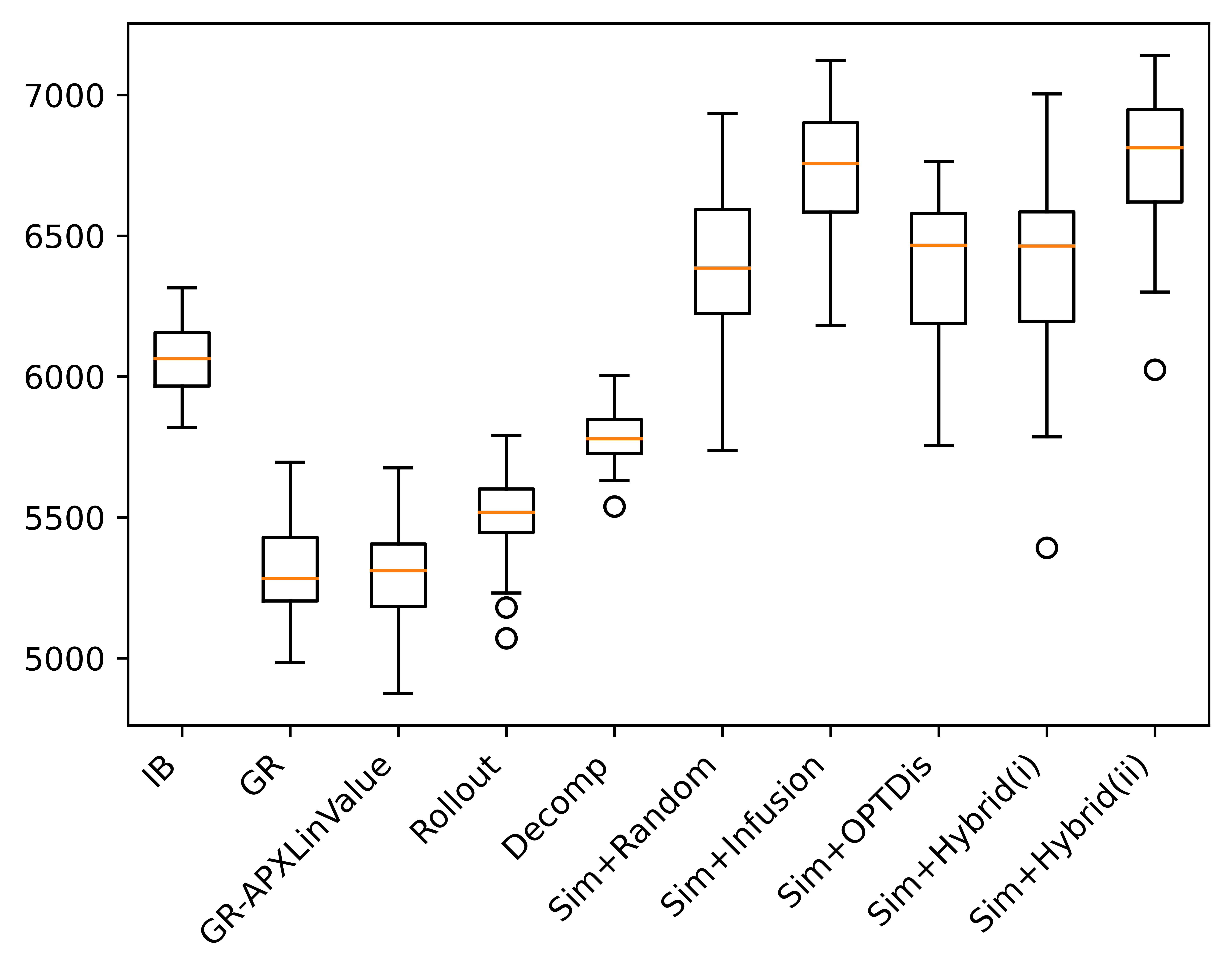}}
      \\
      \subfloat[$\kappa = 2$]
      {\includegraphics[width=0.5\textwidth]{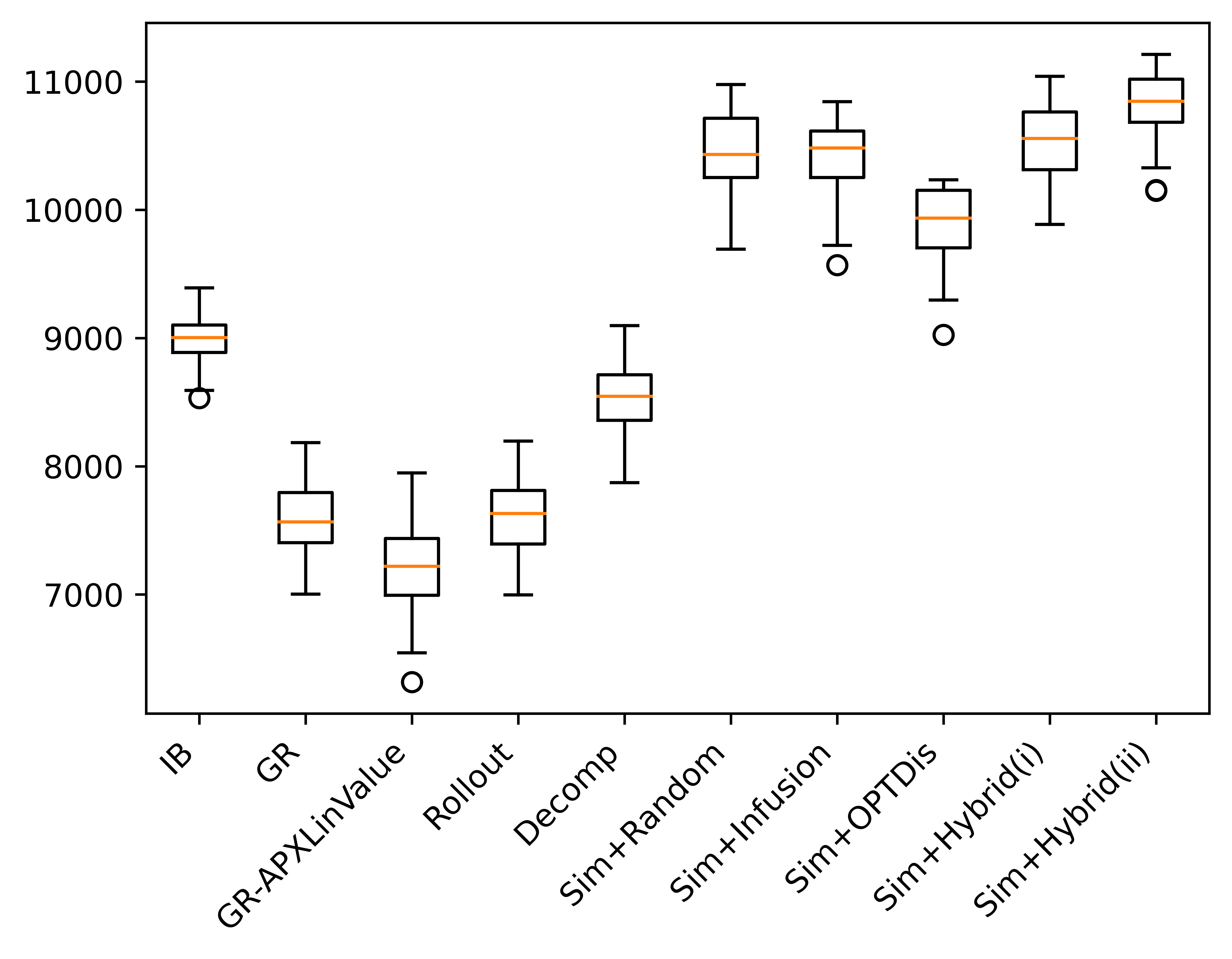}}
      \subfloat[$\kappa = 3$]
      {\includegraphics[width=0.5\textwidth]{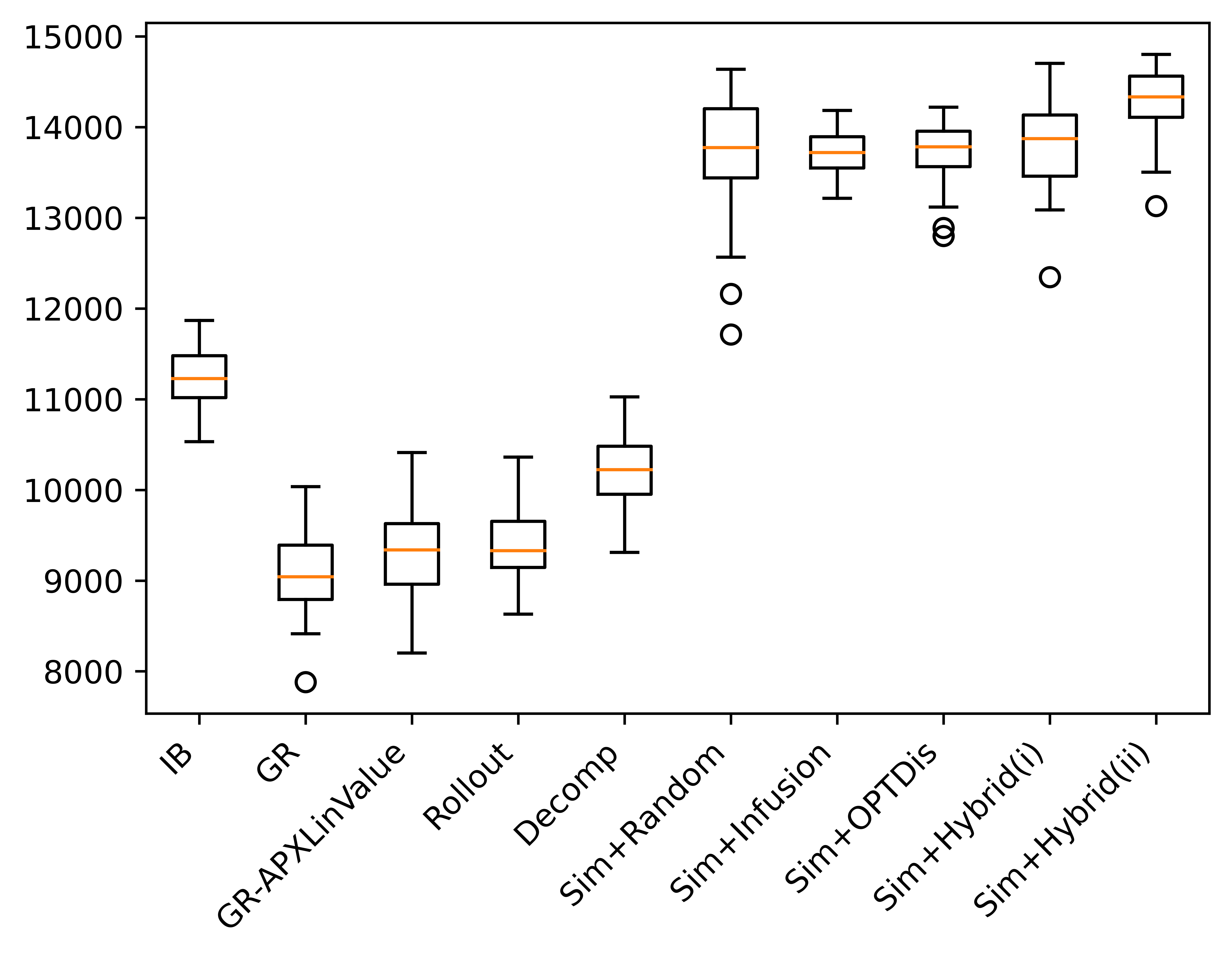}}
  \caption{Box and whisker comparison of different policies in terms of expected revenue when there is no rental.
  Results are based on 50 iterations of Monte-Carlo simulation.  }
   \label{fig:exp1}
\end{figure}

\revcolor{
\begin{table}[ht]
\centering
\caption{Comparing average revenue of different policies; With rentals and $\forall i:c_i=20$; Policies/benchmarks with $^*$ are proposed or analyzed in this paper.}
\label{tab:exp2}
\begin{tabular}[ht]{lcccc}
\toprule
&$\kappa=0$ &$\kappa=1$&$\kappa=2$&$\kappa=3$\\
\midrule
\texttt{Expected-LP}$^*$ & 5224.14 & 10142.05 & 15163.47 & 17552.75 \\
\texttt{IB} & 4600.11 & 7864.58 & 10577.72 & 12096.74 \\
\texttt{GR} & 4321.88 & 7052.60 & 9191.56 & 10387.23 \\
\texttt{GR-APXLinValue} & 3860.43 & 5968.55 & 8148.01 & 9234.32 \\
\texttt{Rollout} & 3970.43 & 6373.85 & 8828.02 & 10012.41 \\
\texttt{Decomp} & 4417.60 & 7148.10 & 9484.07 & 10580.31 \\
\texttt{Sim+Random}$^*$ & 4405.69 & 8110.84 & 12139.91 & 13531.74 \\
\texttt{Sim+Infusion}$^*$ & 4576.72 & 8025.13 & 11531.98 & 13253.71 \\
\texttt{Sim+Hybrid(i)}$^*$ & 4583.40 & 8228.66 & 12109.57 & 13578.05 \\
\texttt{Sim+Hybrid(ii)}$^*$ & 4590.80 & 8540.15 & 12683.29 & 14223.66 \\
\bottomrule
\end{tabular}
\end{table}%
}

\begin{figure}[ht]
  \centering
      \subfloat[$\kappa = 0$]
      {\includegraphics[width=0.5\textwidth]{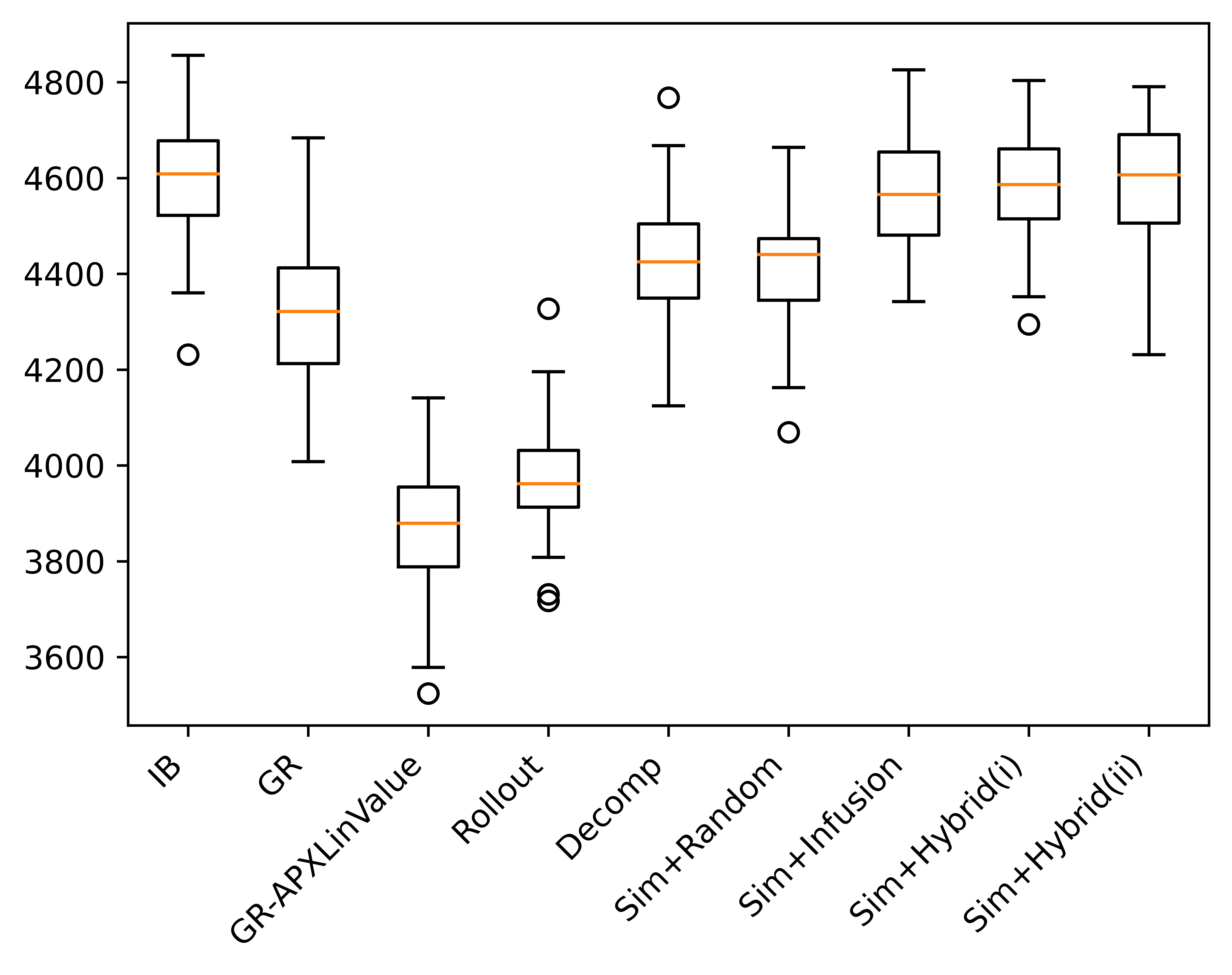}}
      \subfloat[$\kappa = 1$]
      {\includegraphics[width=0.5\textwidth]{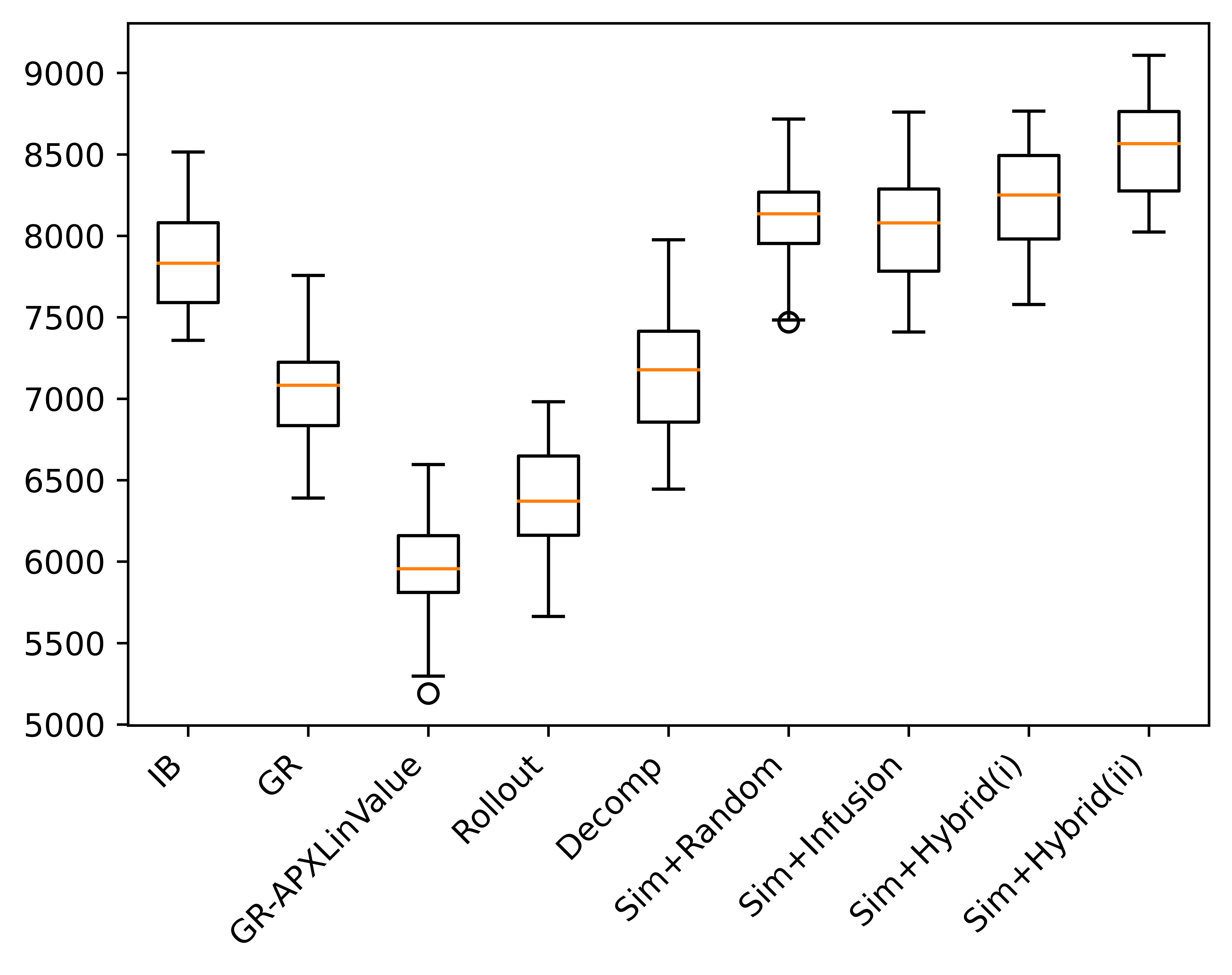}}
      \\
      \subfloat[$\kappa = 2$]
      {\includegraphics[width=0.5\textwidth]{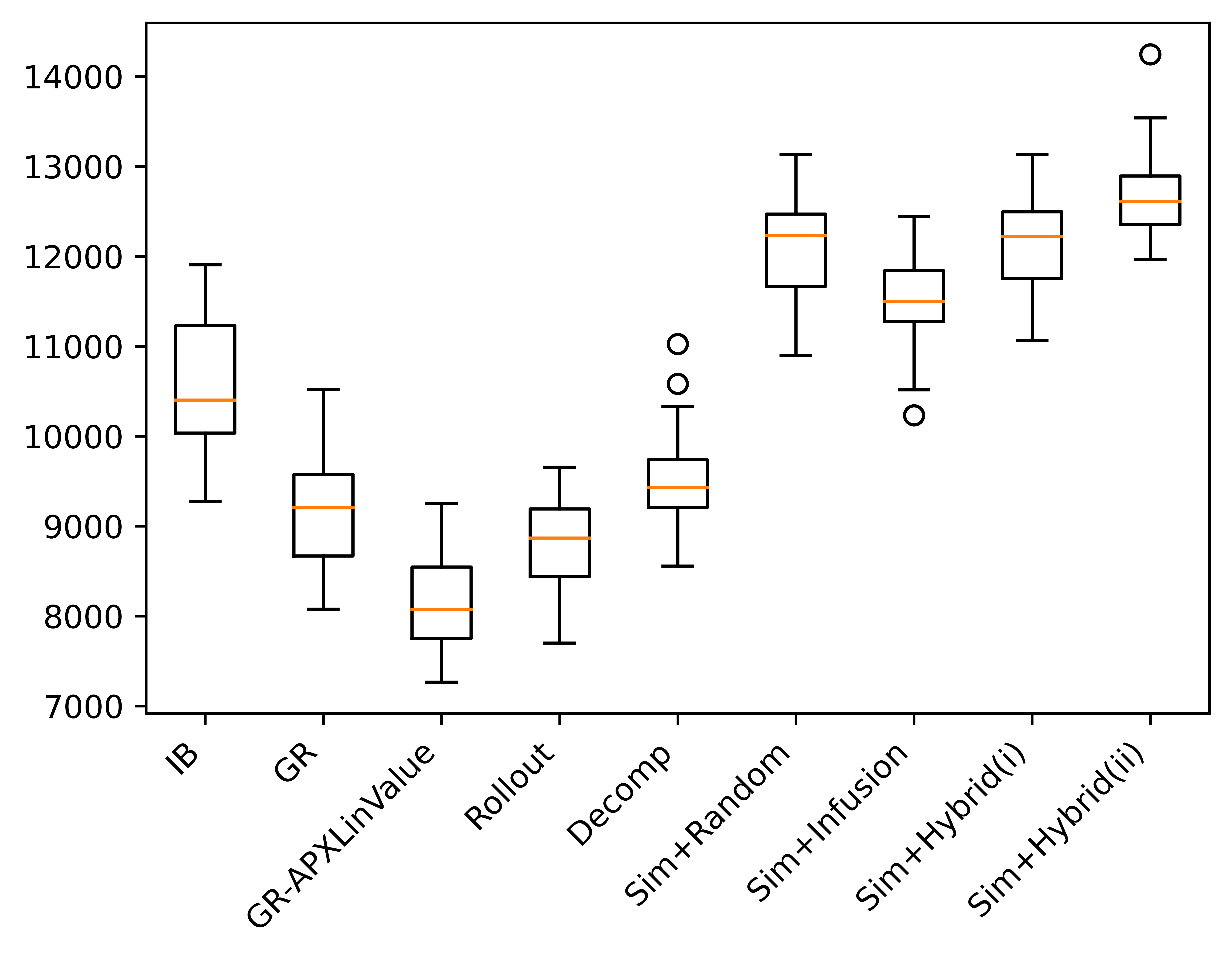}}
      \subfloat[$\kappa = 3$]
      {\includegraphics[width=0.5\textwidth]{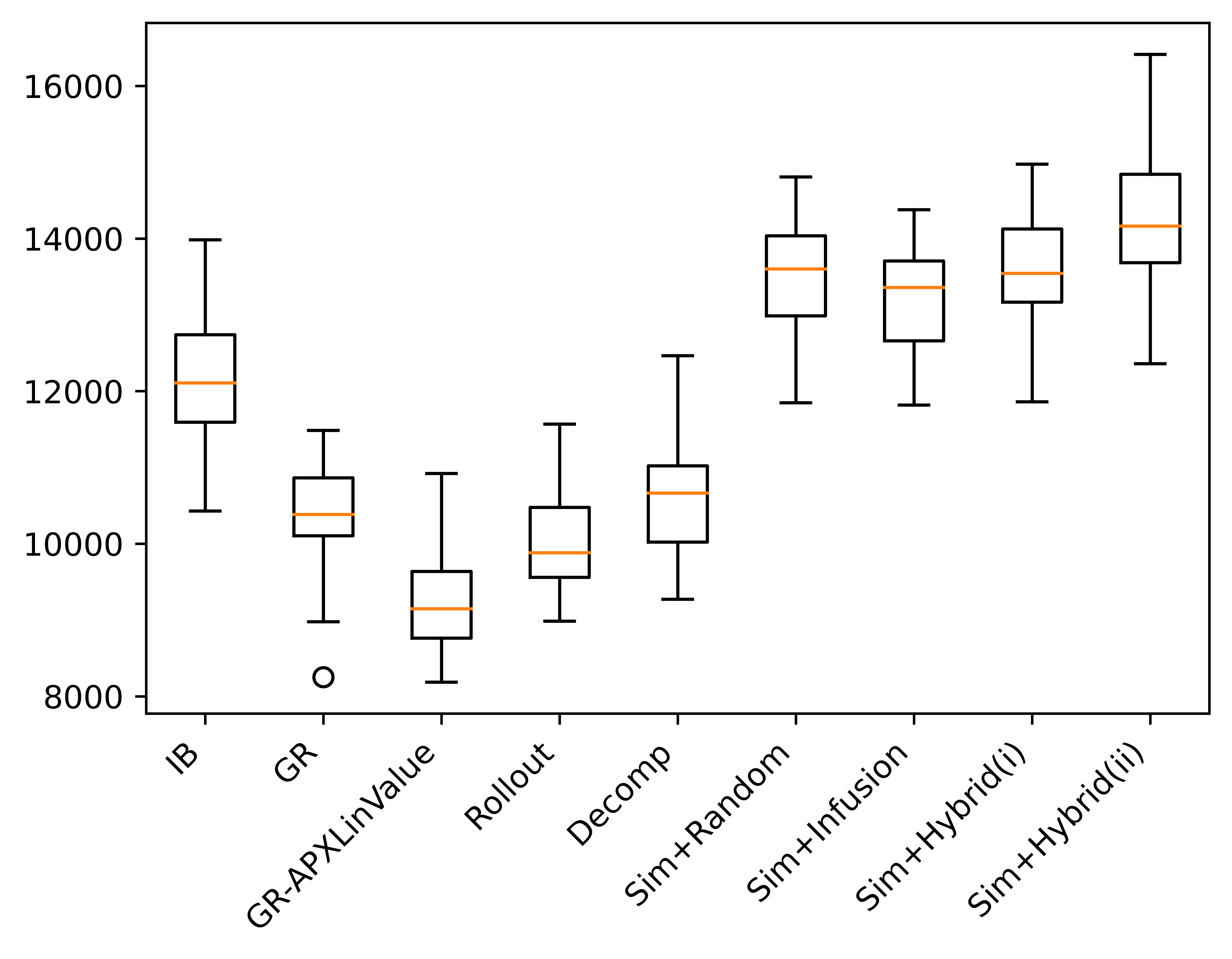}}
  \caption{Box and whisker comparison of different policies in terms of expected revenue when there are type-dependent rentals with stochastic durations.
  Results are based on 50 iterations of Monte-Carlo simulation.  }
  \label{fig:exp2}
\end{figure}

\paragraph{Results.} To summarize our findings from simulations we have run in this section, we consider the two categories mentioned above separately: 
\begin{enumerate}
\item{\emph{No rentals}}: 
\revcolor{in this case, it is clear from both \Cref{tab:exp1} and \Cref{fig:exp1} that 
\texttt{Sim+Hybrid(i)} and
\texttt{Sim+Hybrid(ii)} 
have better performance than others for 
all $\kappa$. Specifically, compare other policies with our hybrid policies by averaging over scenarios $\kappa\in\{0,1,2,3\}$: 
\texttt{Sim+Hybrid(ii)} 
achieves 
15.90\%, 34.59\%, 35.92\%, 31.01\%, 23.35\%, 4.67\%, 4.52\%, 3.08\%, 2.97\%
higher performance than
\texttt{IB},
\texttt{GR},
\texttt{GR+APXLinValue},
\texttt{Rollout},
\texttt{Decomp},
\texttt{Sim+Random},
\texttt{Sim+Infusion},
\texttt{Sim+OPTDis},
\texttt{Sim+Hybrid(i)}
on average,
respectively.\footnote{An important point to mention is that \citet{RST-17} reports a better performance for \texttt{Rollout} compared to \texttt{IB} in a very similar simulation setup. While this looks contradictory at the surface level, we verified the reason behind this discrepancy is that they recompute the approximate values used by \texttt{Rollout} every 100 time periods, which improves its performance. We did not recompute any of the parameters for the sake of having a fair comparison.}
Notably, 
both hybrid policies
\texttt{Sim+Hybrid(i)} and
\texttt{Sim+Hybrid(ii)} 
are better than their components, i.e., \texttt{Sim+Random},
\texttt{Sim+Infusion}
and \texttt{Sim+OPTDis}
in almost all parameter regimes.
In \texttt{Sim+Hybrid(i)},
on average there are 58.9\% products
using the uniform discarding procedure,
and the remaining 41.1\% products 
using inventory-dependent revenue threshold discussed in
\Cref{sec:stochastic-infinite}.
The running times of two hybrid policies 
differ significantly.
Specifically, the running time of 
\texttt{Sim+Hybrid(i)}
is almost the same as all
other policies except 
\texttt{Sim+Hybrid(ii)}
whose running time is $766.39$x slower
due to its
Monte-Carlo simulation.
See \Cref{tab:runtime}.
}

\item \emph{Stochastic and finite rentals}: 
\revcolor{in this case, it is clear from both \Cref{tab:exp2} and \Cref{fig:exp2} that both \texttt{Sim+Hybrid(i)} and 
\texttt{Sim+Hybrid(ii)} 
perform noticeably better than other policies for all $\kappa$.
Specifically, 
\texttt{Sim+Hybrid(ii)} 
achieves 
11.47\%, 25.56\%, 42.92\%, 33.84\%, 22.89\%, 4.77\%, 5.17\%, 3.36\%
higher performance than
\texttt{IB},
\texttt{GR},
\texttt{GR+APXLinValue},
\texttt{Rollout},
\texttt{Decomp},
\texttt{Sim+Random},
\texttt{Sim+Infusion},
\texttt{Sim+Hybrid(i)}
on average,
respectively.
Notably, 
both hybrid policies 
\texttt{Sim+Hybrid(i)} and
\texttt{Sim+Hybrid(ii)} 
are better than their components, i.e., \texttt{Sim+Random}
and \texttt{Sim+Infusion}.
In \texttt{Sim+Hybrid(i)},
on average there are 57.0\% products
using the uniform discarding procedure,
and the remaining 43.0\% products 
using per-unit revenue threshold discussed in
\Cref{sec:small inventory}.
The running times of two hybrid policies 
differ significantly.
Specifically, the running time of 
\texttt{Sim+Hybrid(i)}
is similar to all
other policies except 
\texttt{Sim+Hybrid(ii)}
whose running time is $511.18$x slower
due to its
Monte-Carlo simulation.
See \Cref{tab:runtime}.}
\end{enumerate}

\begin{table}[ht]
\centering
\caption{Comparing normalized average running time of different policies;
Policies/benchmarks with $^*$ are proposed or analyzed in this paper.}
\label{tab:runtime}
\begin{tabular}[t]{lcc}
\toprule
& no rentals & stochastic rentals\\
\midrule
\texttt{IB} & 1.06 & 0.96   \\
\texttt{GR} & 1.04 & 0.92 \\
\texttt{GR-APXLinValue} & 1.10 & 0.94 \\
\texttt{Rollout} & 1.10 & 0.95 \\
\texttt{Decomp} & 4.57 & 1.85 \\
\texttt{Sim+Random}$^*$ & 0.98 & 0.98 \\
\texttt{Sim+Infusion}$^*$ & 1.00 & 0.97 \\
\texttt{Sim+OPTDis}$^*$ & 0.99 &  \\
\texttt{Sim+Hybrid(i)}$^*$ & 1 & 1 \\
\texttt{Sim+Hybrid(ii)}$^*$ & 766.39 & 511.18 \\
\bottomrule
\end{tabular}
\end{table}


 \section{Application of Sub-Assortment Sampling 
in Advanced Booking of Upgrades}
\label{apx:advanced booking}
\revcolor{
To demonstrate the power of our framework and to show why the sub-assortment sampling procedure (\Cref{alg:sample assortment}) can be helpful in other related problems, we  introduce and briefly sketch a variation of our model which we call \emph{``advanced booking of upgrades''} in this section. 
This variation has applications in hotel upsell revenue management.

\paragraph{Advanced booking of upgrades model.} Suppose a hotel booking management system would like to sell upgrade options to the users who have already booked rooms. These upgrade options should be thought of as a combination of reusable resources (e.g., upgrade to a suite, or access to parking) and non-reusable resources (e.g., breakfast at the hotel). 
Each  arriving booking request at time $t$  encodes a check-in time $s_t \geq t$, 
and customer type $z_t$. 
We assume $(s_t, z_t)$ is drawn from known distribution $\{F_t\}^T_{t=1}$.
Customer type $z_t$ specifies the rental fees, 
duration distributions and choice model as our original model. 
After observing the request at time $t$, the platform 
\emph{immediately} display an assortment of possible upgrade options, 
the customer selects one upgrade option based on her choice model given by type $z_t$. 
Then at time $s_t$ when the customer actually checks in, the platform finalizes the decision of whether or not to fulfill the customer’s selected option. 
If the upgrade is allocated, the platform collects the reward.
Note that at the check-in time $s_t$, the platform makes the decision
to finalize the upgrade options based on the 
realized requests between $t$ and $s_t$ as well as their realized selected options.

\paragraph{Bayesian expected LP benchmark.} 
In this model, there exists a Bayesian expected LP benchmark as follows

\begin{align}
\label{eq:exante-stochastic advance booking}\tag{$\text{Expected-LP-AdvBooking}[{\typedistributionsequence}]$}
	\begin{array}{llll}
	{\max
	\limits_{\mathbf \alloc,\mathbf{x} \geq \mathbf 0}}~~&
	\displaystyle\sum\nolimits_{t=1}^\totaltime 
	\displaystyle\sum\nolimits_{s_t=t}^{T}
	\displaystyle\sum\nolimits_{\type_t\in \typespace_t}
	\displaystyle\sum\nolimits_{\assortment \in \assortmentspace}
	\displaystyle\sum\nolimits_{i=1}^n
    \typedistribution_t(s_t\type_t)
	\curreward y_{S, t, s_t, z_t, i}
	&~~~~~\text{s.t.}& \\[1em]
	  &
	\displaystyle\sum\nolimits_{\tpre = 1}^{t}
	\displaystyle\sum\nolimits_{s_t=t}^{T}
	\displaystyle\sum\nolimits_{\type_t\in \typespace_t}
	\displaystyle\sum\nolimits_{\assortment \in \assortmentspace}
	\typedistribution_\tpre(s_\tpre, \type_\tpre)
    \durationcdfi^{\type_{\tpre}}_i(t-s_\tpre)
	y_{\assortment, \tpre, s_\tpre, \type_\tpre,i}
    \leq \inventory_i
	&~~~~~i \in [n],\ t \in [\totaltime]&
	\\[1em]
	&\displaystyle\sum\nolimits_{\assortment\in \assortmentspace}
    x_{S, t, s_t, z_t} \leq 1
	&~~~~~t \in [\totaltime],\ s_t \in [t:\totaltime],\ 
 \type_t\in \typespace & 
	\\[1em]
	&
    y_{S, t, s_t, z_t, i} \leq 
    \curchoice
    x_{S, t, s_t, z_t} 
	&~~~~~
 i\in[n], t \in [\totaltime],\ s_t \in [t:\totaltime],\
 \type_t\in \typespace & 
	\\[1em]
\end{array}
\end{align}
Here, variables 
$\{x_{S, t, s_t, z_t} \}$ correspond to the probabilities 
that assortment $S$ is offered to consumer $t$
given check-in time $s_t$ and type $z_t$ is realized;
and 
variables $\{y_{S, t, s_t, z_t, i}\}$
correspond to the probabilities 
that assortment $S$ is offered to consumer $t$
and upgrade option $i$ is allocated at check-in time $s_t$
given check-in time $s_t$ and type $z_t$ is realized.

Similar to \Cref{prop:relaxation},
the Bayesian expected LP benchmark~\ref{eq:exante-stochastic advance booking}
is an upper bound on the clairvoyant optimum online benchmark,
and hence the
weaker non-clairvoyant optimum online.\footnote{
All the analysis in this section 
for advanced booking of upgrades model is similar 
to the ones in our original model. 
We omit them to avoid repetition.}

\begin{proposition}
For any distributions ${F_t}^T_{t=1}$, 
the expected total revenue of the clairvoyant optimum online benchmark is
upper-bounded by \ref{eq:exante-stochastic advance booking}.
\end{proposition}

\begin{algorithm}
 	\caption{Simulation-based Algorithm 
  for Advanced Booking of Upgrades}
 	\label{alg:advance booking}
 	\KwIn{discarding probability $\gamma\in[0,1]$}
 	\vspace{2mm}
 	\emph{\underline{Pre-processing:}} Compute the optimal assignment 
  $\{x^*_{S, t, s_t, z_t},y^*_{S, t, s_t, z_t, i}\}$ 
  of \ref{eq:exante-stochastic advance booking}
 	\vspace{2mm}
 	
 	\For{$t=1$ to $T$}{
 	\tcc{consumer $t$ with check-in time $s_t$, type $\type_t$ such that $(s_t, z_t)\sim\typedistribution_t$ arrives} 
 	\vspace{1mm}
 	
 	\emph{\underline{Simulation:}} Upon realizing consumer 
  check-in time $s_t$ and type $\type_t$, sample 
  $\hat{\assortment}_t\sim \{x^*_{S, t, s_t, z_t}\}_{S\in \assortmentspace}$
 	
 	\vspace{2mm}
 	\emph{\underline{Discarding:}} Initialize $\bar\assortment_t\gets\hat\assortment_t$
 	
 	\For{each product $i\in\hat\assortment_t$}{
 	\vspace{1mm}
 	Flip an independent coin and remove $i$ from $\bar\assortment_t$ with probability $\gamma$
 	
 	\If {there is no available unit of product $i$}{
 	Remove $i$ from $\bar\assortment_t$}}
 	
 	\vspace{2mm}
 	\emph{\underline{Post-processing:}} Let $\tilde{S}_t\gets \textsc{Sub-assortment Sampling} \left(\choice^{\type_t}, \bar\assortment_t,
  \{y^*_{\hat\assortment_t, t, s_t, z_t, i}\}_{i\in \bar\assortment_t}\right)$
 	
 	\tcc{Send a query call to Procedure~\ref{alg:sample assortment} with appropriate input arguments}
 	\vspace{2mm}
 	
 	Offer assortment $\tilde{\assortment}_t$ to consumer $t$,
  and fulfill consumer $t$'s selected option if it is still available
  at check-in time $s_t$

 	}
\end{algorithm}

\paragraph{Simulation-based Algorithm.}
We
construct a simulation-based algorithm given the optimal solution of above benchmark and using random discarding and
sub-assortment sampling as follows: 
\begin{itemize}
    \item at time $t$, for consumer $t$
    with check-in time $s_t$ and type $z_t$
    \begin{enumerate}
        \item sample assortment $\hat S_t$ based on $\{x_{S, t, s_t, z_t}\}_{S\in \assortmentspace}$;
        \item randomly discard each option with probability $\gamma$,
        and obtain assortment $\bar S_t\subseteq \hat\assortment_t$;
        \item use sub-assortment sampling to match probability $y_{\hat S_t, t, s_t, z_t, i}$, and obtain assortment $\tilde S_t$;
        \item offer assortment $\tilde S_t$ to consumer $t$;
    \end{enumerate}
    \item at check-in time $s_t$, fulfill the customer $t$’s selected option with probability one if it is still available
\end{itemize}
See the formal description of the simulation-based algorithm
in \Cref{alg:advance booking}.

Similar to the analysis in \Cref{thm:competitive-ratio large inventory},
our simulation-based algorithm achieves near-optimal competitive ratio
in the advanced booking of upgrades model when the initial inventory is large. As the proof is almost identical, we omit the details for brevity. 

\begin{theorem}
	\label{thm:competitive-ratio large inventory advance booking}
	By setting $\gamma=\gammaopt$, the competitive ratio of \Cref{alg:advance booking} against the Bayesian expected LP benchmark \ref{eq:exante-stochastic advance booking} is at least $1 -\loss=1-O\left(\sqrt{\log(\mininventory)/\mininventory}\right)$. 
\end{theorem}
}
\end{document}